\documentclass[12pt, a4paper]{article}
\PassOptionsToPackage{compress}{natbib}

\usepackage{times}

\usepackage{nicefrac}       

\usepackage[utf8]{inputenc}
\usepackage[T1]{fontenc}
\usepackage{microtype}
\usepackage{graphicx}
\usepackage{subfigure}
\usepackage{booktabs}
\usepackage[colorlinks=true,citecolor=blue,linkcolor=blue]{hyperref}

\usepackage{url}
\usepackage{nicefrac}
\usepackage{listings}
\usepackage{mathtools}
\usepackage{amsfonts}
\usepackage{amssymb}
\usepackage{amsmath}
\usepackage{amsthm}
\usepackage{docmute}
\usepackage[acronym, nomain]{glossaries}
\usepackage{indentfirst}
\usepackage{bm}
\usepackage{here}
\usepackage{booktabs}       
\usepackage{threeparttable}
\usepackage{cancel}
\usepackage{multirow}
\usepackage{array}
\usepackage{wrapfig}
\usepackage{lipsum}
\usepackage{subcaption}

\usepackage[authoryear,round]{natbib}
\setcitestyle{authoryear,round,semicolon,aysep={,}}
\usepackage{placeins}
\usepackage{algorithm,algorithmic}

\newtheorem{lem}{Lemma}
\newtheorem{thm}{Theorem}
\newtheorem{prop}{Proposition}
\newtheorem{cor}{Corollary}
\newtheorem{assump}{Assumption}
\newtheorem{definition}{Definition}
\newtheorem{rem}{Remark}

\usepackage{cleveref}
\crefname{lem}{lemma}{lemmas}
\crefname{thm}{theorem}{theorems}
\crefname{prop}{proposition}{propositions}
\crefname{cor}{corollary}{corollaries}
\crefname{assump}{assumption}{assumptions}
\crefname{definition}{definition}{definitions}
\crefname{rem}{remark}{remarks}

\Crefname{lem}{Lemma}{Lemmas}
\Crefname{thm}{Theorem}{Theorems}
\Crefname{prop}{Proposition}{Propositions}
\Crefname{cor}{Corollary}{Corollaries}
\Crefname{assump}{Assumption}{Assumptions}
\Crefname{definition}{Definition}{Definitions}
\Crefname{rem}{Remark}{Remarks}

\DeclareMathOperator*{\argmin}{\mathop{\rm argmin}}

\DeclarePairedDelimiterX{\KL}[2]{\mathrm{KL}[}{]}{#1\;\delimsize\|\;#2}

\DeclarePairedDelimiterX\braket[2]{\langle}{\rangle}{#1 \delimsize\vert #2}

\usepackage{thm-restate}
\usepackage{fullpage}
\usepackage{setspace}

\def\KL{\operatorname{KL}}

\title{\vspace{-50pt}\textbf{Generalised Robust Bayes for \\ Joint Inference of Model and Contamination}}

\author{
    \textbf{Masahiro Fujisawa}$^{1,2,3}$\thanks{Corresponding Author: \texttt{masahiro.fujisawa021105@gmail.com}; 1-1 Yamadaoka Suita Osaka 565-0871 Japan}, \ \  
    \textbf{Masaki Adachi}$^{2}$,  \ \ 
    \textbf{Takuo Matsubara}$^{4}$ \\
    \small{$^1$~The University of Osaka, \ $^2$~Lattice Lab., Toyota Motor Corporation} \\
    \small{$^3$~RIKEN AIP, \ $^4$~The University of Sydney Business School} \\
}

\date{}

\begin{document}
\maketitle

\begin{abstract}
\normalsize
Robust generalised Bayesian methods protect parameter inference from contamination but do not generally infer its proportion or identify individual outliers.
We propose H\"older--Bayes, which applies the H\"older divergence to a scaled model to construct a joint generalised posterior over model parameters and the contamination proportion, without specifying an outlier distribution.
Each posterior draw flags a corresponding number of observations least compatible with the fitted model; averaging selections propagates joint uncertainty to outlier-detection scores.
Under fixed, non-negligible contamination, we establish uniform boundedness of the posterior influence function, a finite-sample $n^{-1/2}$ bound on posterior-averaged excess risk, and a Bernstein--von Mises theorem.
Tail overlap between the outlier distribution and the fitted model controls the population risk minimiser's bias and the associated change in limiting posterior covariance.
Interpreting the posterior temperature as an affine volume scaling yields a closed-form, data-driven choice.
Simulations and real-data experiments demonstrate robust estimation, contamination-proportion recovery, and posterior-guided data cleaning.
\end{abstract}

\par

\section{Introduction}
\label{sec:introduction}

A dataset containing outliers gives rise to two complementary statistical questions.
The first concerns robust inference: how to estimate the distribution of the uncontaminated observations without being misled by the contaminated ones. 
The second concerns outlier detection: how to identify which observations are contaminated.
Both questions presuppose the same statistical structure, most commonly the contamination model of~\citet{Huber64} in which observations are drawn from a mixture of a target distribution and an arbitrary outlier distribution. 
An answer to either question would naturally inform the other.
The two problems have nevertheless been addressed by largely separate methodologies, with each inheriting a difficulty that the other's perspective would resolve.

Within the Bayesian paradigm, the sensitivity of the standard posterior to contamination is well documented~\citep{Insua2000}.
Generalised Bayesian inference (GBI)~\citep{Bissiri16} provides a framework for replacing the negative log-likelihood with a loss, yielding a generalised or Gibbs posterior~\citep{Zhang2006}.
Robust choices include losses based on the density-power and $\alpha$-divergences, the $\gamma$-divergence, maximum mean discrepancy, and Stein discrepancies~\citep{Ghosh16,Nakagawa20,cherief20a,Matsubara22,Matsubara24,laplante2025robust}.
A common robustness criterion is uniform boundedness of the posterior influence function (PIF)~\citep{Hooker14,Ghosh16,Matsubara22}, which limits first-order sensitivity to point-mass contamination.
This criterion alone neither estimates a non-negligible contamination proportion nor indicates which observations are outlying.
A further practical difficulty is the temperature parameter that scales the loss in the generalised posterior: existing calibration procedures typically demand iterative optimisation, repeated posterior sampling, or curvature estimation \citep{Grunwald17,Holmes17,Lyddon19,Wu2023}.

Outlier detection methods, for their part, rarely emerge from the inferential model itself.
One route augments the model with an explicit representation of the outliers, for instance through mixture or Dirichlet-process constructions~\citep{PageDunson11,ShotwellSlate11}.
In practice, however, practitioners can typically characterise the clean-data distribution far more readily than that of the outliers.
Specifying a probabilistic model for observations that are, by definition, unanticipated is therefore intrinsically difficult.
The alternative route applies algorithmic detectors such as one-class support vector machines, local outlier factors, or isolation forests~\citep{Scholkopf01,Breunig00,Liu08}, which operate separately from the working model, rank observations by an outlier score, and leave the contamination proportion (equivalently the score threshold) to be supplied by the user.
Closest to the present work, a frequentist line of research has shown that the model parameter and the contamination level can be estimated jointly by minimising the H\"older divergence~\citep{Fujisawa08,Kanamori13,Kanamori15,fujisawa2025scalable}.
Being point-estimation procedures, however, these methods provide no mechanism for propagating uncertainty in the fitted model and contamination level through to the detection of individual outliers.

What is missing is a Bayesian framework that infers the model parameter and the contamination proportion jointly, quantifies uncertainty in both, and converts that joint uncertainty into observation-level outlier identification, all without an explicit outlier model, a user-specified contamination level, or an external anomaly threshold.
The present paper develops such a framework.

In this context, we introduce \emph{H\"older--Bayes}, which applies the H\"older divergence to the scaled model $(1-\epsilon)\mathbb{P}_\theta$ to construct a joint generalised posterior over the model parameter $\theta$ and contamination proportion $\epsilon$, without specifying an outlier distribution.
For each posterior draw, the inferred contamination proportion determines how many observations are flagged, and the fitted model determines their ranking.
Averaging these selections yields the \emph{frequency of detection (FoD)}, which propagates uncertainty in both quantities to observation-level scores (\Cref{sec:holder_pos}).

Our theoretical results cover fixed, non-negligible contamination.
We establish uniform boundedness of the posterior influence function, a finite-sample $n^{-1/2}$ bound on posterior-averaged excess risk, and a Bernstein--von Mises theorem (\Cref{thm:IF_holder,thm:consistency_holder,thm:bvm_holder}).
Tail overlap between the outlier distribution and the fitted model controls the displacement of the population risk minimiser from the true parameter pair (\Cref{thm:error_control_holder}).
It also controls the $2$-Wasserstein distance between the centred Gaussian limit and a centred reference Gaussian (\Cref{thm:bvm_holder_corollary}).
The reference covariance is determined by the curvature of the same population risk at the true parameter pair.
These bounds can remain small without requiring the contamination proportion to vanish.

We also derive an affine-scaling interpretation of the posterior temperature, giving a closed-form, data-driven default (\Cref{sec:practical_use}).
Synthetic and real-data experiments assess parameter estimation, contamination-proportion recovery, and posterior-guided auditing and cleaning (\Cref{sec:experiment}).

\Cref{sec:background,sec:holder_pos,sec:theory,sec:practical_use} develop the framework, theory, and practical settings; \Cref{sec:experiment} reports the experiments.
Proofs and further experimental details are provided in the supplementary material.

\section{Preliminaries}
\label{sec:background}

\paragraph{Notation.}
Let $\mathcal{X}\subseteq\mathbb{R}^k$ be a Borel observation space and let $\Theta\subseteq\mathbb{R}^d$ be a Borel parameter space.
We write $Y_1,\ldots,Y_n$ for the observations and $\mathbb{P}_n\coloneqq\frac{1}{n}\sum_{i=1}^{n}\delta_{Y_i}$ for the empirical distribution.
Probability distributions are denoted by uppercase blackboard-bold letters, such as $\mathbb{P}$ and $\mathbb{Q}$.
When a distribution admits a density with respect to a common dominating measure $\mu$ (e.g.~the Lebesgue measure), the corresponding lowercase letter denotes the density; for example, $p=\mathrm{d}\mathbb{P}/\mathrm{d}\mu$ and $q=\mathrm{d}\mathbb{Q}/\mathrm{d}\mu$.
A statistical model is denoted by $\mathcal{P}=\{\mathbb{P}_\theta:\theta\in\Theta\}$ with corresponding density $p_\theta$.

\subsection{Generalised Bayesian Inference}
\label{sec:gbi}

Given observations $\mathbf{y}_n=(y_1,\ldots,y_n)$ and a prior density
$\pi(\theta)$, the standard posterior density is given by
\begin{align}
    \pi_n(\theta)
    &\propto
    \pi(\theta)
    \prod_{i=1}^{n}
    p_\theta(y_i)
    =
    \pi(\theta)
    \exp
    \left\{
        \sum_{i=1}^{n}
        \log p_\theta(y_i)
    \right\}.
    \label{eq:bayes_pos}
\end{align}
Although standard Bayesian inference is well established under correctly specified models~\citep{Williams80,Zellner88}, the posterior can be highly sensitive to data contamination and the resulting model misspecification~\citep{Insua2000}.
GBI~\citep{Zhang2006, Bissiri16} replaces the negative log-likelihood by a general loss. 
Given a data-dependent loss $L_n(\theta)$, the generalised posterior is defined by
\begin{align}
    \pi_n^L(\theta)
    \propto
    \pi(\theta)
    \exp
    \left\{
        -\beta n L_n(\theta)
    \right\},
    \label{eq:bayes_pos_general}
\end{align}
where $\beta>0$ is a temperature parameter controlling the scale of the
posterior.
The standard posterior is recovered when $L_n(\theta)$ is the average negative log-likelihood and $\beta = 1$.

The appeal of GBI is that the model $\mathbb{P}_\theta$ can be retained while the likelihood contribution is replaced by a loss that is less sensitive to outliers.
This is distinct from classical Bayesian robustness strategies that modify the
prior, temper the likelihood, or explicitly perturb the data-generating model
\citep{Berger86,Berger90,Dey94,Grunwald11a,Grunwald12,Holmes17,Miller19}.
In GBI, robustness is introduced through the loss used to update beliefs.

A standard approach formulates the loss as a statistical divergence, $L_n(\theta) = \operatorname{D}(\mathbb{P}_n \mid \mathbb{P}_\theta)$, between the empirical distribution and the model \citep{Jewson18}.
Common divergences include the density-power divergence (DPD) of order $\gamma>0$, defined, up to conventions on the ordering of its arguments, as 
\begin{align}
    \operatorname{D}^{\gamma}_{\mathrm{DPD}}
    (\mathbb{P}\mid\mathbb{Q})
    &=
    \int_{\mathcal{X}}
    q(x)^{1+\gamma}
    \,\mathrm{d}\mu(x)
    -
    \frac{1+\gamma}{\gamma}
    \int_{\mathcal{X}}
    q(x)^\gamma
    \,\mathrm{d}\mathbb{P}(x)
    +
    \frac{1}{\gamma}
    \int_{\mathcal{X}}
    p(x)^{1+\gamma}
    \,\mathrm{d}\mu(x).
    \label{eq:dpd}
\end{align}
DPD-based posteriors have been actively studied as robust alternatives to the standard posterior~\citep{Ghosh16}.
Another important family comprises pseudo-spherical and $\gamma$-divergences \citep{Nakagawa20}, which involve ratios of density-power integrals. 
Such ratio-based losses are closely related to the H\"older divergence employed in this paper.
Beyond these instances, robust GBI methods have also been developed using the $\alpha$-divergence~\citep{Hooker14}, the maximum mean discrepancy~\citep{cherief20a}, Stein discrepancies~\citep{Matsubara22,Matsubara24}, and score-matching-based discrepancies~\citep{laplante2025robust}.

The robustness of a generalised posterior is often assessed through sensitivity analysis.
The classical influence function measures the infinitesimal effect of a point-mass contamination on an estimator~\citep{hampel11}.
Extending this concept to posterior densities yields the posterior influence function (PIF)~\citep{Hooker14,Ghosh16,Matsubara22}.
A generalised posterior is called \emph{globally bias-robust} when its PIF is uniformly bounded.
This criterion bounds the first-order change in the posterior density caused by point-mass contamination.

PIF-based robustness is, however, inherently infinitesimal: because it evaluates only the impact of a vanishingly small contamination mass, it cannot, by itself, quantify a non-negligible contamination proportion.
Consequently, existing robust GBI methods are designed primarily to protect parameter inference from contamination, rather than to infer the contamination level or identify specific outliers.
The goal of H\"older--Bayes is to extend this perspective from posterior protection to explicit contamination quantification.
To formalise this objective, we next introduce the heavy-contamination setting.

\subsection{Heavy Contamination}
\label{sec:heavy_contamination}

The heavy-contamination setting assumes that the data-generating distribution
contains a non-negligible proportion of outliers.
Let $\mathbb{P}_0$ denote the target distribution of primary interest, and let
$\mathbb{Q}_0$ denote an arbitrary outlier distribution.
We assume that the statistical model $\mathcal{P}=\{\mathbb{P}_\theta:\theta\in\Theta\}$ is well-specified for the target distribution $\mathbb{P}_0$; that is, there exists a true parameter $\theta_0\in\Theta$ such that $\mathbb{P}_{\theta_0}=\mathbb{P}_0$. 
Given the true contamination proportion $\epsilon_0\in[0,0.5)$, the data-generating
distribution $\mathbb{P}$ follows the standard Huber contamination model:
\begin{align}
    \mathbb{P}
    =
    (1-\epsilon_0)\mathbb{P}_{\theta_0}
    +
    \epsilon_0\mathbb{Q}_0.
    \label{eq:huber_contamination}
\end{align}

The term ``heavy contamination'' indicates that $\epsilon_0$ remains fixed and non-negligible in the population limit.
This differs from infinitesimal robustness analysis, such as classical influence-function approaches~\citep{hampel11}, which examines derivatives at zero contamination.
In the heavy-contamination regime, the inferential target remains the parameter $\theta_0$ of $\mathbb{P}_0$, while the observations are generated from the contaminated mixture in Eq.~\eqref{eq:huber_contamination}. 

A central question is therefore to what extent accurate inference on $\theta_0$ remains viable even when the contamination proportion does not vanish.
The crucial factor is not merely the amount of contamination, but where the outlier distribution places its mass relative to the model density.
Following the heavy-contamination literature~\citep{Fujisawa08,Kanamori13,Kanamori15}, we quantify this relationship through the following density-power quantity.
\begin{definition}[Tail-overlap quantity]
\label{def:tail_overlap}
For $\gamma>0$, the \emph{tail-overlap quantity} between the outlier distribution $\mathbb{Q}_0$ and the model density $p_\theta$ is defined as 
\begin{align}
\rho_\gamma(\theta)
\coloneqq
\int_{\mathcal{X}}
p_\theta^\gamma(x)
\mathrm{d}\mathbb{Q}_0(x).
\label{eq:contamination_quantity}
\end{align}
\end{definition}
We refer to $\rho_\gamma(\theta)$ as the \emph{tail-overlap quantity}.
A small value of $\rho_\gamma(\theta)$ indicates that the outlier distribution places most of its mass in regions where the model density $p_\theta$ is small.
Thus, when the outliers lie predominantly in the low-density tails of the model, their contribution through the density-power weighting $p_\theta^\gamma$ is correspondingly small.

The tail overlap plays a specific role in our theory.
We do not impose a fixed upper bound on $\rho_\gamma(\theta)$ as a blanket assumption for the H\"older posterior.
Instead, its effect is retained explicitly in our guarantees: \Cref{thm:error_control_holder} controls the displacement of the population H\"older-risk minimiser from the true parameter pair, while \Cref{thm:bvm_holder_corollary} controls the $2$-Wasserstein distance between the centred Gaussian limit and a centred reference Gaussian whose covariance is determined by the curvature of the same population risk at the true parameter pair.
Under the stated regularity conditions, small tail overlap therefore controls both the population parameter bias and this local covariance discrepancy, even when the contamination proportion remains fixed and non-negligible.
The tail-overlap condition is not required for the posterior influence, excess-risk, or Bernstein--von Mises results themselves.

\section{H\"older Posterior}\label{sec:holder_pos}

We now introduce H\"older--Bayes, a generalised Bayesian framework for joint robust inference of the model parameter $\theta$ and the contamination proportion $\epsilon$.
The key idea is to evaluate the H\"older divergence against the scaled model density $(1-\epsilon)p_\theta$ rather than the original model density $p_\theta$. 
Under the heavy-contamination setting in \Cref{sec:heavy_contamination}, the target component has mixture weight $1-\epsilon_0$.
The factor $1-\epsilon$ scales the total mass of the model $p_\theta$ to match this fraction, rendering $\epsilon$ a parameter to be inferred.
We infer $\epsilon$ jointly with $\theta$ through the resulting generalised posterior.

\subsection{H\"older Divergence with a Scaled Model}
\label{sec:section31}

The H\"older divergence, introduced by \citet{Kanamori13}, is induced by a family of H\"older scores indexed by a power constant $\gamma>0$ and a function $\phi:(0,\infty)\to\mathbb{R}$ satisfying
$\phi(1)=-1$ and $\phi(z)\ge -z^{1+\gamma}$.
For probability distributions $\mathbb{P}$ and $\mathbb{Q}$ with densities $p$ and $q$, respectively, the divergence takes the following form:
\begin{align}
    \mathrm{D}_\phi^\gamma
    (
        \mathbb{P}
        \mid
        \mathbb{Q}
    )
    \coloneqq
    \phi
    \left(
        \frac{
            \int_{\mathcal{X}} q(x)^\gamma\,\mathrm{d}\mathbb{P}(x)
        }{
            \int_{\mathcal{X}} q(x)^{1+\gamma}\,\mathrm{d}\mu(x)
        }
    \right)
    \int_{\mathcal{X}}
        q(x)^{1+\gamma}
    \,\mathrm{d}\mu(x)
    +
    \int_{\mathcal{X}}
        p(x)^{1+\gamma}
    \,\mathrm{d}\mu(x).
    \label{eq:hd}
\end{align}

When both arguments are probability densities, the associated H\"older score is proper under the conditions above; equivalently, $\mathrm{D}_\phi^\gamma(\mathbb{P}\mid\mathbb{Q})$ is non-negative and vanishes if and only if $\mathbb{P}=\mathbb{Q}$~\citep{Kanamori13,gneiting2007strictly}.
Several familiar robust scores arise from particular choices of $\phi$.
For example, the DPD reviewed in \Cref{sec:background} is recovered, up to a positive multiplicative constant, by choosing $\phi(z)=\gamma-(1+\gamma)z$.
The boundary choice $\phi(z)=-z^{1+\gamma}$ yields the pseudo-spherical score, which is closely related to the $\gamma$-divergence of~\citet{Fujisawa08}.
As discussed in \Cref{sec:scaled_mass_risk_interpretation,sec:choice_phi}, this scale-invariant choice is unsuitable for our goal of jointly identifying the contamination proportion.

We now evaluate the H\"older divergence at the scaled model density $(1-\epsilon)p_\theta$.
Since the last term of Eq.~\eqref{eq:hd} does not depend on $(\theta,\epsilon)$, the relevant population \emph{H\"older risk} is
\begin{align}
    H^{(\gamma)}(\theta,\epsilon)
    \coloneqq
    \phi
    \left(
        (1 - \epsilon)^{-1}
        \frac{
            \mathbb{E}_{Y\sim\mathbb{P}}
            [
                p_\theta^\gamma(Y)
            ]
        }{
            \mathbb{E}_{Y\sim\mathbb{P}_\theta}
            [
                p_\theta^\gamma(Y)
            ]
        }
    \right)
    (1 - \epsilon)^{1+\gamma}
    \mathbb{E}_{Y\sim\mathbb{P}_\theta}
    [
        p_\theta^\gamma(Y)
    ].
    \label{eq:hd_main}
\end{align}
The form of Eq.~\eqref{eq:hd_main} reveals the basic mechanism behind joint estimation.
Under the contamination model $\mathbb{P}=(1-\epsilon_0)\mathbb{P}_0+\epsilon_0\mathbb{Q}_0$, the density-power expectation decomposes as
\begin{align}
\mathbb{E}_{Y\sim\mathbb{P}}
\left[
p_\theta^\gamma(Y)
\right]
&=
(1-\epsilon_0)
\mathbb{E}_{Y\sim\mathbb{P}_0}
\left[
p_\theta^\gamma(Y)
\right]
+
\epsilon_0\rho_\gamma(\theta).
\label{eq:holder_contamination_decomposition}
\end{align}
The second term is the contribution of the contaminating distribution through the tail-overlap quantity introduced in \Cref{def:tail_overlap}.
When the outliers lie predominantly in low-density regions of the model, $\rho_\gamma(\theta)$ is small and this contribution is correspondingly suppressed.
This separation between the target and contaminating contributions is what allows the scale factor $1-\epsilon$ to carry information about the proportion of the true target data. The precise mechanism by which $\epsilon$ is identified is developed in \Cref{sec:scaled_mass_risk_interpretation}.

\subsection{Definition of the H\"older Posterior}
\label{sec:section32}

For $\theta\in\Theta$, write $S_{n,\gamma}(\theta)=n^{-1}\sum_{i=1}^n p_\theta^\gamma(y_i)$ and
$C_\gamma(\theta)=\int_{\mathcal X}p_\theta^{1+\gamma}\,\mathrm{d}\mu
=\mathbb{E}_{Y\sim\mathbb P_\theta}[p_\theta^\gamma(Y)]$.
The empirical H\"older risk is
\begin{align}
    H_n^{(\gamma)}(\theta,\epsilon)
    =\phi\!\left(\frac{S_{n,\gamma}(\theta)}
                      {(1-\epsilon)C_\gamma(\theta)}\right)
      (1-\epsilon)^{1+\gamma}C_\gamma(\theta).
    \label{eq:holder_entropy}
\end{align}

We now define the H\"older posterior, a generalised posterior built upon the empirical risk in Eq.~\eqref{eq:holder_entropy}.
\begin{definition}[H\"older posterior]
\label{def:holder_posterior}
Given a prior density $\pi$ on $\Theta\times[0,1)$ and temperature $\beta>0$, the H\"older posterior is
\begin{align}
    \pi_n^H(\theta,\epsilon)
    \propto \exp\{-\beta nH_n^{(\gamma)}(\theta,\epsilon)\}
             \pi(\theta,\epsilon).
    \label{eq:holder_posterior_def}
\end{align}
\end{definition}

Evaluating this posterior requires only $S_{n,\gamma}(\theta)$ and $C_\gamma(\theta)$, without specifying the outlier distribution.
The latter integral, which also appears in R\'enyi entropy, has closed forms for many parametric models~\citep{Zografos2005}; otherwise, numerical integration or Monte Carlo simulation under $\mathbb P_\theta$ can be used~\citep{adachi2022fast}.
Our experiments use closed forms, with the temperature selected as described in \Cref{sec:practical_use}.

\subsection{Mechanism of Contamination-Proportion Estimation} \label{sec:scaled_mass_risk_interpretation}

We now make explicit how the scaled-model formulation enables estimation of the contamination proportion $\epsilon$.
Consider first the DPD specification
$\phi(z)=\gamma-(1+\gamma)z$.
Substituting this choice into Eq.~\eqref{eq:holder_entropy} gives $H_n^{(\gamma)}(\theta,\epsilon)
=\gamma(1-\epsilon)^{1+\gamma}C_\gamma(\theta)-(1+\gamma)(1-\epsilon)^\gamma S_{n,\gamma}(\theta)$.
For fixed $\theta$, differentiation with respect to $\epsilon$ gives
\begin{align}
\frac{\partial}{\partial\epsilon}
H_n^{(\gamma)}(\theta,\epsilon)
=
\gamma(1+\gamma)(1-\epsilon)^{\gamma-1}
\left\{
S_{n,\gamma}(\theta)
-
(1-\epsilon)C_\gamma(\theta)
\right\}.
\end{align}
Hence, whenever the stationary point lies in the parameter space, the risk-minimising contamination proportion satisfies $\epsilon_n^*(\theta)=1-S_{n,\gamma}(\theta)/C_\gamma(\theta)$.

This expression directly reveals the contamination-proportion estimation mechanism.
Recall from Eq.~\eqref{eq:holder_contamination_decomposition} that, at the true model parameter $\theta_0$, $S_{n,\gamma}(\theta_0)\approx\mathbb{E}_{Y\sim\mathbb{P}}[p_{\theta_0}^\gamma(Y)]=(1-\epsilon_0)C_\gamma(\theta_0)+\epsilon_0\rho_\gamma(\theta_0)$.
When the tail overlap is small, the contamination contribution
$\epsilon_0\rho_\gamma(\theta_0)$ is negligible relative to the target contribution
$(1-\epsilon_0)C_\gamma(\theta_0)$.
Hence, for an interior solution,
\begin{align}
\epsilon_n^*(\theta_0)
=
1-
\frac{S_{n,\gamma}(\theta_0)}
{C_\gamma(\theta_0)}
\approx
1-\frac{(1-\epsilon_0)C_\gamma(\theta_0)}
{C_\gamma(\theta_0)}=\epsilon_0.
\label{eq:epsilon_recovery_dpd}
\end{align}
Thus, when the outliers lie predominantly in low-density regions of the target model, the estimate is close to the true contamination proportion even when $\epsilon_0$ itself is non-negligible.
More precisely, replacing the empirical average by its population counterpart gives $\epsilon_{\mathrm{pop}}^*(\theta_0)=\epsilon_0\left\{1-\frac{\rho_\gamma(\theta_0)}{C_\gamma(\theta_0)}\right\}$,
whenever the corresponding solution lies in the interior of the parameter space.
The latter expression makes explicit how non-zero tail overlap induces a deviation from exact recovery.

The DPD calculation above provides an explicit illustration of the role of the scale factor.
For a general choice of $\phi$, estimation of $\epsilon$ requires the H\"older risk to retain dependence on the scale $1-\epsilon$.
For example, the pseudo-spherical choice
$\phi(z)=-z^{1+\gamma}$ is scale-invariant:
\begin{align}
H_n^{(\gamma)}(\theta,\epsilon)
=
-
\left(
(1-\epsilon)^{-1}
\frac{S_{n,\gamma}(\theta)}
{C_\gamma(\theta)}
\right)^{1+\gamma}
(1-\epsilon)^{1+\gamma}
C_\gamma(\theta)
=
-
\frac{
\left(S_{n,\gamma}(\theta)\right)^{1+\gamma}
}{
\left(C_\gamma(\theta)\right)^\gamma
},
\label{eq:psuedo-spherical_holder}
\end{align}
so that $\epsilon$ disappears entirely from the risk.
Such a scale-invariant specification remains useful for inference on $\theta$ alone, but cannot identify the contamination proportion.
The particular choice above also violates the bounded-derivative condition used in \Cref{sec:theory}.
Scale-sensitive choices of $\phi$ suitable for joint inference are discussed in \Cref{sec:choice_phi}.

The profile calculation is used here only to expose the identification mechanism.
The H\"older posterior itself does not plug in the point estimate $\epsilon_n^*(\theta)$.
Instead, it treats $\epsilon$ as a parameter jointly with $\theta$, thereby quantifying uncertainty in the contamination proportion together with uncertainty in the fitted model.


\subsection{Posterior-Based Outlier Detection}
\label{sec:outlier_detection}

We now formalise how the joint H\"older posterior induces probabilistic outlier detection.
Suppose that we obtain posterior samples $\{(\theta^{(s)},\epsilon^{(s)})\}_{s=1}^S$.
Each draw determines both a fitted model parameter and the number of anomalous observations.
Specifically, we define the draw-specific outlier count by $K^{(s)} \coloneqq \operatorname{round}\{n\epsilon^{(s)}\}$.
Conditional on $\theta^{(s)}$, the observations are then ranked according to a pointwise compatibility score, typically the model log-density $\ell_i^{(s)}=\log p_{\theta^{(s)}}(y_i)$, or the conditional log-density $\ell_i^{(s)}=\log p_{\theta^{(s)}}(y_i\mid x_i)$ for regression.

Define $I_i^{(s)}\coloneqq 1$ if the compatibility score $\ell_i^{(s)}$ of observation $i$ falls among the $K^{(s)}$ smallest values, and $I_i^{(s)}\coloneqq 0$ otherwise.
Averaging these draw-specific indicators yields the FoD score
\begin{equation}
\operatorname{FoD}_i
=\frac{1}{S}
\sum_{s=1}^SI_i^{(s)}.
\label{eq:fod}
\end{equation}
The FoD therefore represents the posterior frequency with which observation $i$ is included in the draw-specific audit set.
It marginalises over uncertainty in both the fitted model and the implied outlier count.
Importantly, the FoD is not the posterior probability of an assumed latent contamination indicator; no such indicator or probabilistic model for the outlier-generating distribution is introduced.
Rather, it quantifies posterior uncertainty about whether an observation is selected by the model-based audit rule.

We refer to the number of observations selected for inspection as the audit budget.
If a fixed audit set with deterministic size and membership is required, it can be constructed by summarising the posterior distribution of the outlier count.
For example, the mean-count rule uses
\begin{align}
\widehat K_{\mathrm{mean}}=\operatorname{round}\left(\frac{1}{S}\sum_{s=1}^SK^{(s)}\right),
\end{align}
and selects the $\widehat K_{\mathrm{mean}}$ observations with the highest FoD scores.
Likewise, lower and upper posterior quantiles of $\{K^{(s)}\}_{s=1}^S$ provide smaller and larger audit budgets that reflect uncertainty in the contamination proportion.
The posterior over $\epsilon$ determines how many observations are selected for inspection, while the compatibility rankings induced by $\theta$ determine which observations are selected.
The FoD integrates these two sources of posterior uncertainty without requiring an externally specified contamination proportion or anomaly threshold.
We use this mechanism throughout our experiments to compute FoD summaries and construct likelihood-based audit rules. Further computational details are provided in \Cref{app:experimental-settings}.

\section{Robustness, Risk Bounds, and Asymptotics}
\label{sec:theory}

We next establish the theoretical properties of the H\"older posterior under the heavy-contamination regime.
For notational convenience, throughout this section we write $\alpha\coloneqq 1-\epsilon$ for the inlier proportion and parameterise the H\"{o}lder posterior by $(\theta,\alpha)$.
Accordingly, the data-generating distribution is written as $\mathbb{P} = \alpha_0 \mathbb{P}_{0} + ( 1 - \alpha_0 ) \mathbb{Q}_0$,
where $\alpha_0=1-\epsilon_0$ is the true inlier proportion.

The results below address three complementary aspects of the H\"{o}lder posterior.
First, we establish global bias-robustness through the posterior influence function (PIF), characterising its sensitivity to an infinitesimal outlying perturbation.
Second, we derive a finite-sample excess-risk bound and quantify how closely the population H\"{o}lder-risk minimiser tracks the true parameter pair under heavy contamination.
The latter result makes explicit the role of the tail-overlap quantity introduced in~\Cref{def:tail_overlap}.
Finally, we establish a Bernstein--von Mises theorem and quantify the stability of its local Gaussian covariance relative to the curvature of the same population risk at the true parameter pair.

We use the following basic setup throughout this section.

\begin{assump}[Basic regularity setup]
\label{assump:basic_regular}
The following conditions hold:
\begin{enumerate}
\item The observations $Y_1,\ldots,Y_n$ are i.i.d.~from $\mathbb{P} = \alpha_0 \mathbb{P}_{\theta_0} + (1 - \alpha_0 ) \mathbb{Q}_0$;

\item The model density $p_\theta(x)$ is positive and uniformly bounded by some constant $M$ over $\mathcal{X} \times \Theta$;

\item The function $\phi$ is differentiable on $(0,\infty)$ and satisfies
$
    -L
    \le
    \phi'(z)
    \le
    0
$
for some constant $L>0$.
\end{enumerate}
\end{assump}

\subsection{Bias Robustness}
\label{sec:robustness}

The posterior influence function (PIF) measures the first-order sensitivity of the posterior density to point-mass contamination~\citep{Ghosh16}.
Uniform boundedness of the PIF establishes global bias-robustness~\citep{Matsubara22}.

To define the PIF, consider an infinitesimal point-mass perturbation of the empirical distribution, $\mathbb{P}_{n,\varepsilon,z}\coloneqq(1-\varepsilon)\mathbb{P}_n+\varepsilon\delta_z$,
where $\delta_z$ denotes the Dirac measure at $z$.
The PIF at $(\theta,\alpha)$ is the corresponding directional derivative of the posterior density:
\begin{align}
    \mathrm{PIF}
    (
        z,\theta,\alpha,\mathbb{P}_n
    )
    \coloneqq
    \lim_{\varepsilon\to0}
    \frac{
        \pi_n^H
        (
            \theta,\alpha
            \mid
            \mathbb{P}_{n,\varepsilon,z}
        )
        -
        \pi_n^H
        (
            \theta,\alpha
            \mid
            \mathbb{P}_n
        )
    }{
        \varepsilon
    }.
\end{align}
Here $\pi_n^H(\theta,\alpha\mid\mathbb{P}_{n,\varepsilon,z})$ is the posterior obtained from the perturbed empirical distribution.
The following result gives the PIF and establishes its uniform boundedness.

\begin{restatable}[PIF and Global Bias-Robustness]{thm}{influenceHolder}
\label{thm:IF_holder}
Suppose that \Cref{assump:basic_regular} holds.
For a fixed dataset, assume that the posterior density
$\pi_n^H(\theta,\alpha\mid\mathbb{P}_n)$ is bounded on
$\Theta\times(0,1]$.
Then the PIF of the H\"{o}lder posterior is given by
\begin{align}
    \mathrm{PIF}
    (
        z,\theta,\alpha,\mathbb{P}_n
    )
    &=
    -n\beta\,
    \pi_n^H
    (
        \theta,\alpha
        \mid
        \mathbb{P}_n
    ) \times
    \bigg[
        \mathrm{IF}
        (
            z,\theta,\alpha,\mathbb{P}_n
        )
        -
        \mathbb{E}_{(\bar\theta,\bar\alpha)\sim
        \pi_n^H(\cdot\mid\mathbb{P}_n)}
        \left[
            \mathrm{IF}
            (
                z,\bar\theta,\bar\alpha,\mathbb{P}_n
            )
        \right]
    \bigg],
    \label{eq:pif_holder}
\end{align}
where
$\mathrm{IF}(z,\theta,\alpha,\mathbb{P}_n)$ denotes the standard influence function of the empirical H\"{o}lder risk at $(\theta,\alpha)$:
\begin{align}
    \mathrm{IF}(z,\theta,\alpha,\mathbb P_n)
    =\alpha^\gamma
      \phi'\!\left(\frac{S_{n,\gamma}(\theta)}
                        {\alpha C_\gamma(\theta)}\right)
      \{p_\theta^\gamma(z)-S_{n,\gamma}(\theta)\}.
    \label{eq:influence_function_holder}
\end{align}
Consequently, the PIF is uniformly bounded over all $z,\theta,\alpha$, so the H\"{o}lder posterior is globally bias-robust.

\end{restatable}

The proof is provided in \Cref{proof:IF_holder}.

\Cref{thm:IF_holder} also explains the behaviour in the tails.
If $p_\theta^\gamma(z)\to0$ as $\|z\|_2\to\infty$, the risk influence in Eq.~\eqref{eq:influence_function_holder} converges to
$-\alpha^\gamma\phi'\!\{S_{n,\gamma}(\theta)/[\alpha C_\gamma(\theta)]\}S_{n,\gamma}(\theta)$,
which is independent of $z$.
Thus, increasingly extreme observations do not produce unbounded first-order influence; the centred expression in Eq.~\eqref{eq:pif_holder} transfers this bounded sensitivity to the posterior.

\subsection{Excess-Risk and Population Bias Bounds}
\label{sec:consistency}

Next, we establish a finite-sample excess-risk bound for the H\"{o}lder posterior.
Let
\begin{align}
    \eta^*
    \coloneqq
    (\theta^*,\alpha^*)
    \in
    \argmin_{(\theta,\alpha)\in\Theta\times(0,1]}
    H^{(\gamma)}(\theta,\alpha)
    \label{eq:optimal_param}
\end{align}
denote a population minimiser of the H\"{o}lder risk defined in Eq.~\eqref{eq:hd_main}.
The result below controls the posterior-averaged excess population risk relative to its minimum
$H^{(\gamma)}(\theta^*,\alpha^*)$.

We assume Lipschitz regularity of the powered model densities and sufficient prior mass near the population minimiser.

\begin{assump}[Lipschitz regularity and prior mass]
\label{assump:prior_mass}
The following conditions hold:
\begin{enumerate}
\item The parameter set $\Theta\subseteq\mathbb R^d$ is bounded, with
$d<\infty$.
There exists a measurable $G_\gamma:\mathcal X\to[0,\infty)$ with
$\mathbb E_{Y\sim\mathbb P}[G_\gamma(Y)]<\infty$ such that $|p_\theta^\gamma(x)-p_{\theta'}^\gamma(x)|
    \le G_\gamma(x)\|\theta-\theta'\|_2$, where $x\in\mathcal X$ and $\theta,\theta'\in\Theta$. 
\item There exist constants $c_1,c_2>0$ such that $\int_{B_n(c_1)}\pi(\theta,\alpha)\,
       \mathrm d\theta\,\mathrm d\alpha
    \ge \exp(-c_2\sqrt n)$,
where, writing $\eta=(\theta,\alpha)$,
\[
    B_n(c_1)\coloneqq
    \left\{\eta\in\Theta\times(0,1]:
    H^{(\gamma)}(\eta)-H^{(\gamma)}(\eta^*)
    \le c_1/\sqrt n\right\}.
\]
\end{enumerate}
\end{assump}

On a bounded convex parameter domain, the Lipschitz condition follows from an integrable envelope for $\nabla_\theta p_\theta^\gamma$.
The prior mass condition is standard in posterior concentration theory and requires sufficient mass near the population minimiser~\citep{Ghosal00}.

These conditions yield the following bound.

\begin{restatable}[Posterior Excess-Risk Bound]{thm}{Posconholder}
\label{thm:consistency_holder}
Suppose that \Cref{assump:basic_regular,assump:prior_mass} hold.
Then, for any $\delta\in(0,1)$, with probability at least $1-\delta$,
the following bound holds for all $\beta>0$:
\[
    \mathbb{E}_{(\bar\theta,\bar\alpha)\sim\pi_n^H}
    \left[
        H^{(\gamma)}(\bar\theta,\bar\alpha)
    \right]
    -
    H^{(\gamma)}(\theta^*,\alpha^*)
    \le
    \frac{
        c_1+c_2/\beta
        +
        2LC\sqrt{\log(2/\delta)}
    }{
        \sqrt n
    },
\]
where $L$ is the constant in \Cref{assump:basic_regular} and $C>0$ is a
finite constant that may depend on the model family, $\gamma$ and
$\mathbb P$, but is independent of $n$, $\delta$ and $\beta$.
\end{restatable}

The proof is provided in Appendix~\ref{proof:consistency_holder}.
For any fixed temperature $\beta>0$ and confidence level $\delta$, \Cref{thm:consistency_holder} shows that the posterior-averaged excess population risk is of order $n^{-1/2}$.
This result concerns learning relative to the population optimum $\eta^*$; by itself, it does not determine how $\eta^*$ relates to the true parameter pair $\eta_0\coloneqq(\theta_0,\alpha_0)$.

We next quantify this population-level discrepancy under heavy contamination.
The following result shows that the possible displacement of the population minimiser from $\eta_0$ is controlled by the tail-overlap quantity introduced in \Cref{sec:heavy_contamination}.
It is an analogue of the error-control argument of~\citet{Kanamori13,Kanamori15}, extended to the joint parameter $(\theta,\alpha)$ and to a general choice of $\phi$ satisfying \Cref{assump:basic_regular}.

\begin{restatable}[Population Bias Control]{thm}{Tmp}
\label{thm:error_control_holder}
Suppose \Cref{assump:basic_regular} holds.
Define the level set
\[
    U_\rho
    \coloneqq
    \left\{
        (\theta,\alpha)\in\Theta\times(0,1]:
        H^{(\gamma)}(\theta,\alpha)
        \le
        H^{(\gamma)}(\theta^*,\alpha^*)
        +
        L\epsilon_0\rho_\gamma(\theta^*)
    \right\},
\]
where $\epsilon_0 = 1 - \alpha_0$ and $L$ is the constant in \Cref{assump:basic_regular}.
Then, we have $(\theta_0,\alpha_0)\in U_\rho$.
Assume further that $U_\rho$ is contained in a convex neighbourhood on which $H^{(\gamma)}$ is $\delta$-strongly convex for some $\delta>0$. If $H^{(\gamma)}$ is twice differentiable on this neighbourhood, this condition is equivalently expressed as
$
\nabla_\eta^2 H^{(\gamma)}(\eta)
    \succeq
    \delta I
$
throughout that neighbourhood. 
Then
\[
    \left\|
        (\theta^*,\alpha^*)
        -
        (\theta_0,\alpha_0)
    \right\|_2^2
    \le
    \frac{2L}{\delta}
    \epsilon_0
    \rho_\gamma(\theta^*).
\]
\end{restatable}

The proof is provided in Appendix~\ref{sec:proof_error_control_holder}.
\Cref{thm:error_control_holder} formalises the key robustness intuition in the heavy-contamination regime.
Although the contamination proportion remains fixed and may be non-negligible, the population minimiser remains close to the true parameter pair when the tail-overlap quantity $\rho_\gamma(\theta^*)$ is small.
Thus, robust inference does not rely on the contamination proportion vanishing with the sample size; rather, contamination concentrated in low-density regions of the fitted target model has only a limited effect on the population optimum through the density-power weighting.

Together, \Cref{thm:consistency_holder,thm:error_control_holder} separate two sources of discrepancy.
The former controls the finite-sample excess risk of the H\"{o}lder posterior relative to the population optimum, whereas the latter controls the population-level displacement of that optimum from the true parameter pair under contamination.

\subsection{Bernstein--von Mises Theorem and Local Covariance Stability}
\label{sec:bvm}

Finally, we study the asymptotic behaviour of the H\"{o}lder posterior.
The Bernstein--von Mises (BvM) theorem below shows that, after centring at the empirical risk minimiser and rescaling by $\sqrt n$, the H\"{o}lder posterior converges to a Gaussian distribution whose covariance is determined by the local curvature of the population H\"{o}lder risk at its minimiser.
We then compare this centred Gaussian limit with a centred reference Gaussian obtained by evaluating the curvature of the same population risk at the true parameter pair.
This comparison complements the population location-bias bound in \Cref{thm:error_control_holder} by controlling the associated change in local covariance.

To establish the BvM theorem for the H\"{o}lder posterior, we impose the following identifiability, separation and local smoothness conditions; see, e.g., \citet{van2000asymptotic,Miller21}.

\begin{assump}[Identifiability and risk separation]
\label{assump:existence_minimizer_unique}
The following conditions hold.
\begin{enumerate}
\item The population H\"{o}lder risk $H^{(\gamma)}$ has a unique minimiser $\eta^*$, as defined in Eq.~\eqref{eq:optimal_param}, in $\operatorname{int}(\Theta)\times(0,1)$;

\item There exists a bounded, convex, open neighbourhood $U$ of $\eta^*$ such that
$\bar U\subset\operatorname{int}(\Theta)\times[\underline{\alpha},1)$
for some $\underline{\alpha}>0$.
Almost surely, for every $\varepsilon>0$,
\begin{equation}
    \liminf_{n\to\infty}
    \inf_{\substack{
        \eta\in\Theta\times(0,1]\\
        \|\eta-\eta^*\|_2\ge\varepsilon
    }}
    \left\{
        H_n^{(\gamma)}(\eta)-H_n^{(\gamma)}(\eta^*)
    \right\}>0.
    \label{eq:empirical_risk_separation}
\end{equation}
\end{enumerate}
\end{assump}

Risk separation is standard in BvM theory~\citep[Theorem~4]{Miller21}.
It holds, for example, under uniform almost sure convergence of the empirical risk when the population risk has a positive gap outside every neighbourhood of its minimiser.

\begin{assump}[Local Smoothness Condition]
\label{assump:differentiable_q}
At each $x\in\mathcal{X}$, the scaled model density $(\theta,\alpha) \mapsto \alpha p_\theta(x)$ is three times continuously differentiable with respect to $\eta=(\theta,\alpha)$ on $\bar U$.
The function $\phi$ is of class $C^3$ on $(0,\infty)$.
Furthermore, for $r=1,2,3$, there exist envelope functions
$g_{\gamma,r}$ and $g_{1+\gamma,r}$ such that $\sup_{(\theta,\alpha)\in\bar U}\left\|\nabla_{\eta}^{r}
\left\{\alpha p_\theta(x)\right\}^{\gamma}\right\|\le g_{\gamma,r}(x)$, 
and $\sup_{(\theta,\alpha)\in\bar U}\left\|\nabla_{\eta}^{r}\left\{\alpha p_\theta(x)\right\}^{1+\gamma}\right\|\le g_{1+\gamma,r}(x)$,
for which we assume that $\mathbb{E}_{Y\sim\mathbb{P}}[g_{\gamma,r}(Y)]<\infty$ and $\int g_{1+\gamma,r}(x) \mathrm{d}\mu(x)<\infty$.
\end{assump}

\Cref{assump:differentiable_q} is a local rather than global smoothness condition.
For regular parametric models, its derivative and envelope requirements can typically be verified on the compact set $\bar U$ provided that $\bar U$ excludes singular configurations, such as $\alpha=0$ or degenerate scale parameters.

The derivation of the BvM theorem relies on a series of technical lemmas collected in Appendix~\ref{app:useful_lemma}.
Specifically, Lemma~\ref{lem:uniform_conv_holder} establishes uniform convergence of the empirical H\"{o}lder risk on $\bar U$, and Lemma~\ref{lem:strong_consistency} establishes eventual existence and strong consistency of every sequence of empirical minimisers.
Lemmas~\ref{lem:bounded_derivative_first}--\ref{lem:bounded_derivative_third} provide uniform bounds on the derivatives of the empirical H\"{o}lder risk up to third order, which are used in the Taylor expansion underlying the BvM approximation.
The proof of the following theorem is provided in Appendix~\ref{sec:proof_bvm_holder}.

\begin{thm}[BvM Theorem for H\"{o}lder--Bayes]
\label{thm:bvm_holder}
Suppose that
\Cref{assump:basic_regular,assump:existence_minimizer_unique,assump:differentiable_q}
hold.
Fix the temperature $\beta\in(0,\infty)$ in Eq.~\eqref{eq:holder_posterior_def}.
For all sufficiently large $n$, let $\eta_n$ be any empirical minimiser of $H_n^{(\gamma)}$, whose existence is established in \Cref{lem:strong_consistency}.
Assume that the prior density $\pi$ is continuous and positive at the population minimiser $\eta^*$.
Let $\bm{J}_* \coloneqq \nabla_\eta^2 H^{(\gamma)}(\eta^*)$ denote the Hessian matrix of the population risk $H^{(\gamma)}$ at $\eta^*$.
Assume that $\bm{J}_*$ is positive definite.
Let $\tilde{\pi}_n^H$ denote the density of $\xi\coloneqq \sqrt{n}(\eta-\eta_n)$, where $\eta\sim\pi_n^H$.
Then the H\"{o}lder posterior is asymptotically normal in the sense that
\[
    \int_{\mathbb{R}^{d+1}}
    \left|
        \tilde{\pi}_n^H(\xi)
        -
        \mathcal{N}
        \left(
            \xi
            \mid
            \mathbf{0},
            \bm{\Sigma}_*
        \right)
    \right|
    \mathrm{d}\xi
    \to
    0
    \quad
    \text{a.s.},
\]
where $\mathcal{N}(\cdot\mid\mathbf{0},\bm{\Sigma}_*)$ denotes the normal density with zero mean vector $\mathbf{0}$ and covariance matrix $\bm{\Sigma}_* \coloneqq ( \beta\bm{J}_* )^{-1}$.
\end{thm}

\begin{rem}[Boundary case]
\label{rem:interior_boundary_bvm}
\Cref{thm:bvm_holder} concerns the local behaviour around the population minimiser $\eta^*$ and therefore requires $\eta^*$, rather than the uncontaminated target $\eta_0$, to lie in the interior of the parameter space.
The covariance comparison below additionally assumes that $\eta_0$ lies in the same interior neighbourhood $U$.
In particular, the case $\alpha_0=1$ lies outside the scope of this comparison.
A BvM analysis at a boundary population minimiser would require boundary-constrained asymptotics, which we leave for future work.
\end{rem}

We now work on the same centred $\sqrt n$ scale as \Cref{thm:bvm_holder}.
Let $\bm{J}(\eta)\coloneqq\nabla_\eta^2H^{(\gamma)}(\eta)$, and define
\[
    \bm{J}_0\coloneqq\bm{J}(\eta_0),
    \qquad
    \bm{\Sigma}_0\coloneqq(\beta\bm{J}_0)^{-1},
    \qquad
    \mathsf{G}_*\coloneqq\mathcal{N}(\mathbf{0},\bm{\Sigma}_*),
    \qquad
    \mathsf{G}_0\coloneqq\mathcal{N}(\mathbf{0},\bm{\Sigma}_0).
\]
Here $\mathsf{G}_*$ is the Gaussian limit in \Cref{thm:bvm_holder}, and $\mathsf{G}_0$ is a reference Gaussian defined through the curvature at the true parameter pair $\eta_0$.
Both covariance matrices are computed from the same population risk $H^{(\gamma)}$, evaluated under the contaminated distribution $\mathbb{P}$.
Thus, this comparison keeps the population risk fixed and changes the point at which its curvature is evaluated.
The location difference $\eta^*-\eta_0$ is controlled separately by \Cref{thm:error_control_holder}.

\begin{thm}[Stability of the Centred Gaussian Limit]
\label{thm:bvm_holder_corollary}
Suppose that the conditions of \Cref{thm:bvm_holder,thm:error_control_holder} hold, and assume that $\eta_0\in U$.
Assume that there exist constants
$0<\underline{\lambda}\le\overline{\lambda}<\infty$
and $L_J<\infty$ such that, for all $\eta,\eta'\in U$, $\underline{\lambda}I_{d+1}\preceq\bm{J}(\eta) \preceq \overline{\lambda}I_{d+1}$ and $\left\|\bm{J}(\eta)-\bm{J}(\eta')\right\|_{\mathrm{F}}\le L_J\|\eta-\eta'\|_2$.
Then
\begin{align}
    W_2^2(\mathsf{G}_*,\mathsf{G}_0)
    &\le
    \frac{\overline{\lambda}L_J^2}
         {4\beta\underline{\lambda}^4}
    \|\eta^*-\eta_0\|_2^2
    \le
    C_\beta\epsilon_0\rho_\gamma(\theta^*),
    \label{eq:wasserstein_bias_bound}
\end{align}
where $C_\beta\coloneqq\frac{L\overline{\lambda}L_J^2}{2\beta\delta\underline{\lambda}^4}$
and $L$ and $\delta$ are the constants in \Cref{thm:error_control_holder}.
\end{thm}

The proof is provided in Appendix~\ref{sec:proof_bvm_holder_corollary}.
Since both Gaussian distributions are centred at zero, Eq.~\eqref{eq:wasserstein_bias_bound} measures only their covariance discrepancy.
The first inequality shows how the regularity of the population Hessian transfers the displacement of the minimiser into a change in local Gaussian covariance.
The second inequality then controls this change by the same tail-overlap quantity that governs the population location bias.
Together, \Cref{thm:error_control_holder,thm:bvm_holder_corollary} distinguish the displacement of the population centre from the change in local Gaussian covariance when the curvature is evaluated at $\eta^*$ instead of $\eta_0$.
For fixed regularity constants and temperature, both bounds remain small when $\epsilon_0\rho_\gamma(\theta^*)$ is small; the contamination proportion need not vanish with the sample size.

The results in \Cref{thm:bvm_holder,thm:bvm_holder_corollary} are stated for a fixed temperature $\beta$.
The next section develops the interpretation and data-driven selection of $\beta$.
The following remark records the corresponding conclusions for a temperature with a positive finite deterministic limit.

\begin{rem}
\label{rem:data_dependent_temperature_bvm}
Let $\beta_n$ be a measurable, parameter-independent, data-dependent temperature satisfying $\beta_n\to\beta_\infty$ almost surely for some non-random constant $\beta_\infty\in(0,\infty)$.
Under the remaining conditions of \Cref{thm:bvm_holder}, its conclusion holds with limiting covariance $(\beta_\infty\bm{J}_*)^{-1}$.
The formal statement for the affine-scaling temperature is given in \Cref{app:bvm_data_dependent_temperature}.
Under the additional conditions of \Cref{thm:bvm_holder_corollary}, the centred Gaussian covariance comparison holds with $\beta$ replaced by $\beta_\infty$ in both covariance matrices and in $C_\beta$; see \Cref{sec:proof_bvm_bias_data_dependent_temperature}.
\end{rem}

\section{Practical Settings for H\"older--Bayes}
\label{sec:practical_use}

This section derives a data-driven default for the posterior temperature from affine scaling and specifies the function $\phi$ for joint inference.

\subsection{Affine-Scaling Interpretation of the Posterior Temperature}
\label{sec:temp_choice}
The temperature parameter $\beta$ controls the scale of the loss in a generalised posterior.
Existing calibration procedures are motivated by several different criteria, including matching generalised-posterior credible regions to frequentist confidence regions~\citep{Lyddon19,Wu2023}.
For the H\"older posterior, however, the temperature also admits a direct geometric interpretation through the affine scaling law of the H\"older divergence.

Consider an affine transformation $T(y)=\Omega^{-1/2}(y - m)$,
where $\Omega$ is positive definite and $m\in\mathbb{R}^k$.
Let $T_\#\mathbb{P}$ denote the pushforward measure of a distribution $\mathbb{P}$ under $T$.
Then, the H\"older divergence satisfies
\begin{align}
    \mathrm{D}_\phi^\gamma
    \left(
        T_\#\mathbb{P}
        \mid
        T_\#\mathbb{Q}
    \right)
    =
    \beta_T
    \mathrm{D}_\phi^\gamma
    \left(
        \mathbb{P}
        \mid
        \mathbb{Q}
    \right),
    \qquad
    \beta_T
    =
    |\det\Omega|^{\gamma/2}.
    \label{eq:affine_constant}
\end{align}
Thus, an affine transformation changes the H\"older divergence only through the associated volume-scaling factor $\beta_T$.

The same relation extends directly to the empirical H\"older risk.
Let $H_{n,T}^{(\gamma)}(\theta,\epsilon)$ denote the risk evaluated after applying $T$ to both the empirical distribution and the model.
Then, $H_{n,T}^{(\gamma)}(\theta,\epsilon)=\beta_T H_n^{(\gamma)}(\theta,\epsilon)$.
Consequently,
\[
    \exp\left\{
        -\beta n H_n^{(\gamma)}(\theta,\epsilon)
    \right\}
    =
    \exp\left\{
        -n H_{n,T}^{(\gamma)}(\theta,\epsilon)
    \right\}.
\]
Hence, choosing the posterior temperature is equivalent to choosing an implicit volume scale for the data space.
In particular, any $\beta>0$ can be realised by an affine transformation whose determinant yields the corresponding scaling factor.

This interpretation suggests a natural data-driven default based on standardisation.
For a sample of size $n$, let
\begin{align}
    T_n(y) = \widehat{\Omega}_n^{-1/2}(y-\widehat{m}_n) \quad \text{and} \quad \beta_{n} = |\det\widehat{\Omega}_n|^{\gamma/2}, \label{eq:beta_selection}
\end{align}
where $\widehat m_n$ and $\widehat{\Omega}_n$ are the sample mean and sample covariance, respectively.
Using $\beta_n$ in the H\"older posterior on the original data is therefore equivalent, in terms of the H\"older risk, to applying unit temperature after standardising the data and the model.

We use Eq.~\eqref{eq:beta_selection} as the default temperature throughout our experiments.
The sample covariance provides a simple closed-form choice that implements ordinary data standardisation and is computed only once before posterior inference.
More generally, the affine-scaling interpretation is not tied to this particular estimator: other positive-definite, affine-equivariant scatter estimators could be used when a more robust notion of scale is desired.

Unlike calibration procedures that require iterative optimisation, repeated posterior sampling, or curvature estimation~\citep{Grunwald17,Holmes17,Lyddon19}, the rule in Eq.~\eqref{eq:beta_selection} incurs virtually no additional computational overhead for posterior inference. 
The asymptotic results in \Cref{sec:bvm} continue to apply to a data-dependent temperature whenever $\beta_n$ converges almost surely to a deterministic limit in $(0,\infty)$.
For the covariance-based choice above, this condition is satisfied whenever the sample covariance converges almost surely to a deterministic, finite, positive-definite matrix.

\subsection{Default Choice of \texorpdfstring{$\phi$}{phi} for Joint Inference}
\label{sec:choice_phi}

We choose $\phi$ to satisfy $\phi(1)=-1$, $\phi(z)\ge-z^{1+\gamma}$, and $-L\le\phi'(z)\le0$, while retaining dependence of the risk on $1-\epsilon$.
Our default is $\phi_{\mathrm{DPD}}(z)=\gamma-(1+\gamma)z$: it satisfies these conditions, yields the explicit contamination-proportion calculation in \Cref{sec:scaled_mass_risk_interpretation}, and builds on the established DPD framework~\citep{Basu98}.
The pseudo-spherical choice $\phi(z)=-z^{1+\gamma}$ is unsuitable because its risk does not depend on $\epsilon$.

More generally, admissible choices can be constructed as
\begin{align}
    \phi_s(z)=\gamma-(1+\gamma)s(z),
    \label{eq:generalized_dpd_phi}
\end{align}
where $s(1)=1$, $s(z)\le z$, and $0\le s'(z)\le M_s<\infty$ on $(0,\infty)$.
These conditions imply $\phi_s(1)=-1$,
$\phi_s(z)\ge\gamma-(1+\gamma)z\ge-z^{1+\gamma}$, and
$-(1+\gamma)M_s\le\phi_s'(z)\le0$.
The linear choice $s(z)=z$ gives our default.
We also examine exponential and rational saturating choices that retain dependence on $1-\epsilon$; their definitions and comparisons are provided in \Cref{app:alternative_phi,app:additional-experimental-results}.

\section{Experiments}
\label{sec:experiment}

The experiments examine joint recovery of the target model and contamination proportion, robustness as contamination increases, and the translation of posterior uncertainty into observation-level auditing and data cleaning.
We first use synthetic data with a known target and contamination mechanism, then real regression datasets with injected anomalies.
H\"older--Bayes (H-B) uses the default DPD choice of $\phi$ from \Cref{sec:choice_phi} and the affine-scaling temperature of \Cref{sec:temp_choice}.
For conditional regression, the temperature rule is applied to the observed training-response variance in the coordinates supplied to the likelihood.
Regression H-B fits and all $\gamma$-divergence fits use a minimum Gaussian standard deviation of $0.05$ in their fitting coordinates.
Four Markov chain Monte Carlo (MCMC) chains are used; all reported fits satisfy rank-normalised $\widehat R\le1.01$ and bulk and tail effective sample sizes of at least $400$.
All planned replicates contribute to the reported aggregates.
Prediction is assessed using root mean squared error (RMSE), and outlier ranking using the area under the precision--recall curve (AUPRC).
Prior specifications, inference details, and sensitivity analyses are provided in \Cref{app:experimental-settings,app:additional-experimental-results}.

\subsection{Synthetic Evaluation}
\label{sec:simulation-study}

The first setting is a bivariate location--scale experiment,
\[
    Y_i\sim(1-\epsilon_0)\mathcal{N}(0,I_2)
       +\epsilon_0\mathcal{N}((5,5)^\top,I_2),
\]
where the fitted target model is a diagonal Gaussian.
We use $n=1000$ observations and vary $\epsilon_0$ over $\{0,0.1,0.2,0.3,0.4\}$.

The second setting is contaminated linear regression, with $1000$ training observations and an independent clean test set of size $1000$.
The clean response follows
\[
    Y_i=(1,X_i^\top)\beta^*+\xi_i,
    \qquad X_i\sim\mathcal{N}(0,I_2),
    \qquad \xi_i\sim\mathcal{N}(0,1),
\]
where $\beta^*=(0,\beta_1^*,\beta_2^*)^\top$ includes a zero intercept.
Each training observation is selected independently with probability $\epsilon_0$.
For selected observations, the covariates are unchanged and the response is replaced by
\[
    Y_i^{\mathrm{obs}}=(1,X_i^\top)\beta^*+6+\xi_i',
    \qquad \xi_i'\sim\mathcal{N}(0,1),
\]
with independent replacement noise.
Thus contamination shifts the intercept while preserving the slopes.
Here $\epsilon_0\in\{0,0.05,0.1,0.2,0.3,0.4\}$.
We generate $30$ independent datasets for every setting in each experiment.

We compare H\"older--Bayes with Standard Bayes under a Gaussian working likelihood, the $\gamma$-divergence posterior, and a posterior based on a Student-$t$ likelihood with four degrees of freedom.
The main figures fix $\gamma=0.5$ for location--scale estimation and $\gamma=1.0$ for regression across all contamination levels and metrics.
Sensitivity to $\gamma\in\{0.1,0.5,1.0\}$ is reported in \Cref{app:additional-experimental-results}.
Representative plots use one fixed replicate at $\epsilon_0=0.2$, shared across methods.

\begin{figure}[!tp]
    \centering
    \includegraphics[width=.95\linewidth]{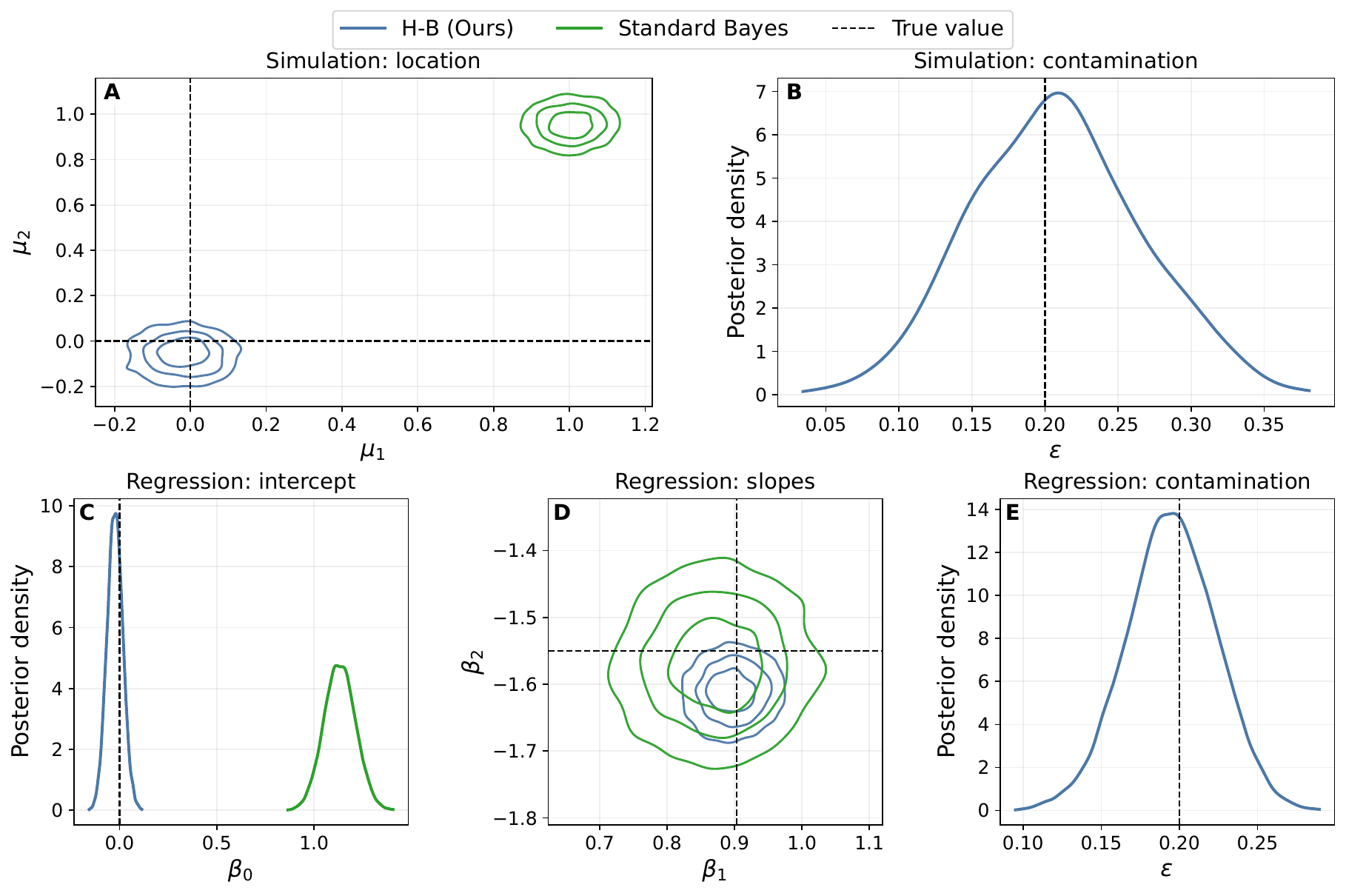}
    \caption{
    Target-parameter and contamination posteriors under $20\%$ contamination for one representative dataset per experiment.
    The top row shows the location--scale experiment with $\gamma=0.5$: (A) joint posterior contours for the two location parameters and (B) the marginal H\"older--Bayes (H-B) posterior for the contamination proportion $\epsilon$.
    The bottom row shows the regression experiment with $\gamma=1.0$: (C) the marginal posterior for the intercept $\beta_0$, (D) joint posterior contours for the slopes $(\beta_1,\beta_2)$, and (E) the marginal H-B posterior for $\epsilon$.
    Model-parameter panels compare H-B with Standard Bayes. Regression coefficients are displayed on the original response and covariate scales.
    Dashed lines indicate the true target-model parameters and the true contamination proportion $\epsilon_0=0.2$.
    }
    \label{fig:posterior}
    \par\smallskip
    \begin{minipage}{.95\linewidth}
    \normalsize\raggedright
    \end{minipage}
\end{figure}

In \Cref{fig:posterior}, Standard Bayes is displaced towards the contaminating cluster in the location--scale experiment, whereas H-B remains near the clean target.
In regression, the intercept reveals the response shift: the H-B posterior for $\beta_0$ lies near zero, while Standard Bayes is centred around $1.13$.
The slope posteriors overlap substantially.
In both examples, the H-B posterior for $\epsilon$ is centred close to $\epsilon_0=0.2$.

\begin{figure}[!tp]
    \centering
    \includegraphics[width=.95\linewidth]{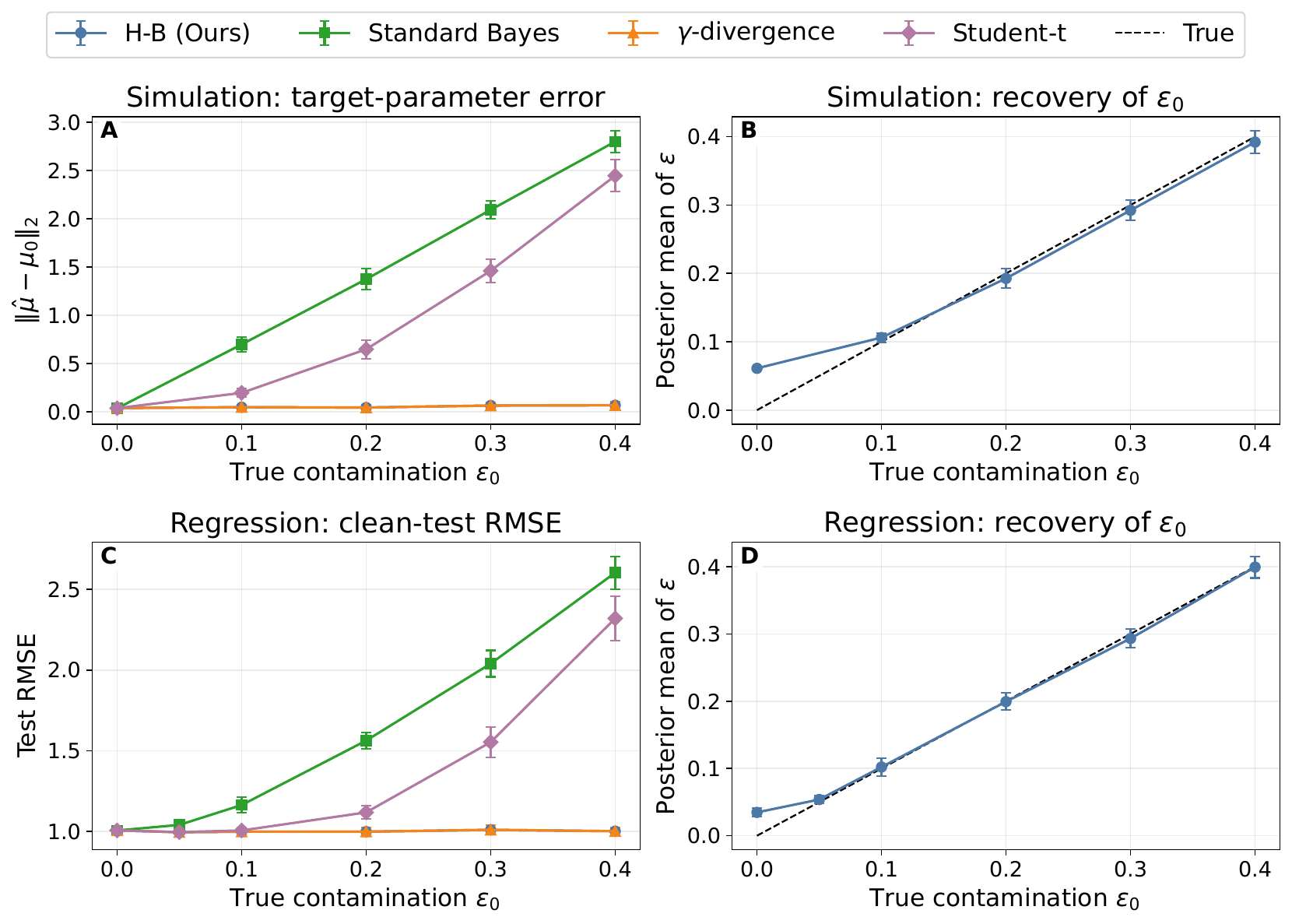}
    \caption{
    Robustness and contamination-proportion recovery as contamination increases.
    The top row shows location--scale estimation with $\gamma=0.5$; the bottom row shows regression with $\gamma=1.0$.
    Panels A and C compare target-model performance under H-B, Standard Bayes, the $\gamma$-divergence posterior, and the Student-$t$ likelihood.
    Panels B and D show the H-B posterior mean of $\epsilon$; dashed lines denote the identity relation.
    Points and error bars show the mean and sample standard deviation (SD), respectively, across all $30$ independent datasets per setting; the error bars do not represent within-dataset posterior intervals.
    }
    \label{fig:accuracy}
    \par\smallskip
    \begin{minipage}{.95\linewidth}
    \normalsize\raggedright
    \end{minipage}
\end{figure}

\Cref{fig:accuracy} shows that H\"older--Bayes and the $\gamma$-divergence posterior retain low target-parameter error in the location--scale experiment as contamination increases.
Standard Bayes and the Student-$t$ likelihood deteriorate substantially at higher contamination levels.
The same distinction appears in regression: H-B and the $\gamma$-divergence posterior maintain clean-test RMSE near the noise level of one, whereas Standard Bayes and Student-$t$ reach mean RMSEs of $2.60$ and $2.32$, respectively, at $\epsilon_0=0.4$.

At positive contamination levels, H-B posterior means of $\epsilon$ track $\epsilon_0$ through $0.4$ in both experiments.
In regression, the mean estimate at $\epsilon_0=0.4$ is $0.399$, with a between-dataset standard deviation of $0.016$.
The uncontaminated settings yield positive posterior means, approximately $0.061$ for location--scale estimation and $0.035$ for regression; this positive baseline should be considered when interpreting small inferred audit fractions.
The error bars describe between-dataset variability, while \Cref{fig:posterior,fig:likelihood-ranking} illustrate within-dataset posterior uncertainty.

\begin{figure}[!tp]
    \centering
    \includegraphics[width=.9\linewidth]{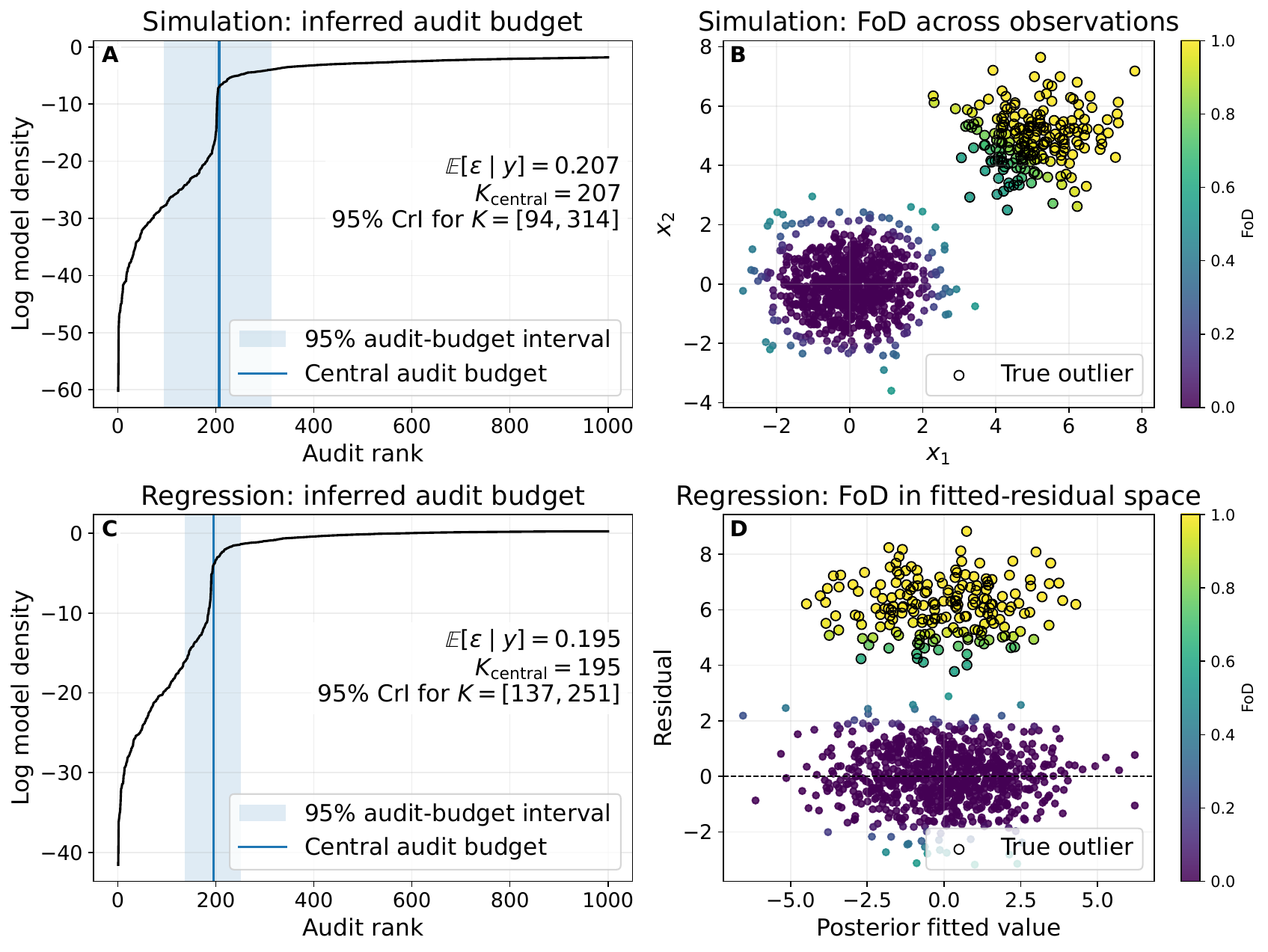}
    \caption{
    Posterior-guided auditing for the representative datasets under $20\%$ contamination.
    The top row shows location--scale estimation with $\gamma=0.5$; the bottom row shows regression with $\gamma=1.0$.
    Panels A and C order observations by increasing pointwise log model density under the posterior-mean fit.
    The vertical line marks the central budget $K_{\mathrm{central}}=\widehat K_{\mathrm{mean}}$ from the mean-count rule in \Cref{sec:outlier_detection}; shading shows a $95\%$ posterior interval for $K=\operatorname{round}(n\epsilon)$.
    Panels B and D colour observations by their FoD scores; black circles indicate injected outliers.
    Regression observations are displayed in fitted-value--residual space.
    }
    \label{fig:likelihood-ranking}
    \par\smallskip
    \begin{minipage}{.95\linewidth}
    \normalsize\raggedright
    \end{minipage}
\end{figure}

\Cref{fig:likelihood-ranking} separates two aspects of the posterior audit: uncertainty in $\epsilon$ determines how many observations to inspect, while uncertainty in the fitted model determines which observations are repeatedly selected.
The central budgets are $207$ for location--scale estimation and $195$ for regression, both close to the nominal count $n\epsilon_0=200$.
Their $95\%$ posterior intervals are $[94,314]$ and $[137,251]$, respectively.
FoD separates the contaminating cluster in the first experiment and concentrates on large residuals from the robust regression fit in the second.
Observations selected inconsistently across posterior draws receive intermediate scores.
Thus H-B supplies both a fitted model and posterior-derived audit budgets with observation-level scores.

Ranking-only comparisons and sensitivity to $\gamma$ and alternative choices of $\phi$ are reported in \Cref{app:additional-experimental-results}.
The competing models can supply compatibility rankings, but do not infer an audit budget through a joint posterior over the contamination proportion.

\subsection{Real-Data Auditing and Cleaning}
\label{sec:realdata-auditing-cleaning}

We use \texttt{cal-housing}, \texttt{kin8nm}, and \texttt{pumadyn-32fh}, with five random $80/20$ train--test splits per dataset shared across methods and settings.
Covariates and responses are standardised using the pre-injection training means and standard deviations only.
Each training response is selected independently with probability $\epsilon_0=0.2$ and receives an additive $\mathcal{N}(6,1)$ perturbation in standardised response units.
Test data are not contaminated, and injected labels are used only for evaluation and the oracle comparison.
The main real-data results use $\gamma=0.5$ throughout.
Predictive metrics are evaluated in the original response units; Monte Carlo negative log predictive density (MC-NLPD) averages the negative log of the posterior-averaged predictive density over test observations.

The Gaussian linear working model need not describe the original data exactly.
Accordingly, $\epsilon$ is interpreted as a model-relative audit fraction, which may reflect incompatibility with the working model as well as injected contamination.
Sensitivity analyses and the uncontaminated controls are reported in \Cref{app:additional-experimental-results}.

\Cref{tab:realdata-initial-audit} shows lower mean test RMSE for H-B than Standard Bayes on all three datasets, but the gains depend on the metric and dataset.
On \texttt{kin8nm} and \texttt{pumadyn-32fh}, H-B also improves MC-NLPD and attains FoD AUPRC values of approximately $0.999$.
On \texttt{cal-housing}, mean RMSE falls from $1.569$ to $0.966$, but its standard deviation across splits is $0.529$; mean MC-NLPD is slightly worse than Standard Bayes ($2.163$ versus $2.134$), with substantial split variability.
FoD AUPRC is $0.924$, and the mean audit fraction is $0.240$, exceeding the injected fraction $0.2$.
These results support the model-relative interpretation of the audit while showing that improved average point prediction does not ensure improved predictive density or equally accurate rankings on every dataset.
Further method comparisons are reported in \Cref{app:additional-experimental-results}.

\begin{table}[!htbp]
\centering
\caption{
Predictive and auditing performance under $20\%$ injected response contamination, with $\gamma=0.5$ for H-B.
RMSE and Monte Carlo negative log predictive density (MC-NLPD) are reported as mean (SD) over five train--test splits; lower values are better.
For the audit fraction, posterior means and endpoints of $95\%$ credible intervals (CrIs) are averaged separately across splits.
FoD AUPRC is averaged across splits.
Bold: lowest displayed predictive means within each dataset; all ties included.}
\label{tab:realdata-initial-audit}
\scriptsize
\resizebox{\textwidth}{!}{%
\begin{tabular}{llcccc}
\toprule
Dataset & Method & Test RMSE & Test MC-NLPD & Audit fraction & FoD AUPRC \\
& & Mean (SD) & Mean (SD) & Mean [mean $95\%$ CrI] & Mean \\
\midrule
\texttt{cal-housing} & Standard Bayes & $1.5691\;(0.0262)$ & $\mathbf{2.134}\;(0.009)$ & -- & -- \\
 & H-B (DPD) & $\mathbf{0.9661}\;(0.5291)$ & $2.163\;(2.025)$ & $0.240\;[0.224,\,0.256]$ & $0.924$ \\
\addlinespace
\texttt{kin8nm} & Standard Bayes & $0.3837\;(0.0103)$ & $0.687\;(0.016)$ & -- & -- \\
 & H-B (DPD) & $\mathbf{0.2063}\;(0.0033)$ & $\mathbf{-0.157}\;(0.018)$ & $0.206\;[0.177,\,0.234]$ & $0.999$ \\
\addlinespace
\texttt{pumadyn-32fh} & Standard Bayes & $0.0452\;(0.0012)$ & $-1.452\;(0.012)$ & -- & -- \\
 & H-B (DPD) & $\mathbf{0.0270}\;(0.0004)$ & $\mathbf{-2.193}\;(0.015)$ & $0.209\;[0.179,\,0.239]$ & $0.999$ \\
\bottomrule
\end{tabular}%
}
\end{table}

\Cref{tab:realdata-cleaning} compares downstream refits. For cleaning, observations are ranked by Gaussian log likelihood at the H-B posterior-mean fit, using the posterior-derived count budgets; see \Cref{app:realdata-benchmark}.
For \texttt{kin8nm} and \texttt{pumadyn-32fh}, the central H-B budget nearly matches the oracle in RMSE and MC-NLPD.
Lower budgets favour precision; upper budgets favour recall.
On \texttt{cal-housing}, all budgets reduce mean RMSE but have large split variability and higher MC-NLPD than no removal, while the oracle remains substantially better.

The \texttt{cal-housing} aggregates are sensitive to one split with a large prediction error at an extreme covariate.
All five splits are retained in both tables.
An exploratory analysis omitting the same split from every method, an examination of the unstable split, and predictive-coverage diagnostics are reported in \Cref{app:additional-experimental-results}.

\begin{table}[!htbp]
\centering
\caption{
Downstream cleaning using H-B audit budgets, with $\gamma=0.5$.
The same Standard Bayes Gaussian regression is refitted after removal and evaluated on the unchanged test set.
For $K=\operatorname{round}(n_{\mathrm{train}}\epsilon)$, the lower and upper budgets are the $2.5\%$ and $97.5\%$ posterior quantiles rounded down and up, respectively; the central budget is the rounded posterior mean of $K$.
The oracle prioritises injected outliers using the central H-B budget.
RMSE and MC-NLPD are mean (SD); removal counts, precision, and recall are split averages.
Bold: lowest displayed predictive means within each dataset, excluding the oracle reference; all ties included.}
\label{tab:realdata-cleaning}
\scriptsize
\resizebox{\textwidth}{!}{%
\begin{tabular}{llcccc}
\toprule
Dataset & Cleaning rule & Removed $K$ & Precision/recall & Test RMSE & Test MC-NLPD \\
& & Mean & Mean & Mean (SD) & Mean (SD) \\
\midrule
\texttt{cal-housing} & No removal & $0.0$ & -- & $1.5691\;(0.0262)$ & $\mathbf{2.134}\;(0.009)$ \\
 & Oracle (central) & $3968.2$ & $0.833\,/\,1.000$ & $0.7226\;(0.0091)$ & $1.093\;(0.012)$ \\
 & H-B lower budget & $3707.8$ & $0.891\,/\,1.000$ & $1.0274\;(0.6648)$ & $2.289\;(2.493)$ \\
 & H-B central budget & $3968.2$ & $0.833\,/\,1.000$ & $1.0017\;(0.6109)$ & $2.419\;(2.629)$ \\
 & H-B upper budget & $4222.4$ & $0.783\,/\,1.000$ & $\mathbf{0.9776}\;(0.5586)$ & $2.464\;(2.560)$ \\
\addlinespace
\texttt{kin8nm} & No removal & $0.0$ & -- & $0.3837\;(0.0103)$ & $0.687\;(0.016)$ \\
 & Oracle (central) & $1350.8$ & $0.978\,/\,1.000$ & $0.2038\;(0.0028)$ & $-0.172\;(0.014)$ \\
 & H-B lower budget & $1161.0$ & $1.000\,/\,0.879$ & $0.2070\;(0.0027)$ & $-0.106\;(0.016)$ \\
 & H-B central budget & $1350.8$ & $0.974\,/\,0.996$ & $\mathbf{0.2037}\;(0.0030)$ & $\mathbf{-0.170}\;(0.017)$ \\
 & H-B upper budget & $1532.8$ & $0.861\,/\,0.999$ & $0.2057\;(0.0032)$ & $-0.135\;(0.026)$ \\
\addlinespace
\texttt{pumadyn-32fh} & No removal & $0.0$ & -- & $0.0452\;(0.0012)$ & $-1.452\;(0.012)$ \\
 & Oracle (central) & $1371.0$ & $0.963\,/\,1.000$ & $0.0265\;(0.0003)$ & $-2.210\;(0.011)$ \\
 & H-B lower budget & $1164.0$ & $1.000\,/\,0.881$ & $0.0268\;(0.0004)$ & $-2.163\;(0.015)$ \\
 & H-B central budget & $1371.0$ & $0.960\,/\,0.997$ & $\mathbf{0.0266}\;(0.0003)$ & $\mathbf{-2.208}\;(0.012)$ \\
 & H-B upper budget & $1576.0$ & $0.837\,/\,0.999$ & $0.0270\;(0.0004)$ & $-2.169\;(0.023)$ \\
\bottomrule
\end{tabular}%
}
\end{table}

\FloatBarrier

\section{Conclusion}
\label{sec:conclusion}

H\"older--Bayes jointly infers a target model and a model-implied contamination proportion by applying the H\"older divergence to $(1-\epsilon)\mathbb{P}_\theta$, without specifying an outlier distribution.
The posterior supplies audit budgets, while FoD propagates uncertainty in the fitted model and contamination proportion into observation-level scores.

Under fixed, non-negligible contamination, we establish bounded posterior influence, a finite-sample $n^{-1/2}$ excess-risk bound, and a Bernstein--von Mises theorem.
Tail-overlap bounds control the population minimiser's displacement from the true parameter pair and the $2$-Wasserstein distance between the centred Gaussian limit and a centred reference Gaussian determined by the curvature of the population risk at that pair.
The posterior temperature also admits an affine-scaling interpretation and a closed-form, data-driven default.
Synthetic results demonstrate robust estimation and contamination-proportion recovery.
On real data, posterior-guided cleaning improves average test RMSE, while predictive-density accuracy and split variability remain dataset-dependent.

Open questions include data-driven selection of $\gamma$, theoretical characterisation of the model-relative audit fraction under misspecification, and extensions to dependent or hierarchical observations.
The \texttt{cal-housing} results motivate refits that account for observation selection and more flexible conditional models; whether these preserve both robustness and predictive-density accuracy remains to be established.

\section*{Data Availability Statement}
The numerical experiments use simulated data and the publicly available
\texttt{cal-housing}, \texttt{kin8nm}, and \texttt{pumadyn-32fh} datasets.
The simulation settings and data preprocessing are described in the
supplementary material (\Cref{app:experimental-settings}).
Code for generating the simulation data, preprocessing the benchmark
datasets and reproducing the numerical experiments is available at \url{https://github.com/msfuji0211/holder-bayes}.
The archive includes the inference implementation, experiment settings,
and scripts for generating the figures and tables.

\section*{Acknowledgment}
MF was supported by JSPS KAKENHI Grant Number JP25K21286, Japan.
MA was supported by JST, ARiSE, Grant Number JPMJAR2628, Japan.

\bibliography{reference}
\bibliographystyle{plainnat}

\clearpage
\appendix

\begin{center}
    \textbf{\huge Supplementary Material}
\end{center}

\vspace{40pt}

This supplementary material contains proofs of all theoretical results presented in the main text, together with additional details on the numerical experiments.
\Cref{app:useful_lemma} collects intermediate lemmas used throughout the proofs.
\Cref{sec:proofs} presents the proofs of the main theoretical results.
\Cref{app:experimental-settings} provides the additional experimental details.

\paragraph{Notation and Definitions}
Throughout the appendix, we write
\[
    \eta=(\theta,\alpha),
    \qquad
    \eta'=(\theta',\alpha'),
    \qquad
    \eta_n=(\theta_n,\alpha_n),
    \qquad
    \eta^*=(\theta^*,\alpha^*).
\]
For notational compactness in the proofs, we also write
\[
    q_\eta(x)
    \coloneqq
    \alpha p_\theta(x) ,
    \qquad
    Z_{1+\gamma}(\eta)
    \coloneqq
    \int q_\eta^{1+\gamma}(x)\,\mathrm{d}\mu(x) .
\]
The function $q_\eta$ is the density of the scaled measure
$\alpha\mathbb{P}_\theta$ with respect to the dominating measure $\mu$.
For later use, define
\[
    h_n(\eta)
    \coloneqq
    \frac{
        n^{-1}\sum_{i=1}^{n} q_\eta^\gamma(y_i)
    }{
        Z_{1+\gamma}(\eta)
    },
    \qquad
    h(\eta)
    \coloneqq
    \frac{
        \mathbb{E}_{Y\sim\mathbb{P}}
        [
            q_\eta^\gamma(Y)
        ]
    }{
        Z_{1+\gamma}(\eta)
    }.
\]
With this notation, the empirical and population H\"{o}lder risks are expressed as 
\[
    H_n^{(\gamma)}(\eta)
    =
    \phi\left(h_n(\eta)\right)Z_{1+\gamma}(\eta),
    \qquad
    H^{(\gamma)}(\eta)
    =
    \phi\left(h(\eta)\right)Z_{1+\gamma}(\eta).
\]

\section{Useful lemmas}
\label{app:useful_lemma}

In this section, we collect technical lemmas used in the proofs of the main
theoretical results.
For sample-dependent events, $\Pr$ denotes probability under the joint sampling
law of the observations; for i.i.d.\ data from $\mathbb{P}$, this law is
$\mathbb{P}^{\otimes n}$.

\begin{lem}[Lipschitz continuity of $\phi$]
\label{lem:lipschitz_phi}
Suppose that \Cref{assump:basic_regular} holds.
Then $\phi$ is Lipschitz continuous on $(0,\infty)$ with Lipschitz constant
$L>0$.
That is, for all $u,v\in(0,\infty)$,
\[
    |\phi(u)-\phi(v)|
    \le
    L|u-v|.
\]
\end{lem}

\begin{proof}
Without loss of generality, assume $u>v$.
By the mean value theorem, there exists $c\in(v,u)$ such that
\[
    \phi(u)-\phi(v)
    =
    \phi'(c)(u-v).
\]
By \Cref{assump:basic_regular}, $|\phi'(c)|\le L$.
Therefore,
\[
    |\phi(u)-\phi(v)|
    =
    |\phi'(c)||u-v|
    \le
    L|u-v|.
\]
\end{proof}

\begin{lem}[McDiarmid's inequality]
\label{lem:mcdiarmids}
Let $g$ be a function of $n$ variables
$\mathbf{y}_n=(y_1,\ldots,y_n)$, and define
\[
    \delta_i g(\mathbf{y}_n)
    \coloneqq
    g(y_1,\ldots,y_i,\ldots,y_n)
    -
    g(y_1,\ldots,y_i',\ldots,y_n),
\]
for any $y_i'\in\mathcal{X}$.
Let
\[
    \|\delta_i g\|_\infty
    \coloneqq
    \sup_{\mathbf{y}_n\in\mathcal{X}^n,\;y_i'\in\mathcal{X}}
    |\delta_i g(\mathbf{y}_n)|.
\]
If $Y_1,\ldots,Y_n$ are independent random variables, then
\[
    \Pr
    \left(
        g(Y_1,\ldots,Y_n)
        -
        \mathbb{E}[g(Y_1,\ldots,Y_n)]
        \ge
        \epsilon
    \right)
    \le
    \exp
    \left(
        -
        \frac{
            2\epsilon^2
        }{
            \sum_{i=1}^{n}\|\delta_i g\|_\infty^2
        }
    \right).
\]
\end{lem}

\begin{lem}[Pointwise convergence]
\label{lem:pointwise_conv_holder}
Suppose that \Cref{assump:basic_regular} holds.
Then, for each fixed $\eta=(\theta,\alpha)\in\Theta\times(0,1]$,
\[
    H_n^{(\gamma)}(\eta)
    -
    H^{(\gamma)}(\eta)
    \overset{a.s.}{\longrightarrow}
    0.
\]
\end{lem}

\begin{proof}
Fix $\eta=(\theta,\alpha)\in\Theta\times(0,1]$ and define
\[
    A_n(\theta)
    \coloneqq
    \frac{1}{n}\sum_{i=1}^{n}p_\theta^\gamma(y_i),
    \qquad
    A(\theta)
    \coloneqq
    \mathbb{E}_{Y\sim\mathbb{P}}
    [
        p_\theta^\gamma(Y)
    ],
\]
and
\[
    C(\theta)
    \coloneqq
    \mathbb{E}_{Y\sim\mathbb{P}_\theta}
    [
        p_\theta^\gamma(Y)
    ].
\]
Then
\[
    H_n^{(\gamma)}(\theta,\alpha)
    =
    \phi\left(
        \alpha^{-1}
        \frac{A_n(\theta)}{C(\theta)}
    \right)
    \alpha^{1+\gamma}C(\theta),
\]
and
\[
    H^{(\gamma)}(\theta,\alpha)
    =
    \phi\left(
        \alpha^{-1}
        \frac{A(\theta)}{C(\theta)}
    \right)
    \alpha^{1+\gamma}C(\theta).
\]
By \Cref{assump:basic_regular}, $p_\theta$ is uniformly bounded.
Hence
\[
    \mathbb{E}_{Y\sim\mathbb{P}}
    \left[
        |p_\theta^\gamma(Y)|
    \right]
    <
    \infty.
\]
The strong law of large numbers gives
\[
    A_n(\theta)
    \overset{a.s.}{\longrightarrow}
    A(\theta).
\]
Since $C(\theta)>0$ and $\phi$ is continuous on $(0,\infty)$,
the continuous mapping theorem yields
\[
    H_n^{(\gamma)}(\theta,\alpha)
    \overset{a.s.}{\longrightarrow}
    H^{(\gamma)}(\theta,\alpha).
\]
\end{proof}

\begin{lem}[Finiteness of the normalising term and its derivatives]
\label{lem:bounded_Z_derivative}
Suppose that
\Cref{assump:basic_regular,assump:existence_minimizer_unique,assump:differentiable_q}
hold.
Then, for $r=0,1,2,3$,
\[
    \sup_{\eta\in\bar U}
    \left\|
        \nabla_\eta^r Z_{1+\gamma}(\eta)
    \right\|
    <
    \infty.
\]
Moreover, there exists a constant $c_Z>0$ such that
\[
    \inf_{\eta\in\bar U}
    Z_{1+\gamma}(\eta)
    \ge
    c_Z.
\]
\end{lem}

\begin{proof}
Recall that
\[
    Z_{1+\gamma}(\eta)
    =
    \int q_\eta^{1+\gamma}(x)\,\mathrm{d}\mu(x)
    =
    \alpha^{1+\gamma}
    \int p_\theta^{1+\gamma}(x)\,\mathrm{d}\mu(x).
\]
By \Cref{assump:existence_minimizer_unique}, the local set satisfies
\[
    \bar U
    \subset
    \operatorname{int}(\Theta)\times[\underline{\alpha},1]
\]
for some $\underline{\alpha}>0$.
Thus $\alpha$ is bounded away from zero on $\bar U$.

By \Cref{assump:basic_regular}, $p_\theta$ is positive and uniformly bounded.
Since $\alpha\in(0,1]$, $q_\eta$ is also uniformly bounded on $\bar U$.
Moreover,
\[
    \int q_\eta(x)\,\mathrm{d}\mu(x)
    =
    \alpha
    \le
    1.
\]
Therefore,
\[
    Z_{1+\gamma}(\eta)
    =
    \int q_\eta^{1+\gamma}(x)\,\mathrm{d}\mu(x)
    \le
    \left(
        \sup_{\eta\in\bar U,\,x\in\mathcal X}
        q_\eta(x)
    \right)^\gamma
    \int q_\eta(x)\,\mathrm{d}\mu(x)
    <
    \infty.
\]

For the lower bound, observe that $Z_{1+\gamma}(\eta)>0$ for every
$\eta\in\bar U$.
By \Cref{assump:differentiable_q}, $Z_{1+\gamma}(\eta)$ is continuous on the
compact set $\bar U$.
Hence it attains a positive minimum on $\bar U$; denote this minimum by
$c_Z>0$.

For $r=1,2,3$, \Cref{assump:differentiable_q} justifies differentiation under
the integral sign:
\[
    \nabla_\eta^r Z_{1+\gamma}(\eta)
    =
    \int
    \nabla_\eta^r q_\eta^{1+\gamma}(x)
    \,\mathrm{d}\mu(x).
\]
The corresponding envelope condition implies that these derivatives are finite
and uniformly bounded over $\eta\in\bar U$.
\end{proof}

\begin{lem}[Finiteness of $\phi(h_n)$ and its derivatives]
\label{lem:bounded_phi_derivative}
Suppose that
\Cref{assump:basic_regular,assump:existence_minimizer_unique,assump:differentiable_q}
hold.
Then, almost surely for all sufficiently large $n$,
\[
    \sup_{\eta\in\bar U}
    \left|
        \phi^{(r)}(h_n(\eta))
    \right|
    <
    \infty
    \qquad
    \text{for } r=0,1,2,3.
\]
The same conclusion holds for $\phi^{(r)}(h(\eta))$.
\end{lem}

\begin{proof}
By \Cref{lem:bounded_Z_derivative}, $Z_{1+\gamma}(\eta)$ is bounded above and
bounded away from zero on $\bar U$.
Let
\[
    A_n(\eta)
    \coloneqq
    \frac{1}{n}
    \sum_{i=1}^{n}
    q_\eta^\gamma(y_i),
    \qquad
    A(\eta)
    \coloneqq
    \mathbb{E}_{\mathbb{P}}
    [
        q_\eta^\gamma(Y)
    ].
\]
By \Cref{assump:differentiable_q}, the class
$\{q_\eta^\gamma:\eta\in\bar U\}$ has an integrable envelope and is locally
Lipschitz in $\eta$.
Therefore,
\[
    \sup_{\eta\in\bar U}
    |A_n(\eta)-A(\eta)|
    \overset{a.s.}{\longrightarrow}
    0.
\]
Moreover, $A(\eta)>0$ for every $\eta\in\bar U$ and $A(\eta)$ is continuous on
the compact set $\bar U$.
Thus
\[
    \inf_{\eta\in\bar U}A(\eta)>0.
\]
Consequently, almost surely for all sufficiently large $n$,
\[
    h_n(\eta)
    =
    \frac{A_n(\eta)}{Z_{1+\gamma}(\eta)}
\]
takes values in a compact subinterval of $(0,\infty)$ uniformly over
$\eta\in\bar U$.
The same is true for $h(\eta)$.
Since $\phi$ is of class $C^3$ on $(0,\infty)$, the functions
$\phi,\phi',\phi''$, and $\phi'''$ are bounded on this compact subinterval.
This proves the claim.
\end{proof}

\begin{lem}[Uniform boundedness of the first derivative]
\label{lem:bounded_derivative_first}
Suppose that
\Cref{assump:basic_regular,assump:existence_minimizer_unique,assump:differentiable_q}
hold.
Then the gradient of the empirical H\"{o}lder risk is uniformly bounded on
$\bar U$ almost surely for all sufficiently large $n$:
\[
    \limsup_{n\to\infty}
    \sup_{\eta\in\bar U}
    \left\|
        \nabla_\eta H_n^{(\gamma)}(\eta)
    \right\|_2
    <
    \infty
    \quad
    \text{a.s.}
\]
\end{lem}

\begin{proof}
By the chain rule,
\[
    \nabla_\eta H_n^{(\gamma)}(\eta)
    =
    \phi'(h_n(\eta))
    Z_{1+\gamma}(\eta)
    \nabla_\eta h_n(\eta)
    +
    \phi(h_n(\eta))
    \nabla_\eta Z_{1+\gamma}(\eta).
\]
By \Cref{lem:bounded_Z_derivative,lem:bounded_phi_derivative}, the factors
$\phi(h_n(\eta))$, $\phi'(h_n(\eta))$, $Z_{1+\gamma}(\eta)$, and
$\nabla_\eta Z_{1+\gamma}(\eta)$ are uniformly bounded on $\bar U$ almost surely
for all sufficiently large $n$.

It remains to bound $\nabla_\eta h_n(\eta)$.
Let
\[
    A_n(\eta)
    \coloneqq
    \frac{1}{n}
    \sum_{i=1}^{n}
    q_\eta^\gamma(y_i).
\]
Then
\[
    h_n(\eta)
    =
    \frac{A_n(\eta)}{Z_{1+\gamma}(\eta)}
\]
and
\[
    \nabla_\eta h_n(\eta)
    =
    \frac{
        \nabla_\eta A_n(\eta)Z_{1+\gamma}(\eta)
        -
        A_n(\eta)\nabla_\eta Z_{1+\gamma}(\eta)
    }{
        Z_{1+\gamma}(\eta)^2
    }.
\]
The denominator is uniformly bounded away from zero by
\Cref{lem:bounded_Z_derivative}.
The terms $A_n(\eta)$ and $\nabla_\eta A_n(\eta)$ are uniformly bounded on
$\bar U$ almost surely for all sufficiently large $n$ by the envelope condition
in \Cref{assump:differentiable_q} and the strong law of large numbers.
Therefore $\nabla_\eta h_n(\eta)$ is uniformly bounded on $\bar U$ almost surely
for all sufficiently large $n$.
Combining these bounds yields the result.
\end{proof}

\begin{lem}[Uniform boundedness of the second derivative]
\label{lem:bounded_derivative_second}
Suppose that
\Cref{assump:basic_regular,assump:existence_minimizer_unique,assump:differentiable_q}
hold.
Then the Hessian of the empirical H\"{o}lder risk is uniformly bounded on
$\bar U$ almost surely for all sufficiently large $n$:
\[
    \limsup_{n\to\infty}
    \sup_{\eta\in\bar U}
    \left\|
        \nabla_\eta^2 H_n^{(\gamma)}(\eta)
    \right\|_2
    <
    \infty
    \quad
    \text{a.s.}
\]
\end{lem}

\begin{proof}
The second derivative of
\[
    H_n^{(\gamma)}(\eta)
    =
    \phi(h_n(\eta))Z_{1+\gamma}(\eta)
\]
is a finite sum of products involving
\[
    \phi(h_n),\quad
    \phi'(h_n),\quad
    \phi''(h_n),
\]
the derivatives of $h_n$ up to order two, and the derivatives of
$Z_{1+\gamma}$ up to order two.

By \Cref{lem:bounded_Z_derivative,lem:bounded_phi_derivative}, the terms
involving $Z_{1+\gamma}$ and $\phi$ are uniformly bounded on $\bar U$ almost
surely for all sufficiently large $n$.
The derivatives of $h_n=A_n/Z_{1+\gamma}$ up to order two are obtained by
repeated applications of the quotient rule.
They are finite sums of products involving derivatives of $A_n$ and
$Z_{1+\gamma}$ up to order two, divided by powers of $Z_{1+\gamma}$.
The denominator is uniformly bounded away from zero by
\Cref{lem:bounded_Z_derivative}, and the derivatives of $A_n$ are uniformly
bounded on $\bar U$ by the envelope condition in
\Cref{assump:differentiable_q} and the strong law of large numbers.
Thus $\nabla_\eta^2 h_n(\eta)$ is uniformly bounded on $\bar U$ almost surely
for all sufficiently large $n$.

Combining these bounds proves the claim.
\end{proof}

\begin{lem}[Uniform boundedness of the third derivative]
\label{lem:bounded_derivative_third}
Suppose that
\Cref{assump:basic_regular,assump:existence_minimizer_unique,assump:differentiable_q}
hold.
Then the third derivative of the empirical H\"{o}lder risk is uniformly bounded
on $\bar U$ almost surely for all sufficiently large $n$:
\[
    \limsup_{n\to\infty}
    \sup_{\eta\in\bar U}
    \left\|
        \nabla_\eta^3 H_n^{(\gamma)}(\eta)
    \right\|_2
    <
    \infty
    \quad
    \text{a.s.}
\]
\end{lem}

\begin{proof}
The third derivative of
$H_n^{(\gamma)}(\eta)=\phi(h_n(\eta))Z_{1+\gamma}(\eta)$ is a finite sum of
products involving
\[
    \phi(h_n),\quad
    \phi'(h_n),\quad
    \phi''(h_n),\quad
    \phi'''(h_n),
\]
the derivatives of $h_n$ up to order three, and the derivatives of
$Z_{1+\gamma}$ up to order three.

The terms involving $\phi$ are uniformly bounded by
\Cref{lem:bounded_phi_derivative}, and the terms involving $Z_{1+\gamma}$ are
uniformly bounded by \Cref{lem:bounded_Z_derivative}.
The derivatives of $h_n=A_n/Z_{1+\gamma}$ up to order three are finite sums of
products of derivatives of $A_n$ and $Z_{1+\gamma}$ up to order three, divided
by powers of $Z_{1+\gamma}$.
Again, $Z_{1+\gamma}$ is uniformly bounded away from zero, and the derivatives
of $A_n$ are uniformly bounded on $\bar U$ almost surely for all sufficiently
large $n$ by the envelope condition and the strong law of large numbers.
Thus the third derivative is uniformly bounded on $\bar U$ almost surely for all
sufficiently large $n$.
\end{proof}

\begin{lem}[Uniform convergence on the local neighbourhood]
\label{lem:uniform_conv_holder}
Suppose that
\Cref{assump:basic_regular,assump:existence_minimizer_unique,assump:differentiable_q}
hold.
Then
\[
    \sup_{\eta\in\bar U}
    \left|
        H_n^{(\gamma)}(\eta)
        -
        H^{(\gamma)}(\eta)
    \right|
    \overset{a.s.}{\longrightarrow}
    0.
\]
\end{lem}

\begin{proof}
Pointwise almost sure convergence follows from
\Cref{lem:pointwise_conv_holder}.
It remains to show stochastic equicontinuity on $\bar U$.
By the mean value theorem, for any $\eta,\eta'\in\bar U$,
\[
    \left|
        H_n^{(\gamma)}(\eta)
        -
        H_n^{(\gamma)}(\eta')
    \right|
    \le
    \left(
        \sup_{\tilde\eta\in\bar U}
        \left\|
            \nabla_{\tilde\eta}H_n^{(\gamma)}(\tilde\eta)
        \right\|_2
    \right)
    \|\eta-\eta'\|_2.
\]
By \Cref{lem:bounded_derivative_first}, the Lipschitz constant on the right-hand
side is almost surely bounded for all sufficiently large $n$.
Since $\bar U$ is compact, the standard finite-net argument implies uniform
convergence on $\bar U$.
\end{proof}

\begin{lem}[Existence and strong consistency of empirical minimisers]
\label{lem:strong_consistency}
Suppose that
\Cref{assump:basic_regular,assump:existence_minimizer_unique,assump:differentiable_q}
hold, and write
\[
    \mathcal M_n
    \coloneqq
    \argmin_{\eta\in\Theta\times(0,1]}H_n^{(\gamma)}(\eta).
\]
Almost surely, $\mathcal M_n$ is non-empty and contained in $U$ for all
sufficiently large $n$, and
\[
    \sup_{\eta\in\mathcal M_n}\|\eta-\eta^*\|_2\longrightarrow0.
\]
In particular, every sequence of empirical minimisers converges almost
surely to $\eta^*$.
\end{lem}

\begin{proof}
Work on the probability-one event in
\eqref{eq:empirical_risk_separation}.
Choose $r>0$ such that the closed Euclidean ball
$\overline B_r(\eta^*)$ is contained in $U$.
The separation condition implies that, for all sufficiently large $n$,
\[
    \inf_{\substack{
        \eta\in\Theta\times(0,1]\\
        \|\eta-\eta^*\|_2\ge r
    }}
    H_n^{(\gamma)}(\eta)>H_n^{(\gamma)}(\eta^*).
\]
The empirical risk is continuous on $\overline B_r(\eta^*)$ by
\Cref{assump:differentiable_q}, so it attains its minimum on this compact
ball.  Since $\eta^*$ belongs to the ball, the displayed inequality shows
that this is a global minimum and that every global minimiser lies in
$B_r(\eta^*)\subset U$.

For any $\varepsilon>0$, the same separation condition gives, for all
sufficiently large $n$,
\[
    \inf_{\substack{
        \eta\in\Theta\times(0,1]\\
        \|\eta-\eta^*\|_2\ge\varepsilon
    }}
    H_n^{(\gamma)}(\eta)>H_n^{(\gamma)}(\eta^*).
\]
Every $\eta\in\mathcal M_n$ satisfies
$H_n^{(\gamma)}(\eta)\le H_n^{(\gamma)}(\eta^*)$, and hence must lie in
$B_\varepsilon(\eta^*)$ eventually.  This proves the assertion
simultaneously for all empirical minimisers.
\end{proof}

\begin{lem}[Almost sure convergence of $\phi$ and its derivatives]
\label{lem:as_conv_phi}
Suppose that
\Cref{assump:basic_regular,assump:existence_minimizer_unique,assump:differentiable_q}
hold.
Then, for each fixed $\eta\in U$ and for $r=0,1,2,3$,
\[
    \phi^{(r)}(h_n(\eta))
    \overset{a.s.}{\longrightarrow}
    \phi^{(r)}(h(\eta)),
\]
where $\phi^{(0)}=\phi$.
\end{lem}

\begin{proof}
For fixed $\eta\in U$, the strong law of large numbers gives
\[
    \frac{1}{n}
    \sum_{i=1}^{n}
    q_\eta^\gamma(y_i)
    \overset{a.s.}{\longrightarrow}
    \mathbb{E}_{\mathbb{P}}
    [
        q_\eta^\gamma(Y)
    ].
\]
Since $Z_{1+\gamma}(\eta)>0$ does not depend on $n$,
\[
    h_n(\eta)
    \overset{a.s.}{\longrightarrow}
    h(\eta).
\]
The functions $\phi,\phi',\phi''$, and $\phi'''$ are continuous on
$(0,\infty)$ by \Cref{assump:differentiable_q}.
The continuous mapping theorem yields the desired convergence.
\end{proof}

\begin{lem}[Almost sure convergence of the derivatives of $h_n$]
\label{lem:as_conv_grad_h}
Suppose that
\Cref{assump:basic_regular,assump:existence_minimizer_unique,assump:differentiable_q}
hold.
Then, for each fixed $\eta\in U$ and for $r=1,2,3$,
\[
    \nabla_\eta^r h_n(\eta)
    \overset{a.s.}{\longrightarrow}
    \nabla_\eta^r h(\eta).
\]
\end{lem}

\begin{proof}
Write
\[
    A_n(\eta)
    \coloneqq
    \frac{1}{n}\sum_{i=1}^{n}q_\eta^\gamma(y_i),
    \qquad
    A(\eta)
    \coloneqq
    \mathbb{E}_{\mathbb{P}}
    [
        q_\eta^\gamma(Y)
    ].
\]
For $r=0,1,2,3$, the envelope condition in
\Cref{assump:differentiable_q} and the strong law of large numbers imply
\[
    \nabla_\eta^r A_n(\eta)
    =
    \frac{1}{n}
    \sum_{i=1}^{n}
    \nabla_\eta^r q_\eta^\gamma(y_i)
    \overset{a.s.}{\longrightarrow}
    \mathbb{E}_{\mathbb{P}}
    [
        \nabla_\eta^r q_\eta^\gamma(Y)
    ]
    =
    \nabla_\eta^r A(\eta).
\]
Here the equality on the right uses differentiation under the expectation,
justified by the envelope condition.

Since
\[
    h_n(\eta)
    =
    \frac{A_n(\eta)}{Z_{1+\gamma}(\eta)}
\]
and $Z_{1+\gamma}(\eta)>0$, each derivative
$\nabla_\eta^r h_n(\eta)$ for $r=1,2,3$ is a continuous function of
$\nabla_\eta^k A_n(\eta)$ and $\nabla_\eta^k Z_{1+\gamma}(\eta)$ for
$k\le r$.
The corresponding population derivative has the same expression with $A_n$
replaced by $A$.
The continuous mapping theorem therefore yields the claimed convergence.
\end{proof}

\begin{prop}[Asymptotic convergence of the score derivatives]
\label{prop:as_conv_derivatives}
Suppose that
\Cref{assump:basic_regular,assump:existence_minimizer_unique,assump:differentiable_q}
hold.
The following convergences hold almost surely.
\begin{enumerate}
\item \textbf{Pointwise convergence at $\eta^*$.}
For $r=1,2,3$,
\[
    \nabla_\eta^r H_n^{(\gamma)}(\eta^*)
    \overset{a.s.}{\longrightarrow}
    \nabla_\eta^r H^{(\gamma)}(\eta^*).
\]
\item \textbf{Convergence of the Hessian at $\eta_n$.}
Let
\[
    \bm{H}_n(\eta)
    \coloneqq
    \nabla_\eta^2 H_n^{(\gamma)}(\eta),
    \qquad
    \bm{H}(\eta)
    \coloneqq
    \nabla_\eta^2 H^{(\gamma)}(\eta).
\]
Then
\[
    \bm{H}_n(\eta_n)
    \overset{a.s.}{\longrightarrow}
    \bm{H}(\eta^*).
\]
Moreover, if $\eta^*$ is an interior minimiser of $H^{(\gamma)}$, then
$\bm{H}(\eta^*)$ is symmetric and positive semidefinite.
\end{enumerate}
\end{prop}

\begin{proof}
For Part 1, observe that
$\nabla_\eta^r H_n^{(\gamma)}(\eta)$ is a finite sum of products involving
$\phi^{(k)}(h_n(\eta))$, derivatives of $h_n(\eta)$, and derivatives of
$Z_{1+\gamma}(\eta)$ up to order $r$.
By \Cref{lem:as_conv_phi,lem:as_conv_grad_h}, each empirical component converges
almost surely to its population counterpart at the fixed point $\eta^*$.
The continuous mapping theorem gives
\[
    \nabla_\eta^r H_n^{(\gamma)}(\eta^*)
    \overset{a.s.}{\longrightarrow}
    \nabla_\eta^r H^{(\gamma)}(\eta^*)
    \qquad
    (r=1,2,3).
\]

For Part 2, use the triangle inequality:
\[
    \left\|
        \bm{H}_n(\eta_n)
        -
        \bm{H}(\eta^*)
    \right\|_2
    \le
    \left\|
        \bm{H}_n(\eta_n)
        -
        \bm{H}_n(\eta^*)
    \right\|_2
    +
    \left\|
        \bm{H}_n(\eta^*)
        -
        \bm{H}(\eta^*)
    \right\|_2.
\]
The second term converges to zero almost surely by Part 1 with $r=2$.
For the first term, by the mean value theorem and the convexity of $U$,
\[
    \left\|
        \bm{H}_n(\eta_n)
        -
        \bm{H}_n(\eta^*)
    \right\|_2
    \le
    \left(
        \sup_{\eta\in\bar U}
        \left\|
            \nabla_\eta^3 H_n^{(\gamma)}(\eta)
        \right\|_2
    \right)
    \|\eta_n-\eta^*\|_2.
\]
The supremum is almost surely bounded for all sufficiently large $n$ by
\Cref{lem:bounded_derivative_third}, and
$\eta_n\to\eta^*$ almost surely by \Cref{lem:strong_consistency}.
Therefore the first term also converges to zero almost surely.

The symmetry of $\bm{H}(\eta^*)$ follows from the $C^2$ smoothness of
$H^{(\gamma)}$ in a neighbourhood of $\eta^*$.
If $\eta^*$ is an interior minimiser of $H^{(\gamma)}$, the second-order
necessary condition for a local minimum implies that $\bm{H}(\eta^*)$ is
positive semidefinite.
\end{proof}

\section{Proofs}
\label{sec:proofs}

\subsection{Proof of Theorem~\ref{thm:IF_holder}}
\label{proof:IF_holder}
\begin{proof}
We first derive the posterior influence function and then verify its boundedness.
Throughout this proof, for the contaminated empirical distribution
\[
    \mathbb{P}_{n,\varepsilon,z}
    =
    (1-\varepsilon)\mathbb{P}_n+\varepsilon\delta_z,
\]
we write
\[
    H_n^{(\gamma)}
    (
        \theta,\alpha
        \mid
        \mathbb{P}_{n,\varepsilon,z}
    )
\]
for the empirical H\"{o}lder risk obtained by replacing $\mathbb{P}_n$ with
$\mathbb{P}_{n,\varepsilon,z}$.

For compactness in the proof, recall the appendix notation
\[
    \eta=(\theta,\alpha),
    \qquad
    q_\eta(x)=\alpha p_\theta(x).
\]
Define
\[
    A_n(\eta)
    \coloneqq
    \frac{1}{n}
    \sum_{i=1}^n
    q_\eta^\gamma(y_i),
    \qquad
    Z_{1+\gamma}(\eta)
    \coloneqq
    \int q_\eta^{1+\gamma}(x)\,\mathrm{d}\mu(x),
\]
and
\[
    h_n(\eta)
    \coloneqq
    \frac{
        A_n(\eta)
    }{
        Z_{1+\gamma}(\eta)
    }.
\]
Under the contaminated empirical distribution,
\[
    A_{n,\varepsilon,z}(\eta)
    =
    (1-\varepsilon)A_n(\eta)
    +
    \varepsilon q_\eta^\gamma(z),
\]
and therefore
\[
    H_n^{(\gamma)}
    (
        \eta
        \mid
        \mathbb{P}_{n,\varepsilon,z}
    )
    =
    \phi
    \left(
        \frac{
            A_{n,\varepsilon,z}(\eta)
        }{
            Z_{1+\gamma}(\eta)
        }
    \right)
    Z_{1+\gamma}(\eta).
\]
Since $Z_{1+\gamma}(\eta)$ does not depend on $\varepsilon$, differentiating at
$\varepsilon=0$ gives
\[
    \left.
    \frac{\mathrm{d}}{\mathrm{d}\varepsilon}
    H_n^{(\gamma)}
    (
        \eta
        \mid
        \mathbb{P}_{n,\varepsilon,z}
    )
    \right|_{\varepsilon=0}
    =
    \phi'(h_n(\eta))
    \left\{
        q_\eta^\gamma(z)
        -
        A_n(\eta)
    \right\}.
\]
Equivalently, since $q_\eta(x)=\alpha p_\theta(x)$,
\[
    q_\eta^\gamma(z)-A_n(\eta)
    =
    \alpha^\gamma
    \left\{
        p_\theta^\gamma(z)
        -
        \frac{1}{n}
        \sum_{i=1}^{n}
        p_\theta^\gamma(y_i)
    \right\},
\]
and
\[
    h_n(\eta)
    =
    \alpha^{-1}
    \frac{
        \frac{1}{n}\sum_{i=1}^{n}p_\theta^\gamma(y_i)
    }{
        \mathbb{E}_{Y\sim\mathbb{P}_\theta}
        [
            p_\theta^\gamma(Y)
        ]
    }.
\]
Thus the influence function of the H\"{o}lder risk is
\[
    \mathrm{IF}
    (
        z,\theta,\alpha,\mathbb{P}_n
    )
    =
    \alpha^\gamma
    \phi'
    \left(
        \alpha^{-1}
        \frac{
            \frac{1}{n}\sum_{i=1}^{n}p_\theta^\gamma(y_i)
        }{
            \mathbb{E}_{Y\sim\mathbb{P}_\theta}
            [
                p_\theta^\gamma(Y)
            ]
        }
    \right)
    \left\{
        p_\theta^\gamma(z)
        -
        \frac{1}{n}
        \sum_{i=1}^{n}
        p_\theta^\gamma(y_i)
    \right\}.
\]

We now differentiate the posterior density.  Let
\[
    Z_n(\varepsilon)
    \coloneqq
    \int
    \exp
    \left\{
        -\beta n
        H_n^{(\gamma)}
        (
            \eta
            \mid
            \mathbb{P}_{n,\varepsilon,z}
        )
    \right\}
    \pi(\eta)
    \,\mathrm{d}\eta
\]
be the normalising constant.  Then
\[
    \pi_n^H
    (
        \eta
        \mid
        \mathbb{P}_{n,\varepsilon,z}
    )
    =
    \frac{
        \exp
        \left\{
            -\beta n
            H_n^{(\gamma)}
            (
                \eta
                \mid
                \mathbb{P}_{n,\varepsilon,z}
            )
        \right\}
        \pi(\eta)
    }{
        Z_n(\varepsilon)
    }.
\]
Taking the derivative of the log-posterior gives
\[
    \left.
    \frac{\mathrm{d}}{\mathrm{d}\varepsilon}
    \log
    \pi_n^H
    (
        \eta
        \mid
        \mathbb{P}_{n,\varepsilon,z}
    )
    \right|_{\varepsilon=0}
    =
    -\beta n
    \mathrm{IF}
    (
        z,\eta,\mathbb{P}_n
    )
    -
    \left.
    \frac{\mathrm{d}}{\mathrm{d}\varepsilon}
    \log Z_n(\varepsilon)
    \right|_{\varepsilon=0}.
\]
By differentiating under the integral sign, justified by the boundedness of
$\phi'$ and $q_\eta^\gamma$ under \Cref{assump:basic_regular}, we obtain
\[
    \left.
    \frac{\mathrm{d}}{\mathrm{d}\varepsilon}
    \log Z_n(\varepsilon)
    \right|_{\varepsilon=0}
    =
    -\beta n
    \mathbb{E}_{\bar\eta\sim\pi_n^H(\cdot\mid\mathbb{P}_n)}
    \left[
        \mathrm{IF}
        (
            z,\bar\eta,\mathbb{P}_n
        )
    \right].
\]
Therefore,
\[
    \left.
    \frac{\mathrm{d}}{\mathrm{d}\varepsilon}
    \log
    \pi_n^H
    (
        \eta
        \mid
        \mathbb{P}_{n,\varepsilon,z}
    )
    \right|_{\varepsilon=0}
    =
    -\beta n
    \left\{
        \mathrm{IF}
        (
            z,\eta,\mathbb{P}_n
        )
        -
        \mathbb{E}_{\bar\eta\sim\pi_n^H(\cdot\mid\mathbb{P}_n)}
        \left[
            \mathrm{IF}
            (
                z,\bar\eta,\mathbb{P}_n
            )
        \right]
    \right\}.
\]
Multiplying by
$\pi_n^H(\eta\mid\mathbb{P}_n)$ yields
\[
    \mathrm{PIF}
    (
        z,\eta,\mathbb{P}_n
    )
    =
    -\beta n
    \pi_n^H
    (
        \eta
        \mid
        \mathbb{P}_n
    )
    \left[
        \mathrm{IF}
        (
            z,\eta,\mathbb{P}_n
        )
        -
        \mathbb{E}_{\bar\eta\sim\pi_n^H(\cdot\mid\mathbb{P}_n)}
        \left[
            \mathrm{IF}
            (
                z,\bar\eta,\mathbb{P}_n
            )
        \right]
    \right],
\]
which is the claimed formula.

It remains to verify boundedness.  By \Cref{assump:basic_regular},
\[
    |\phi'(u)|\le L
    \qquad
    \text{for all }u\in(0,\infty),
\]
and there exists $M<\infty$ such that
\[
    \sup_{\theta\in\Theta}
    \sup_{x\in\mathcal X}
    p_\theta(x)
    \le
    M.
\]
Since $\alpha\in(0,1]$,
\[
    \left|
        \mathrm{IF}
        (
            z,\theta,\alpha,\mathbb{P}_n
        )
    \right|
    \le
    L\alpha^\gamma
    \left(
        p_\theta^\gamma(z)
        +
        \frac{1}{n}
        \sum_{i=1}^{n}
        p_\theta^\gamma(y_i)
    \right)
    \le
    2LM^\gamma.
\]
Hence the risk influence function is uniformly bounded in
$(z,\theta,\alpha)$.

Assume, as part of the regularity conditions for the PIF result, that the
posterior density is bounded for the fixed dataset $\mathbb{P}_n$; for example,
this follows if the prior density is bounded and the posterior normalising
constant is positive.  Then
\[
    \sup_{\eta}
    \pi_n^H
    (
        \eta
        \mid
        \mathbb{P}_n
    )
    <
    \infty.
\]
Using the preceding bound on $\mathrm{IF}$,
\[
    \sup_{z,\eta}
    \left|
        \mathrm{PIF}
        (
            z,\eta,\mathbb{P}_n
        )
    \right|
    \le
    4\beta n LM^\gamma
    \sup_{\eta}
    \pi_n^H
    (
        \eta
        \mid
        \mathbb{P}_n
    )
    <
    \infty.
\]
Therefore, the H\"{o}lder posterior is globally bias-robust.

Finally, if the model density vanishes in the tails, namely
$p_\theta(z)\to0$ as $\|z\|\to\infty$, then
$q_\eta^\gamma(z)=\alpha^\gamma p_\theta^\gamma(z)\to0$.
Consequently,
\[
    \lim_{\|z\|\to\infty}
    \mathrm{IF}
    (
        z,\eta,\mathbb{P}_n
    )
    =
    -
    \phi'(h_n(\eta))
    \frac{1}{n}
    \sum_{i=1}^{n}
    q_\eta^\gamma(y_i),
\]
or, equivalently,
\[
    \lim_{\|z\|\to\infty}
    \mathrm{IF}
    (
        z,\theta,\alpha,\mathbb{P}_n
    )
    =
    -
    \alpha^\gamma
    \phi'
    \left(
        \alpha^{-1}
        \frac{
            \frac{1}{n}\sum_{i=1}^{n}p_\theta^\gamma(y_i)
        }{
            \mathbb{E}_{Y\sim\mathbb{P}_\theta}
            [
                p_\theta^\gamma(Y)
            ]
        }
    \right)
    \frac{1}{n}
    \sum_{i=1}^{n}
    p_\theta^\gamma(y_i).
\]
This establishes the stated tail behaviour.
\end{proof}

\subsection{Proof of Theorem~\ref{thm:consistency_holder}}
\label{proof:consistency_holder}

\begin{proof}
Write $\eta=(\theta,\alpha)$ and $\eta^*=(\theta^*,\alpha^*)$.
Define
\[
    A_n(\theta)=\frac1n\sum_{i=1}^n p_\theta^\gamma(Y_i),
    \qquad
    A(\theta)=\mathbb E_{Y\sim\mathbb P}[p_\theta^\gamma(Y)],
    \qquad
    C(\theta)=\int p_\theta^{1+\gamma}\,\mathrm d\mu.
\]
By \Cref{lem:lipschitz_phi}, for every dataset and every
$\eta\in\Theta\times(0,1]$,
\begin{align}
    |H_n^{(\gamma)}(\theta,\alpha)-H^{(\gamma)}(\theta,\alpha)|
    &\le
      L\left|\frac{A_n(\theta)-A(\theta)}{\alpha C(\theta)}\right|
       \alpha^{1+\gamma}C(\theta) \notag\\
    &=L\alpha^\gamma|A_n(\theta)-A(\theta)| \notag\\
    &\le L|A_n(\theta)-A(\theta)|.
    \label{eq:thm2_lipschitz_uniform}
\end{align}
Consequently, if
\[
    Z_n\coloneqq\sup_{\theta\in\Theta}|A_n(\theta)-A(\theta)|,
    \qquad
    \Delta_n\coloneqq
      \sup_{\eta\in\Theta\times(0,1]}
      |H_n^{(\gamma)}(\eta)-H^{(\gamma)}(\eta)|,
\]
then $\Delta_n\le LZ_n$.

We first bound $\mathbb E Z_n$.
Let $\mathcal F_\gamma=\{p_\theta^\gamma:\theta\in\Theta\}$,
$B_\gamma=M^\gamma$, $D=\operatorname{diam}(\Theta)$, and
$s_\gamma=\mathbb E_{Y\sim\mathbb P}[G_\gamma(Y)]$.
The class has constant envelope $B_\gamma$.
It is pointwise measurable: a countable dense subset of $\Theta$ yields
pointwise approximations, by the Lipschitz condition in
\Cref{assump:prior_mass}.
Thus its suprema below are measurable.

For measurable functions $\ell\le u$, the bracket $[\ell,u]$ consists of
all functions $f$ satisfying $\ell\le f\le u$ pointwise.
Let $N_{[\,]}(t,\mathcal F_\gamma,L^2(\mathbb P))$ denote the minimum
number of brackets of $L^2(\mathbb P)$ width at most $t$ required to cover
$\mathcal F_\gamma$.

For $r>0$, cover $\Theta$ by Euclidean balls of radius $r$ with centres
$\theta_1,\ldots,\theta_N\in\Theta$.
Each such ball gives the pointwise bracket $[\ell_j,u_j]$, where
\[
    \ell_j=(p_{\theta_j}^\gamma-rG_\gamma)_+,
    \qquad
    u_j=\min\{p_{\theta_j}^\gamma+rG_\gamma,B_\gamma\}.
\]
Its width satisfies $0\le u_j-\ell_j\le\min\{2rG_\gamma,B_\gamma\}$.
Consequently,
\[
    \|u_j-\ell_j\|_{L^2(\mathbb P)}^2
    \le 2rB_\gamma s_\gamma.
\]
For $s_\gamma>0$, choose $r=t^2/(2B_\gamma s_\gamma)$.
The Euclidean covering bound $N\le(1+2D/r)^d$ then gives
\[
    N_{[\,]}(t,\mathcal F_\gamma,L^2(\mathbb P))
    \le\left(1+\frac{4DB_\gamma s_\gamma}{t^2}\right)^d.
\]
If $s_\gamma=0$, a single bracket of $L^2(\mathbb P)$ width zero
covers the class, and the same bound holds.
It follows that
\[
\begin{aligned}
    J_\gamma
    &\coloneqq
      \int_0^1
      \sqrt{1+\log N_{[\,]}(uB_\gamma,\mathcal F_\gamma,L^2(\mathbb P))}
      \,\mathrm du \\
    &\le
      \int_0^1
      \sqrt{1+d\log\left(1+\frac{4Ds_\gamma}{u^2B_\gamma}\right)}
      \,\mathrm du
      <\infty.
\end{aligned}
\]
The standard bracketing maximal inequality
\citep[eq.~(1.2)]{vanDerVaartWellner2011} now implies
\begin{equation}
    \mathbb E Z_n
    \le \frac{K_\gamma}{\sqrt n}
    \quad\text{for all }n\ge1,
    \label{eq:thm2_expected_uniform_error}
\end{equation}
where one may take $K_\gamma=c_0B_\gamma J_\gamma$ for a universal
constant $c_0<\infty$.

If only the $i$-th observation is replaced, then
\[
    |Z_n(Y_1,\ldots,Y_i,\ldots,Y_n)
       -Z_n(Y_1,\ldots,Y_i',\ldots,Y_n)|
    \le \frac{B_\gamma}{n}.
\]
McDiarmid's inequality and
\eqref{eq:thm2_expected_uniform_error} therefore imply that, with
probability at least $1-\delta$,
\begin{equation}
    Z_n\le
    \frac{K_\gamma+B_\gamma\sqrt{\log(1/\delta)/2}}{\sqrt n}.
    \label{eq:thm2_uniform_concentration}
\end{equation}
This is a single event on which the risk approximation in
\eqref{eq:thm2_lipschitz_uniform} holds for every $\eta$.

We next use the Gibbs variational identity on each dataset.
Let $B_n(c_1)$ be the population sublevel set in
\Cref{assump:prior_mass}, and define the restricted prior density
\[
    \rho_n(\eta)
    =\frac{\mathbf1_{B_n(c_1)}(\eta)\pi(\eta)}{\Pi(B_n(c_1))},
    \qquad
    \Pi(B_n(c_1))=\int_{B_n(c_1)}\pi(\eta)\,\mathrm d\eta.
\]
The prior mass condition gives
\begin{equation}
    \mathbb E_{\rho_n}H^{(\gamma)}-H^{(\gamma)}(\eta^*)
       \le\frac{c_1}{\sqrt n},
    \qquad
    \operatorname{KL}(\rho_n\|\pi)
       =-\log\Pi(B_n(c_1))\le c_2\sqrt n.
    \label{eq:thm2_comparison_prior}
\end{equation}
For every $\beta>0$, the Gibbs variational identity yields
\[
    \mathbb E_{\pi_n^H}H_n^{(\gamma)}
       +\frac{\operatorname{KL}(\pi_n^H\|\pi)}{\beta n}
    \le
    \mathbb E_{\rho_n}H_n^{(\gamma)}
       +\frac{\operatorname{KL}(\rho_n\|\pi)}{\beta n}.
\]
All expectations here are finite: Lipschitz continuity gives
$|\phi(z)|\le |\phi(1)|+L(z+1)$ for $z>0$, and hence both
$|H_n^{(\gamma)}(\eta)|$ and $|H^{(\gamma)}(\eta)|$ are bounded by
$(|\phi(1)|+2L)B_\gamma$ on $\Theta\times(0,1]$.
Using $\operatorname{KL}(\pi_n^H\|\pi)\ge0$, we obtain the
dataset-wise inequality
\begin{align}
    \mathbb E_{\pi_n^H}H^{(\gamma)}-H^{(\gamma)}(\eta^*)
    &\le
       \mathbb E_{\pi_n^H}H_n^{(\gamma)}
         -H^{(\gamma)}(\eta^*)+\Delta_n \notag\\
    &\le
       \mathbb E_{\rho_n}H_n^{(\gamma)}
         -H^{(\gamma)}(\eta^*)
         +\frac{\operatorname{KL}(\rho_n\|\pi)}{\beta n}
         +\Delta_n \notag\\
    &\le
       \mathbb E_{\rho_n}H^{(\gamma)}
         -H^{(\gamma)}(\eta^*)
         +\frac{\operatorname{KL}(\rho_n\|\pi)}{\beta n}
         +2\Delta_n \notag\\
    &\le
       \frac{c_1+c_2/\beta}{\sqrt n}+2LZ_n.
    \label{eq:thm2_variational_bound}
\end{align}
Since $\log(2/\delta)\ge\log2$ and
$\log(1/\delta)\le\log(2/\delta)$,
\[
    K_\gamma+B_\gamma\sqrt{\log(1/\delta)/2}
    \le C\sqrt{\log(2/\delta)},
    \qquad
    C=\frac{K_\gamma}{\sqrt{\log2}}+\frac{B_\gamma}{\sqrt2}.
\]
Combining this inequality with
\eqref{eq:thm2_uniform_concentration} and
\eqref{eq:thm2_variational_bound} proves the stated bound.
\end{proof}

\subsection{Proof of Theorem~\ref{thm:error_control_holder}}
\label{sec:proof_error_control_holder}

\begin{proof}
Let
\[
    \alpha_0
    \coloneqq
    1-\epsilon_0,
    \qquad
    \eta_0
    \coloneqq
    (\theta_0,\alpha_0).
\]
Recall that the data-generating distribution is
\[
    \mathbb{P}
    =
    \alpha_0\mathbb{P}_{\theta_0}
    +
    \epsilon_0\mathbb{Q}_0 .
\]
For any $(\theta,\alpha)\in\Theta\times(0,1]$, define
\[
    C(\theta)
    \coloneqq
    \mathbb{E}_{Y\sim\mathbb{P}_\theta}
    \left[
        p_\theta^\gamma(Y)
    \right],
\]
\[
    A_0(\theta)
    \coloneqq
    \mathbb{E}_{Y\sim\mathbb{P}_{\theta_0}}
    \left[
        p_\theta^\gamma(Y)
    \right],
    \qquad
    A_Q(\theta)
    \coloneqq
    \mathbb{E}_{Y\sim\mathbb{Q}_0}
    \left[
        p_\theta^\gamma(Y)
    \right].
\]
By definition,
\[
    A_Q(\theta)
    =
    \rho_\gamma(\theta).
\]
Since
\[
    \mathbb{E}_{Y\sim\mathbb{P}}
    \left[
        p_\theta^\gamma(Y)
    \right]
    =
    \alpha_0 A_0(\theta)
    +
    \epsilon_0 A_Q(\theta),
\]
the population H\"{o}lder risk can be written as
\[
    H^{(\gamma)}(\theta,\alpha)
    =
    \phi
    \left(
        a(\theta,\alpha)
        +
        b(\theta,\alpha)
    \right)
    \alpha^{1+\gamma}C(\theta),
\]
where
\[
    a(\theta,\alpha)
    \coloneqq
    \frac{\alpha_0}{\alpha}
    \frac{
        A_0(\theta)
    }{
        C(\theta)
    },
    \qquad
    b(\theta,\alpha)
    \coloneqq
    \frac{\epsilon_0}{\alpha}
    \frac{
        A_Q(\theta)
    }{
        C(\theta)
    }.
\]

By the mean value theorem, for each $(\theta,\alpha)$ there exists a point
$c(\theta,\alpha)$ between
$a(\theta,\alpha)$ and
$a(\theta,\alpha)+b(\theta,\alpha)$ such that
\[
    \phi(a+b)
    =
    \phi(a)
    +
    \phi'(c)b.
\]
Therefore,
\[
\begin{aligned}
    H^{(\gamma)}(\theta,\alpha)
    &=
    \phi
    \left(
        a(\theta,\alpha)
    \right)
    \alpha^{1+\gamma}C(\theta)
    +
    \phi'
    \left(
        c(\theta,\alpha)
    \right)
    \epsilon_0
    \alpha^\gamma
    A_Q(\theta).
\end{aligned}
\]
Define the H\"{o}lder score of the scaled target measure
$\alpha_0\mathbb{P}_{\theta_0}$ against the scaled model
$\alpha\mathbb{P}_{\theta}$ by
\[
    S_\phi^{(\gamma)}
    \left(
        \alpha_0\mathbb{P}_{\theta_0},
        \alpha\mathbb{P}_{\theta}
    \right)
    \coloneqq
    \phi
    \left(
        \frac{\alpha_0}{\alpha}
        \frac{
            A_0(\theta)
        }{
            C(\theta)
        }
    \right)
    \alpha^{1+\gamma}C(\theta).
\]
Then the preceding decomposition becomes
\begin{equation}
    H^{(\gamma)}(\theta,\alpha)
    =
    S_\phi^{(\gamma)}
    \left(
        \alpha_0\mathbb{P}_{\theta_0},
        \alpha\mathbb{P}_{\theta}
    \right)
    +
    \phi'
    \left(
        c(\theta,\alpha)
    \right)
    \epsilon_0
    \alpha^\gamma
    \rho_\gamma(\theta).
    \label{eq:holder_bias_decomposition}
\end{equation}
Since \(\phi'\le0\), the second term in
Eq.~\eqref{eq:holder_bias_decomposition} is non-positive. Hence
\begin{equation}
    H^{(\gamma)}(\theta,\alpha)
    \le
    S_\phi^{(\gamma)}
    \left(
        \alpha_0\mathbb{P}_{\theta_0},
        \alpha\mathbb{P}_{\theta}
    \right).
    \label{eq:H_le_clean_score}
\end{equation}

We now apply this inequality at the true parameter
$\eta_0=(\theta_0,\alpha_0)$:
\[
    H^{(\gamma)}(\theta_0,\alpha_0)
    \le
    S_\phi^{(\gamma)}
    \left(
        \alpha_0\mathbb{P}_{\theta_0},
        \alpha_0\mathbb{P}_{\theta_0}
    \right).
\]
By the properness of the H\"{o}lder score for non-negative measures,
\[
    S_\phi^{(\gamma)}
    \left(
        \alpha_0\mathbb{P}_{\theta_0},
        \alpha_0\mathbb{P}_{\theta_0}
    \right)
    \le
    S_\phi^{(\gamma)}
    \left(
        \alpha_0\mathbb{P}_{\theta_0},
        \alpha^*\mathbb{P}_{\theta^*}
    \right).
\]
Using the decomposition
in Eq.~\eqref{eq:holder_bias_decomposition} at
$(\theta^*,\alpha^*)$, we have
\[
\begin{aligned}
    S_\phi^{(\gamma)}
    \left(
        \alpha_0\mathbb{P}_{\theta_0},
        \alpha^*\mathbb{P}_{\theta^*}
    \right)
    &=
    H^{(\gamma)}(\theta^*,\alpha^*)
    -
    \phi'
    \left(
        c(\theta^*,\alpha^*)
    \right)
    \epsilon_0
    (\alpha^*)^\gamma
    \rho_\gamma(\theta^*).
\end{aligned}
\]
Because \(-L\le \phi'\le0\) and \(0<\alpha^*\le1\),
\[
    -
    \phi'
    \left(
        c(\theta^*,\alpha^*)
    \right)
    \epsilon_0
    (\alpha^*)^\gamma
    \rho_\gamma(\theta^*)
    \le
    L\epsilon_0\rho_\gamma(\theta^*).
\]
Combining the preceding inequalities yields
\begin{equation}
    H^{(\gamma)}(\theta_0,\alpha_0)
    \le
    H^{(\gamma)}(\theta^*,\alpha^*)
    +
    L\epsilon_0\rho_\gamma(\theta^*).
    \label{eq:holder_level_bound}
\end{equation}
This proves that
\[
    (\theta_0,\alpha_0)
    \in
    U_\rho,
\]
where \(U_\rho\) is the level set defined in
\Cref{thm:error_control_holder}.

It remains to convert the level-set bound into a parameter-distance bound.
By assumption, \(U_\rho\) is contained in a convex neighbourhood on which
\(H^{(\gamma)}\) is \(\delta\)-strongly convex.
Since both \(\eta^*=(\theta^*,\alpha^*)\) and
\(\eta_0=(\theta_0,\alpha_0)\) belong to this neighbourhood, strong convexity
and the fact that \(\eta^*\) minimises \(H^{(\gamma)}\) imply
\begin{equation}
    H^{(\gamma)}(\eta_0)
    -
    H^{(\gamma)}(\eta^*)
    \ge
    \frac{\delta}{2}
    \left\|
        \eta_0-\eta^*
    \right\|_2^2.
    \label{eq:strong_conv_distance}
\end{equation}
For completeness, this inequality follows directly from strong convexity:
for \(t\in(0,1)\), let
\[
    \eta_t
    =
    (1-t)\eta^*+t\eta_0.
\]
Since \(\eta^*\) minimises \(H^{(\gamma)}\),
\[
    H^{(\gamma)}(\eta^*)
    \le
    H^{(\gamma)}(\eta_t).
\]
On the other hand, strong convexity gives
\[
    H^{(\gamma)}(\eta_t)
    \le
    (1-t)H^{(\gamma)}(\eta^*)
    +
    tH^{(\gamma)}(\eta_0)
    -
    \frac{\delta}{2}t(1-t)
    \left\|
        \eta_0-\eta^*
    \right\|_2^2.
\]
Combining these two displays and dividing by \(t>0\), we obtain
\[
    H^{(\gamma)}(\eta_0)
    -
    H^{(\gamma)}(\eta^*)
    \ge
    \frac{\delta}{2}(1-t)
    \left\|
        \eta_0-\eta^*
    \right\|_2^2.
\]
Letting \(t\downarrow0\) gives Eq.~\eqref{eq:strong_conv_distance}.

Finally, combining
Eq.~\eqref{eq:holder_level_bound} and Eq.~\eqref{eq:strong_conv_distance}, we obtain
\[
    \frac{\delta}{2}
    \left\|
        \eta_0-\eta^*
    \right\|_2^2
    \le
    L\epsilon_0\rho_\gamma(\theta^*).
\]
Therefore,
\[
    \left\|
        (\theta^*,\alpha^*)
        -
        (\theta_0,\alpha_0)
    \right\|_2^2
    \le
    \frac{2L}{\delta}
    \epsilon_0
    \rho_\gamma(\theta^*).
\]
This completes the proof.
\end{proof}

\subsection{Proof of Theorem~\ref{thm:bvm_holder}}
\label{sec:proof_bvm_holder}

We use the following sufficient conditions for the Gibbs posterior BvM
approximation, obtained from \citet[Theorem~4]{Miller21}.  The proof below
makes explicit the separation condition needed for the reduction.

\begin{thm}[General Gibbs posterior BvM theorem]
\label{thm:general_gibbs_bvm_appendix}
Let $\Psi\subseteq\mathbb{R}^v$ be a Borel set.
Let $D:\Psi\to\mathbb{R}$ be a measurable population loss and let
$D_n:\Psi\to\mathbb{R}$ be an empirical loss depending on the data.
Let
\[
    \pi_n^D(\xi)
    \propto
    \exp\{-nD_n(\xi)\}\pi(\xi)
\]
be the corresponding Gibbs posterior, where $\pi$ is a probability density
and the normalising constant is finite and positive for every $n$.
Suppose that, for some bounded, convex, open set $U\subseteq\Psi$, the following
conditions hold:
\begin{enumerate}
\item[C1] $D_n(\xi)\to D(\xi)$ almost surely for each fixed $\xi\in U$.
\item[C2] $D_n$ is three times continuously differentiable on $U$, and
\[
    \limsup_{n\to\infty}
    \sup_{\xi\in U}
    \left\|
        \nabla^r D_n(\xi)
    \right\|
    <
    \infty
    \quad
    \text{a.s.}
    \qquad
    (r=1,2,3).
\]
\item[C3] A point $\xi_*\in U$ uniquely minimises $D$ on $\Psi$.
Almost surely, for every $\varepsilon>0$,
\[
    \liminf_{n\to\infty}
    \inf_{\substack{\xi\in\Psi\\\|\xi-\xi_*\|_2\ge\varepsilon}}
    \{D_n(\xi)-D_n(\xi_*)\}>0.
\]
\item[C4] With $H_n(\xi)\coloneqq\nabla^2D_n(\xi)$,
\[
    H_n(\xi_*)\to H_*
    \quad
    \text{a.s.}
\]
for some fixed positive-definite matrix $H_*$.
\item[C5] The prior density $\pi$ is continuous and positive at $\xi_*$.
\end{enumerate}
Then, almost surely, empirical minimisers exist for all sufficiently large
$n$ and every sequence of empirical minimisers converges to $\xi_*$.
For any such sequence, let $\tilde\pi_n^D$ denote the density of
\[
    \sqrt n(\xi-\xi_n),
    \qquad
    \xi\sim\pi_n^D,
\]
where $\xi_n$ is an empirical minimiser for all sufficiently large $n$.
Then
\[
    \int_{\mathbb{R}^v}
    \left|
        \tilde\pi_n^D(u)
        -
        \frac{1}{Z_*}
        \exp
        \left\{
            -\frac{1}{2}
            u^\top H_*u
        \right\}
    \right|
    \mathrm{d}u
    \to
    0
    \quad
    \text{a.s.},
\]
where $Z_*$ is the normalising constant of the Gaussian kernel on the
right-hand side.
\end{thm}

\begin{proof}
Work on a probability-one event on which C2--C4 hold.
Choose $r>0$ with $\overline B_r(\xi_*)\subset U$.
By C3, the loss outside $B_r(\xi_*)$ is strictly larger than
$D_n(\xi_*)$ for all sufficiently large $n$.
Continuity on the compact ball $\overline B_r(\xi_*)$ therefore gives a
global minimiser in $B_r(\xi_*)$.
Applying C3 with any $\varepsilon>0$ shows that every global minimiser
eventually belongs to $B_\varepsilon(\xi_*)$.
Thus any sequence $\xi_n$ of empirical minimisers converges to $\xi_*$
and satisfies $\nabla D_n(\xi_n)=0$ eventually.

The third-derivative bound in C2 and C4 imply
\[
    H_n(\xi_n)\longrightarrow H_*.
\]
Taylor's theorem consequently gives, for $h$ in a fixed neighbourhood
of zero and all sufficiently large $n$,
\[
    D_n(\xi_n+h)
    =D_n(\xi_n)+\tfrac12h^\top H_n(\xi_n)h+R_n(h),
    \qquad |R_n(h)|\le C\|h\|_2^3,
\]
with a finite constant $C$ on this sample path.
For every $\varepsilon>0$, $\xi_n\to\xi_*$ and
$D_n(\xi_n)\le D_n(\xi_*)$ also give
\[
\begin{aligned}
    &\liminf_{n\to\infty}
      \inf_{\substack{\xi\in\Psi\\\|\xi-\xi_n\|_2\ge\varepsilon}}
      \{D_n(\xi)-D_n(\xi_n)\} \\
    &\qquad\ge
      \liminf_{n\to\infty}
      \inf_{\substack{\xi\in\Psi\\\|\xi-\xi_*\|_2\ge\varepsilon/2}}
      \{D_n(\xi)-D_n(\xi_*)\}>0.
\end{aligned}
\]
Extend $\pi$ by zero and $D_n$ by $D_n(\xi_*)+1$ outside $\Psi$.
The quadratic expansion and separation conditions of
\citet[Theorem~4]{Miller21} now hold on this sample path.
That theorem yields the stated total variation convergence.
\end{proof}

\begin{proof}[Proof of Theorem~\ref{thm:bvm_holder}]
We apply \Cref{thm:general_gibbs_bvm_appendix} with
\[
    \Psi
    =
    \Theta\times(0,1],
    \qquad
    D_n(\eta)
    =
    \beta H_n^{(\gamma)}(\eta),
    \qquad
    D(\eta)
    =
    \beta H^{(\gamma)}(\eta).
\]
Since $\beta\in(0,\infty)$ is fixed, the empirical and population minimisers of
$D_n$ and $D$ coincide with those of $H_n^{(\gamma)}$ and $H^{(\gamma)}$.
Moreover, \Cref{assump:basic_regular,lem:lipschitz_phi} give
\[
    \sup_{\eta\in\Theta\times(0,1]}
    |H_n^{(\gamma)}(\eta)|
    \le (|\phi(1)|+2L)M^\gamma,
\]
so the Gibbs normalising constant is finite and positive.

\paragraph{Condition C1.}
For each fixed $\eta\in U$, \Cref{lem:pointwise_conv_holder} gives
\[
    H_n^{(\gamma)}(\eta)
    \to
    H^{(\gamma)}(\eta)
    \quad
    \text{a.s.}
\]
Multiplying by the fixed constant $\beta$ yields
\[
    D_n(\eta)
    \to
    D(\eta)
    \quad
    \text{a.s.}
\]

\paragraph{Condition C2.}
For $r=1,2,3$,
\[
    \nabla_\eta^rD_n(\eta)
    =
    \beta
    \nabla_\eta^rH_n^{(\gamma)}(\eta).
\]
By
\Cref{lem:bounded_derivative_first,lem:bounded_derivative_second,lem:bounded_derivative_third},
the derivatives of $H_n^{(\gamma)}$ up to order three are uniformly bounded on
$U$ almost surely for all sufficiently large $n$.  Hence the same holds for
$D_n$.

\paragraph{Condition C3.}
Since $\beta>0$, the unique minimiser of $D$ is $\eta^*$.
For every $\varepsilon>0$, \eqref{eq:empirical_risk_separation} gives
\[
\begin{aligned}
    &\liminf_{n\to\infty}
      \inf_{\substack{
          \eta\in\Theta\times(0,1]\\
          \|\eta-\eta^*\|_2\ge\varepsilon
      }}
      \{D_n(\eta)-D_n(\eta^*)\} \\
    &\qquad=
      \beta\liminf_{n\to\infty}
      \inf_{\substack{
          \eta\in\Theta\times(0,1]\\
          \|\eta-\eta^*\|_2\ge\varepsilon
      }}
      \{H_n^{(\gamma)}(\eta)-H_n^{(\gamma)}(\eta^*)\}>0
      \quad\text{a.s.}
\end{aligned}
\]
Thus C3 holds.  The eventual existence, localisation and consistency of
every sequence of empirical minimisers also follow from
\Cref{lem:strong_consistency}.

\paragraph{Condition C4.}
Let
\[
    \bm{J}_n(\eta)
    \coloneqq
    \nabla_\eta^2H_n^{(\gamma)}(\eta),
    \qquad
    \bm{J}(\eta)
    \coloneqq
    \nabla_\eta^2H^{(\gamma)}(\eta),
    \qquad
    \bm{J}_*
    \coloneqq
    \bm{J}(\eta^*).
\]
Then
\[
    \nabla_\eta^2D_n(\eta)
    =
    \beta\bm{J}_n(\eta).
\]
By Part 1 of \Cref{prop:as_conv_derivatives} with $r=2$,
\[
    \bm{J}_n(\eta^*)
    \to
    \bm{J}_*
    \quad
    \text{a.s.}
\]
Therefore,
\[
    \nabla_\eta^2D_n(\eta^*)
    =
    \beta\bm{J}_n(\eta^*)
    \to
    \beta\bm{J}_*
    \quad
    \text{a.s.}
\]
By assumption, $\bm{J}_*$ is positive definite.  Since $\beta>0$,
$\beta\bm{J}_*$ is positive definite.  Thus C4 holds with
\[
    H_*
    =
    \beta\bm{J}_*.
\]

\paragraph{Condition C5.}
By assumption, the prior density $\pi$ is continuous and positive at $\eta^*$.
Thus C5 holds.

Applying \Cref{thm:general_gibbs_bvm_appendix} yields
\[
    \int_{\mathbb{R}^{d+1}}
    \left|
        \tilde{\pi}_n^H(u)
        -
        \frac{1}{Z_*}
        \exp
        \left\{
            -\frac{1}{2}
            u^\top
            \left(
                \beta\bm{J}_*
            \right)
            u
        \right\}
    \right|
    \mathrm{d}u
    \to
    0
    \quad
    \text{a.s.}
\]
The normalised Gaussian kernel on the right-hand side is the density of
\[
    \mathcal{N}
    \left(
        \mathbf{0},
        (\beta\bm{J}_*)^{-1}
    \right).
\]
Since
\[
    \bm{\Sigma}_*
    =
    (\beta\bm{J}_*)^{-1},
\]
the desired conclusion follows.
\end{proof}

\subsection{BvM theorem under data-dependent affine-scaling temperatures}
\label{app:bvm_data_dependent_temperature}

\begin{cor}[BvM theorem under a data-dependent affine-scaling temperature]
\label{cor:bvm_data_dependent_temperature}
Suppose that the assumptions of \Cref{thm:bvm_holder} hold, except that the
fixed temperature parameter $\beta$ is replaced by a data-dependent
affine-scaling temperature $\beta_{T_n}$.
Assume that $\beta_{T_n}$ is measurable with respect to the data, does not
depend on the parameter $\eta=(\theta,\alpha)$, and satisfies, for some
non-random constant $\beta_{T_\infty}\in(0,\infty)$,
\[
    \beta_{T_n}
    \to
    \beta_{T_\infty}
    \quad
    \text{a.s.}
\]
Define
\[
    \pi_{n,T}^{H}(\eta)
    \propto
    \exp
    \left\{
        -\beta_{T_n}nH_n^{(\gamma)}(\eta)
    \right\}
    \pi(\eta).
\]
For all sufficiently large $n$, let $\eta_n$ be any empirical minimiser of $H_n^{(\gamma)}$, and let
\[
    \bm{J}_*
    =
    \nabla_\eta^2H^{(\gamma)}(\eta^*).
\]
Define
\[
    \bm{\Sigma}_{*,T_\infty}
    \coloneqq
    \left(
        \beta_{T_\infty}\bm{J}_*
    \right)^{-1}.
\]
Let $\tilde\pi_{n,T}^H$ denote the density of
\[
    \xi
    =
    \sqrt n(\eta-\eta_n),
    \qquad
    \eta\sim\pi_{n,T}^{H}.
\]
Then
\[
    \int_{\mathbb{R}^{d+1}}
    \left|
        \tilde\pi_{n,T}^{H}(\xi)
        -
        \mathcal{N}
        \left(
            \xi
            \mid
            \mathbf{0},
            \bm{\Sigma}_{*,T_\infty}
        \right)
    \right|
    \mathrm{d}\xi
    \to
    0
    \quad
    \text{a.s.}
\]
\end{cor}

\begin{proof}
Write
\[
    b_n
    \coloneqq
    \beta_{T_n},
    \qquad
    b_\infty
    \coloneqq
    \beta_{T_\infty}.
\]
By assumption,
\[
    b_n\to b_\infty\in(0,\infty)
    \quad
    \text{a.s.}
\]
We apply the same general Gibbs posterior BvM theorem with
\[
    D_n(\eta)
    =
    b_nH_n^{(\gamma)}(\eta),
    \qquad
    D(\eta)
    =
    b_\infty H^{(\gamma)}(\eta).
\]
Since $b_n>0$ for all sufficiently large $n$ almost surely, multiplication by
$b_n$ does not change the empirical minimisers.  Similarly, multiplication by
$b_\infty>0$ does not change the population minimiser.
The Gibbs normalising constant is finite and positive by the same
bounded-risk argument as in the fixed-temperature proof.

To verify C3, note that $b_n\ge b_\infty/2$ eventually almost surely.
For each $\varepsilon>0$, the unscaled risk infimum in
\eqref{eq:empirical_risk_separation} is also eventually positive.
Consequently,
\[
\begin{aligned}
    &\liminf_{n\to\infty}
      \inf_{\substack{
          \eta\in\Theta\times(0,1]\\
          \|\eta-\eta^*\|_2\ge\varepsilon
      }}
      \{D_n(\eta)-D_n(\eta^*)\} \\
    &\qquad\ge
      \frac{b_\infty}{2}\liminf_{n\to\infty}
      \inf_{\substack{
          \eta\in\Theta\times(0,1]\\
          \|\eta-\eta^*\|_2\ge\varepsilon
      }}
      \{H_n^{(\gamma)}(\eta)-H_n^{(\gamma)}(\eta^*)\}>0
      \quad\text{a.s.}
\end{aligned}
\]

For each fixed $\eta\in U$,
\[
    H_n^{(\gamma)}(\eta)
    \to
    H^{(\gamma)}(\eta)
    \quad
    \text{a.s.}
\]
by \Cref{lem:pointwise_conv_holder}.  Hence
\[
    D_n(\eta)
    =
    b_nH_n^{(\gamma)}(\eta)
    \to
    b_\infty H^{(\gamma)}(\eta)
    =
    D(\eta)
    \quad
    \text{a.s.}
\]

For $r=1,2,3$,
\[
    \nabla_\eta^rD_n(\eta)
    =
    b_n\nabla_\eta^rH_n^{(\gamma)}(\eta).
\]
Since $b_n$ is eventually bounded above almost surely, the derivative
boundedness follows from
\Cref{lem:bounded_derivative_first,lem:bounded_derivative_second,lem:bounded_derivative_third}.

For the Hessian condition,
\[
    \nabla_\eta^2D_n(\eta^*)
    =
    b_n
    \nabla_\eta^2H_n^{(\gamma)}(\eta^*)
    \to
    b_\infty\bm{J}_*
    \quad
    \text{a.s.}
\]
by \Cref{prop:as_conv_derivatives}.  Since $\bm{J}_*$ is positive definite and
$b_\infty>0$, the limiting Hessian $b_\infty\bm{J}_*$ is positive definite.

The prior condition is unchanged.  The general Gibbs posterior BvM theorem
therefore gives the Gaussian limit with covariance
\[
    (b_\infty\bm{J}_*)^{-1}
    =
    (\beta_{T_\infty}\bm{J}_*)^{-1}.
\]
This proves the claim.
\end{proof}

\begin{rem}[Sufficient condition for the sample-covariance temperature]
For the sample-covariance choice
\[
    \beta_{T_n}
    =
    |\det\widehat{\Omega}_n|^{\gamma/2},
\]
the condition
$\beta_{T_n}\to\beta_{T_\infty}\in(0,\infty)$ holds if
\[
    \widehat{\Omega}_n
    \to
    \Omega_{\mathbb{P}}
    \quad
    \text{a.s.}
\]
for some finite positive-definite matrix $\Omega_{\mathbb{P}}$.
A standard sufficient condition is
\[
    \mathbb{E}_{\mathbb{P}}\|Y\|^2<\infty,
    \qquad
    \operatorname{Var}_{\mathbb{P}}(Y)\succ0,
\]
in which case
\[
    \Omega_{\mathbb{P}}
    =
    \operatorname{Var}_{\mathbb{P}}(Y).
\]
\end{rem}

\subsection{Proof of Theorem~\ref{thm:bvm_holder_corollary}} 
\label{sec:proof_bvm_holder_corollary}

\begin{proof}
Let $\Delta\coloneqq\eta^*-\eta_0$,
$\bm{J}_*\coloneqq\bm{J}(\eta^*)$, and
$\bm{J}_0\coloneqq\bm{J}(\eta_0)$, where both Hessians are taken from the same population risk $H^{(\gamma)}$.
Recall that
\[
    \bm{\Sigma}_*=(\beta\bm{J}_*)^{-1},
    \qquad
    \bm{\Sigma}_0=(\beta\bm{J}_0)^{-1},
\]
and that $\mathsf{G}_*=\mathcal{N}(\mathbf{0},\bm{\Sigma}_*)$ and
$\mathsf{G}_0=\mathcal{N}(\mathbf{0},\bm{\Sigma}_0)$ are centred Gaussian distributions.
The Gaussian Wasserstein formula therefore gives
\begin{align}
    W_2^2(\mathsf{G}_*,\mathsf{G}_0)
    &=d_{\mathrm B}^2(\bm{\Sigma}_*,\bm{\Sigma}_0),
    \label{eq:w2_gaussian_decomposition}
\end{align}
where
\[
    d_{\mathrm B}^2(A,B)
    \coloneqq
    \operatorname{tr}\left\{
        A+B-2(B^{1/2}AB^{1/2})^{1/2}
    \right\}
\]
is the squared Bures distance between positive-definite matrices.
Coupling $\bm{\Sigma}_*^{1/2}Z$ and $\bm{\Sigma}_0^{1/2}Z$ using a common
$Z\sim\mathcal{N}(\mathbf{0},I_{d+1})$ yields
\[
    d_{\mathrm B}^2(\bm{\Sigma}_*,\bm{\Sigma}_0)
    \le
    \|\bm{\Sigma}_*^{1/2}-\bm{\Sigma}_0^{1/2}\|_{\mathrm F}^2.
\]

The Hessian bounds imply
\[
    \bm{\Sigma}_*,\bm{\Sigma}_0
    \succeq
    \frac{1}{\beta\overline{\lambda}}I_{d+1}.
\]
Consequently, the matrix square-root bound gives
\[
    \|\bm{\Sigma}_*^{1/2}-\bm{\Sigma}_0^{1/2}\|_{\mathrm F}
    \le
    \frac{\sqrt{\beta\overline{\lambda}}}{2}
    \|\bm{\Sigma}_*-\bm{\Sigma}_0\|_{\mathrm F}.
\]
For completeness, this inequality follows by setting
$A=\bm{\Sigma}_*^{1/2}$, $B=\bm{\Sigma}_0^{1/2}$, and $X=A-B$.
The identity $AX+XB=\bm{\Sigma}_*-\bm{\Sigma}_0$, together with
\[
    \langle X,AX+XB\rangle_{\mathrm F}
    \ge
    \frac{2}{\sqrt{\beta\overline{\lambda}}}\|X\|_{\mathrm F}^2
\]
and the Cauchy--Schwarz inequality, gives the stated bound.

The inverse identity and the assumed Hessian regularity also give
\begin{align*}
    \|\bm{\Sigma}_*-\bm{\Sigma}_0\|_{\mathrm F}
    &=\beta^{-1}\|\bm{J}_*^{-1}-\bm{J}_0^{-1}\|_{\mathrm F}\\
    &=\beta^{-1}
      \|\bm{J}_*^{-1}(\bm{J}_0-\bm{J}_*)\bm{J}_0^{-1}\|_{\mathrm F}\\
    &\le
      \beta^{-1}\underline{\lambda}^{-2}L_J\|\Delta\|_2.
\end{align*}
Combining these inequalities with Eq.~\eqref{eq:w2_gaussian_decomposition} yields
\[
    W_2^2(\mathsf{G}_*,\mathsf{G}_0)
    \le
    \frac{\overline{\lambda}L_J^2}
         {4\beta\underline{\lambda}^4}
    \|\eta^*-\eta_0\|_2^2.
\]
Finally, \Cref{thm:error_control_holder} gives
\[
    \|\eta^*-\eta_0\|_2^2
    \le
    \frac{2L}{\delta}\epsilon_0\rho_\gamma(\theta^*).
\]
Substitution proves Eq.~\eqref{eq:wasserstein_bias_bound} with
\[
    C_\beta
    =\frac{L\overline{\lambda}L_J^2}
          {2\beta\delta\underline{\lambda}^4}.
\]
\end{proof}

\subsection{Local covariance stability under data-dependent affine-scaling temperatures}
\label{sec:proof_bvm_bias_data_dependent_temperature}

\begin{cor}[Local covariance stability under a data-dependent affine-scaling temperature]
\label{cor:bvm_bias_data_dependent_temperature}
Suppose that the conditions of \Cref{cor:bvm_data_dependent_temperature} hold, with
\[
    \beta_{T_n}\to\beta_{T_\infty}\in(0,\infty)
    \qquad\text{a.s.}
\]
for a non-random limit $\beta_{T_\infty}$.
Assume also the population bias and Hessian conditions of \Cref{thm:bvm_holder_corollary}, including $\eta_0\in U$.
Define
\[
    \bm{\Sigma}_{*,T_\infty}
    \coloneqq(\beta_{T_\infty}\bm{J}_*)^{-1},
    \qquad
    \bm{\Sigma}_{0,T_\infty}
    \coloneqq(\beta_{T_\infty}\bm{J}_0)^{-1},
\]
and let
\[
    \mathsf{G}_{*,T_\infty}
    \coloneqq\mathcal{N}(\mathbf{0},\bm{\Sigma}_{*,T_\infty}),
    \qquad
    \mathsf{G}_{0,T_\infty}
    \coloneqq\mathcal{N}(\mathbf{0},\bm{\Sigma}_{0,T_\infty}).
\]
Here $\bm{J}_*=\nabla_\eta^2H^{(\gamma)}(\eta^*)$ and
$\bm{J}_0=\nabla_\eta^2H^{(\gamma)}(\eta_0)$ are evaluated under the same population risk.
Then
\[
    W_2^2(\mathsf{G}_{*,T_\infty},\mathsf{G}_{0,T_\infty})
    \le
    \frac{\overline{\lambda}L_J^2}
         {4\beta_{T_\infty}\underline{\lambda}^4}
    \|\eta^*-\eta_0\|_2^2
    \le
    C_T\epsilon_0\rho_\gamma(\theta^*),
\]
where
\[
    C_T
    \coloneqq
    \frac{L\overline{\lambda}L_J^2}
         {2\beta_{T_\infty}\delta\underline{\lambda}^4}.
\]
\end{cor}

\begin{proof}
By \Cref{cor:bvm_data_dependent_temperature}, the centred Gaussian limit has covariance
$(\beta_{T_\infty}\bm{J}_*)^{-1}$.
The proof of \Cref{thm:bvm_holder_corollary} is deterministic for any fixed positive temperature.
Applying it with $\beta=\beta_{T_\infty}$ gives
\[
    W_2^2(\mathsf{G}_{*,T_\infty},\mathsf{G}_{0,T_\infty})
    \le
    \frac{\overline{\lambda}L_J^2}
         {4\beta_{T_\infty}\underline{\lambda}^4}
    \|\eta^*-\eta_0\|_2^2.
\]
The population bias bound in \Cref{thm:error_control_holder} then yields the stated value of $C_T$.
\end{proof}

\section{Alternative Choices of \texorpdfstring{$\phi$}{phi}}
\label{app:alternative_phi}

The main text uses the linear DPD specification

\[
    s(z)=z,
    \qquad
    \phi_{\mathrm{DPD}}(z)
    =
    \gamma-(1+\gamma)z
\]

as the default implementation of H\"older--Bayes.
As discussed in \Cref{sec:choice_phi}, however, the construction

\[
    \phi_s(z)
    =
    \gamma-(1+\gamma)s(z)
\]

also permits nonlinear choices of $s$ satisfying the required H\"older admissibility and bounded-derivative conditions.
Here we give two examples used to assess the sensitivity of the proposed framework to the particular choice of $\phi$.

We consider the exponential saturation map
\begin{equation}
s_{\exp,\kappa}(z)
=
1+
\frac{
1-\exp\{-\kappa(z-1)\}
}{
\kappa
},
\label{eq:exp_sat_map}
\end{equation}
and the rational saturation map
\begin{equation}
s_{\mathrm{rat},\kappa}(z)
=
1+
\frac{
z-1
}{
1+\kappa(z-1)
}.
\label{eq:rat_sat_map}
\end{equation}
The corresponding H\"older specifications are referred to as ExpSatDPD and RatSatDPD, respectively.

Both maps satisfy $s(1)=1$.
Their derivatives are
\begin{equation}
s'_{\exp,\kappa}(z)
=
\exp\{-\kappa(z-1)\},
\qquad
s'_{\mathrm{rat},\kappa}(z)
=
\frac{1}{
{1+\kappa(z-1)}^2
}.
\label{eq:saturated_derivatives}
\end{equation}
For fixed $\kappa>0$, the exponential derivative is positive and bounded on $z>0$.
For the rational specification, $0<\kappa<1$ guarantees that
$1+\kappa(z-1)>0$ for all $z>0$, and hence its derivative is likewise positive and bounded.
Moreover, both maps satisfy

\[
    s(z)\le z,
    \qquad z>0.
\]

Consequently, the construction
$\phi_s(z)=\gamma-(1+\gamma)s(z)$
satisfies

\[
    \phi_s(1)=-1,
    \qquad
    \phi_s(z)\ge -z^{1+\gamma},
\]

together with the bounded-derivative condition used in \Cref{sec:theory}.

Both families recover the linear DPD specification as the saturation vanishes:
\begin{equation}
\lim_{\kappa\to0}
s_{\exp,\kappa}(z)
=
\lim_{\kappa\to0}
s_{\mathrm{rat},\kappa}(z)
=
z.
\label{eq:saturation_to_dpd}
\end{equation}
For positive $\kappa$, their derivatives decrease towards zero as $z$ increases.
They therefore attenuate the response of $\phi$ to increasingly large values of the argument $z$ while retaining dependence on the inlier proportion $1-\epsilon$.

For the sensitivity experiments, we use the fixed values

\[
    \kappa_{\exp}=0.5,
    \qquad
    \kappa_{\mathrm{rat}}=0.25.
\]

These numerical values are family-specific and should not be interpreted as a common saturation scale.
To compare their local behaviour, write $u=z-1$.
Around $z=1$, the exponential map satisfies
\begin{equation}
s_{\exp,\kappa}(1+u)
=
1+u-\frac{\kappa}{2}u^2+O(u^3),
\label{eq:exp_sat_local}
\end{equation}
whereas the rational map satisfies
\begin{equation}
s_{\mathrm{rat},\kappa}(1+u)
=
1+u-\kappa u^2+O(u^3).
\label{eq:rat_sat_local}
\end{equation}
Thus, $\kappa_{\mathrm{rat}}=0.25$ approximately matches the local quadratic saturation of
$\kappa_{\exp}=0.5$.
We fix these values rather than tuning an additional hyperparameter, since the purpose of these variants is to examine sensitivity to the choice of $\phi$, rather than to introduce separately optimised procedures.

The experimental protocol for these sensitivity variants is described together with the remaining experimental settings in \Cref{app:experimental-settings}, and their full numerical comparisons with the default DPD specification are reported in \Cref{app:additional-experimental-results}.

\section{Experimental settings}
\label{app:experimental-settings}

This section specifies the experiments in Section~\ref{sec:experiment} and the supplementary comparisons below.
The verified collection contains $450$ location--scale datasets, $540$ synthetic regression datasets, $90$ real-data benchmark configurations, and $90$ cleaning configurations.
All planned datasets and splits contribute to the primary summaries.
Data, contamination indicators, and perturbations are shared across methods within a condition and replicate.
Synthetic configurations have $30$ independent replicates; real-data configurations have five fixed train--test splits.

\subsection{Posterior specification and computation}
\label{app:posterior-computation}

\paragraph{Methods and priors.}
The default H\"older--Bayes specification is DPD.
The ExpSatDPD and RatSatDPD variants in Section~\ref{app:alternative_phi} use fixed $\kappa_{\exp}=0.5$ and $\kappa_{\mathrm{rat}}=0.25$.
They are sensitivity specifications, without dataset-specific tuning of $\kappa$.
The comparison methods are Standard Bayes (the ordinary Gaussian posterior, labelled NLL in the result files), the $\gamma$-divergence posterior (GammaDiv), and a posterior based on a Student-$t$ likelihood with fixed degrees of freedom $\nu=4$ (StudentT in the result files).
The robustification grid is $\gamma\in\{0.1,0.5,1.0\}$.
Main synthetic results use $\gamma=0.5$ for location--scale estimation and $\gamma=1$ for regression; main real-data results use $\gamma=0.5$.

Priors are defined in the coordinates passed to each likelihood.
Each Gaussian location or regression coefficient has an independent $\mathcal{N}(0,10^2)$ prior, and each log-scale has an $\mathcal{N}(0,1)$ prior.
H-B has the independent prior $\alpha\sim\operatorname{Beta}(1,1)$, with $\epsilon=1-\alpha$.
Sampling is performed in $\rho=\operatorname{logit}(\alpha)$, with the corresponding change-of-variables Jacobian included in the target.
For regression H-B, the log-scale prior is truncated to $\log\sigma>\log(0.05)$ on the standardised response scale.
The corresponding Gaussian density and its analytic density-power integral are unchanged within that support.
Location--scale H-B has no such lower bound.
GammaDiv truncates every log-scale below $\log(0.05)$, in the original density coordinates or standardised regression coordinates.
Ordinary Gaussian and Student-$t$ fits, including downstream Gaussian refits, have no imposed scale floor.
The value $0.05$ is a scale restriction, not a lower bound on $\epsilon$.

\paragraph{Temperature and the GammaDiv target.}
Write $\beta_{\mathrm{temp}}$ for the temperature multiplier, to distinguish it from regression coefficients.
For density data $Y_i\in\mathbb R^2$, H-B uses
\[
 \beta_{\mathrm{temp}}=\det\{\widehat{\operatorname{Cov}}(Y)\}^{\gamma/2},
\]
where the sample covariance uses divisor $n-1$ and the same original observations as the likelihood.
For conditional regression, the response is the density variable; the covariates do not contribute a density Jacobian.
Accordingly, H-B uses
\[
 \beta_{\mathrm{temp}}=
 \widehat{\operatorname{Var}}(\widetilde Y^{\mathrm{obs}})^{\gamma/2},
\]
with divisor $n-1$, evaluated in the observed training-response coordinates actually supplied to the likelihood.
In particular, this is not the covariance determinant of the joint covariate--response vector.
Synthetic regression standardises the contaminated training responses, whereas the real-data protocol fixes standardisation before injecting response shifts.
Thus these two response-variance temperatures need not agree numerically.

GammaDiv uses its separate fixed multiplier $1$.
To identify the implemented version precisely, put $A_\gamma(\theta)=\int p_\theta^{1+\gamma}$, integrating over the response in regression.
Its log posterior, up to a parameter-independent constant, is
\[
 \log\pi(\theta)+\frac1\gamma\sum_{i=1}^n
 \left\{
 \frac{p_\theta(Y_i\mid X_i)^\gamma}
 {A_\gamma(\theta)^{\gamma/(1+\gamma)}}-1
 \right\}.
\]
For density models, $X_i$ is omitted.
This specifies the monotone density-power version used in the verified runs; it is not the logarithmic $\gamma$-cross-entropy target.

\paragraph{MCMC and verification.}
Every fit uses four chains.
The rerun configuration uses $4000$ warm-up iterations and $4000$ retained iterations per chain, with target acceptance $0.99$ and initial step size $0.1$ for the No-U-Turn Sampler (NUTS) stages.
Previously saved draws are reused only after input, target, and diagnostic checks.
Consequently, the archived per-fit records, rather than function defaults, identify the computation used for each retained fit.

H-B uses NUTS output to construct Student-$t$ mixture proposals, with a prior component.
The mixture is frozen before independence Metropolis--Hastings (MH) sampling.
The acceptance ratio uses the exact target and proposal density.
The proposal has seven degrees of freedom, covariance inflation $1.5$, and prior-component weight $0.1$.
GammaDiv constructs a frozen proposal from multiple observed-data starting regions and local curvature estimates, also with an exact MH correction.
Its proposal weights affect sampling efficiency, not the posterior probabilities assigned to modes.
Ordinary Gaussian and Student-$t$ fits use NUTS, with a dense mass matrix in the rerun sampler.

Fits with inadequate initial mixing are repaired before aggregation using fixed sampling recipes.
For H-B, the repair composes independence MH, a symmetric local proposal, and a random-walk update of $\rho$.
These repair runs use $4000$ warm-up and $8000$ retained iterations per chain; a fixed extension of one density fit also retains $8000$ iterations per chain.
All four chains and rejected states of an adopted fit remain in the analysis.
No replicate is removed to satisfy convergence requirements.
Reported chain diagnostics require rank-normalised $\widehat R\le1.01$ and bulk and tail effective sample sizes of at least $400$, including the derived scale and contamination coordinates where applicable.
Across the $8010$ reported diagnostic blocks, the largest $\widehat R$ is $1.00837$, and the smallest bulk and tail effective sample sizes are $700.612$ and $401.872$.
These are parameter-chain checks; they do not by themselves certify Monte Carlo accuracy for a predictive density far into a tail.

\subsection{Synthetic experiments}
\label{app:exp1-details}

\paragraph{Location--scale model.}
For each condition and replicate, $n=1000$ observations are generated independently from
\[
 (1-\epsilon_0)\mathcal{N}(0,I_2)+\epsilon_0\mathcal{N}((5,5)^\top,I_2),
 \qquad \epsilon_0\in\{0,0.1,0.2,0.3,0.4\}.
\]
The contamination indicators are Bernoulli draws; $n\epsilon_0$ is a nominal count, not a forced realised count.
The fitted Gaussian target has a free mean and diagonal covariance, with parameter vector $(\mu_1,\mu_2,\log\sigma_1,\log\sigma_2)$.
Likelihoods, priors, and temperatures use the original density-data coordinates.
Only the shifted-Gaussian scenario A is included in the verified collection; the alternative heavy-tail generator available in the code is not a reported experiment.
Location error is $\|\widehat\mu\|_2$, with $\widehat\mu$ the posterior mean.

\paragraph{Contaminated regression.}
Each replicate draws $b^*\sim \mathcal{N}(0,I_2)$ and independent training and test covariates from $\mathcal{N}(0,I_2)$.
The clean response is $Y=(1,X^\top)\beta^*+\xi$, where $\beta^*=(0,(b^*)^\top)^\top$ and $\xi\sim \mathcal{N}(0,1)$.
There are $1000$ training and $1000$ independent clean test observations.
For each training observation, $B_i\sim\operatorname{Bernoulli}(\epsilon_0)$ independently.
If $B_i=1$, the response is replaced by
\[
 Y_i^{\mathrm{obs}}=(1,X_i^\top)\beta^*+6+\xi_i',
 \qquad \xi_i'\sim \mathcal{N}(0,1),
\]
with independent replacement noise; the covariate is unchanged.
The regression grid is $\epsilon_0\in\{0,0.05,0.1,0.2,0.3,0.4\}$.
This mechanism shifts the intercept while preserving the clean slopes.
Covariates and observed training responses are centred and scaled using training means and standard deviations, and an intercept column is added.
The test design uses the same transformation.
Predictions, coefficients, and scales are mapped back to the original response coordinates for evaluation.
The main accuracy measure is clean-test RMSE.
Saved summaries additionally include coefficient $\ell_2$ error, coefficient RMSE, and absolute scale error relative to $\sigma^*=1$.

\paragraph{Posterior auditing and score approximation.}
Each draw flags the observations with the smallest pointwise Gaussian log densities, using its rounded contamination count clipped to $[0,n]$.
FoD is the average of these indicators, as specified in Algorithm~\ref{alg:outlier_detection}.
For conditional models the scores are conditional log densities at the observed covariates.

\begin{algorithm}[p]
\small\setstretch{1.05}
\caption{Posterior-Based Outlier Detection with H\"{o}lder--Bayes}
\label{alg:outlier_detection}
\begin{algorithmic}[1]
\REQUIRE Observations $\{y_i\}_{i=1}^{n}$; posterior samples
$\{(\theta^{(s)},\alpha^{(s)})\}_{s=1}^{S}$; compatibility scores
$\ell_i^{(s)}$; count summary rule
$r\in\{\mathrm{mean},\mathrm{lower},\mathrm{upper}\}$ if a fixed set is needed
\ENSURE Frequency-of-Detection scores $\{\operatorname{FoD}_i\}_{i=1}^{n}$ and,
optionally, a fixed audit set $\widehat{\mathcal O}$

\STATE Initialise $I_i^{(s)}\leftarrow0$ for all $s=1,\ldots,S$ and
$i=1,\ldots,n$.
\FOR{$s=1,\ldots,S$}
    \STATE Set $\epsilon^{(s)}\leftarrow1-\alpha^{(s)}$ and
    $K^{(s)}\leftarrow\operatorname{round}\{n\epsilon^{(s)}\}$.
    \STATE Truncate $K^{(s)}$ to $\{0,\ldots,n\}$.
    \IF{$K^{(s)}>0$}
        \STATE Compute compatibility scores $\ell_i^{(s)}$ under
        $\theta^{(s)}$ for $i=1,\ldots,n$.
        \STATE Let $\mathcal I_s$ be the indices of the $K^{(s)}$ smallest
        values among $\{\ell_i^{(s)}\}_{i=1}^{n}$.
        \STATE Set $I_i^{(s)}\leftarrow1$ for all $i\in\mathcal I_s$.
    \ENDIF
\ENDFOR
\STATE Compute
\[
    \operatorname{FoD}_i
    \leftarrow
    S^{-1}\sum_{s=1}^S I_i^{(s)},
    \qquad i=1,\ldots,n .
\]
\IF{a fixed audit set is required}
    \STATE Choose $\widehat K$ from $\{K^{(s)}\}_{s=1}^S$ using rule $r$.
    \STATE Let $\widehat{\mathcal O}$ be the indices of the $\widehat K$
    largest values among $\{\operatorname{FoD}_i\}_{i=1}^{n}$.
    \STATE \RETURN $\{\operatorname{FoD}_i\}_{i=1}^{n}$ and
    $\widehat{\mathcal O}$.
\ELSE
    \STATE \RETURN $\{\operatorname{FoD}_i\}_{i=1}^{n}$.
\ENDIF
\end{algorithmic}
\end{algorithm}

FoD and baseline training-data ranking scores use a fixed-seed random subset of at most $400$ retained draws.
This changes the score approximation, not the fitted posterior.
Baseline ranking uses negative log posterior-averaged Gaussian density for NLL and GammaDiv, and Student-$t$ density for StudentT.
In the saved location--scale summaries, the H-B $\epsilon$ summaries also use the selected scoring draws; regression $\epsilon$ summaries use all retained draws.
Synthetic recall at an estimated count uses the integer part of the mean selected $K^{(s)}$.
This count is distinct from the rounded central budget used in real-data cleaning.
Ranking summaries include the area under the precision--recall curve (AUPRC), the area under the receiver operating characteristic curve (AUROC), and recall at supplied counts; FoD also has a Brier score against the injected labels.
They do not establish that FoD is a calibrated probability of an observation's generative label.
At $\epsilon_0=0$, ranking against synthetic outlier labels is undefined.
Positive posterior audit fractions can remain even without injected contamination; a central equal-tailed interval from an interior continuous posterior need not cover the boundary value zero.

\paragraph{Figures and aggregation.}
Main Figures~\ref{fig:posterior} and~\ref{fig:likelihood-ranking} use one fixed replicate at $\epsilon_0=0.2$, shared across methods, with the main $\gamma$ values stated above.
The posterior figure includes the regression intercept as well as the two slopes.
The representative count annotations in Figure~\ref{fig:likelihood-ranking} are recomputed from all retained posterior draws: rounded means $207$ and $195$, with displayed $95\%$ intervals $[94,314]$ and $[137,251]$ for density and regression, respectively.
These annotations need not equal a count computed from the $400$-draw scoring subset.
Accuracy curves and their error bars are means and sample standard deviations across all $30$ datasets, rather than within-dataset posterior uncertainty.
Gamma and $\phi$ comparisons use the same replicate structure and contamination mechanisms.

\subsection{Real-data model-auditing and cleaning experiments}
\label{app:realdata-benchmark}

\paragraph{Data and splitting.}
The datasets are \texttt{cal-housing}, \texttt{kin8nm}, and \texttt{pumadyn-32fh}.
Non-finite rows are removed before an $80\%/20\%$ random train--test split.
For each of five replicate identifiers, row membership is fixed across methods, $\gamma$, and contamination conditions.
The test responses receive no synthetic contamination.

\begin{table}[t]
\centering
\caption{Verified real-data sizes. $p$ excludes the intercept.}
\label{tab:realdata-split-dimensions}
\begin{tabular}{lrrrr}
\toprule
Dataset & $p$ & Total & Training & Test \\
\midrule
\texttt{cal-housing} & 8 & 20640 & 16512 & 4128 \\
\texttt{kin8nm} & 8 & 8192 & 6553 & 1639 \\
\texttt{pumadyn-32fh} & 32 & 8192 & 6553 & 1639 \\
\bottomrule
\end{tabular}
\end{table}

\paragraph{Preprocessing and injection.}
Let $\bar x_j$ and $s_{x,j}$ be the mean and standard deviation of training covariate $j$.
Let $(\bar y,s_y)$ describe the training response before injection.
Covariate $j$ is transformed to $(x_j-\bar x_j)/s_{x,j}$, and an intercept is added.
Responses use $\widetilde y=(y-\bar y)/s_y$.
The same transformations are kept for all refits and test evaluations.
The two conditions are no injection and $20\%$ independent Bernoulli response injection:
\[
 \widetilde y_i^{\mathrm{obs}}=\widetilde y_i+m_i\zeta_i,
 \quad m_i\sim\operatorname{Bernoulli}(0.2),
 \quad \zeta_i\sim \mathcal{N}(6,1).
\]
Covariates are unchanged.
The no-injection data need not satisfy a Gaussian linear model, so their inferred $\epsilon$ is a model-relative audit fraction, not a known false-positive rate.
The $90$ configurations per stage are $3$ datasets, $3$ $\gamma$ values, $2$ injection conditions, and $5$ splits.

\paragraph{Prediction and ranking.}
All predictions use the original response units.
Root mean squared error (RMSE) and mean absolute error (MAE) compare the posterior mean response with the test response.
For saved draws, the reported benchmark score is
\[
 \operatorname{MC\!\text{-}\!NLPD}
 =-\frac1{n_{\mathrm{test}}}\sum_j
 \log\left\{\frac1S\sum_s
 \mathcal N\left(y_j;\bar y+s_y\widetilde x_j^\top\beta^{(s)},
                 s_y^2(\sigma^{(s)})^2\right)\right\}.
\]
All retained draws enter this test-density average, evaluated with log-sum-exp arithmetic.
The saved implementation uses this common Gaussian predictive construction for every benchmark method, including parameters fitted under the Student-$t$ likelihood.
For StudentT, this column is therefore a Gaussian working-density score, not its Student-$t$ posterior-predictive score; its $\sigma$ draws are Student-$t$ scale parameters.
Its training-data ranking, in contrast, uses the Student-$t$ density with $\nu=4$.
This distinction should be taken into account when comparing predictive-density scores.
The plug-in Gaussian NLPD saved separately uses $\widehat\sigma=s_y\exp\{\overline{\log\sigma}\}$; it is not the reported MC-NLPD.

H-B FoD and count samples use at most $400$ draws, while its reported $\epsilon$ mean and credible-interval endpoints use all retained draws.
For the no-injection condition, label-based ranking and removal-quality measures are not applicable.
For the injected condition, AUPRC and AUROC compare scores with the realised Bernoulli labels.
NLL, GammaDiv, and StudentT provide rankings but do not infer an audit fraction.
Their recall at an estimated count uses the central DPD count as a supplied cutoff.

\paragraph{Operational cleaning budgets.}
From the selected integer counts $K^{(s)}$, let
\[
 K_{\mathrm c}=\operatorname{round}(\overline K),\qquad
 K_{\mathrm l}=\lfloor q_{0.025}(K)\rfloor,\qquad
 K_{\mathrm u}=\lceil q_{0.975}(K)\rceil.
\]
Each is clipped to $[0,\lfloor n_{\mathrm{train}}/2\rfloor]$ for cleaning.
They are the central, lower, and upper budgets; the code suffixes optimistic and pessimistic denote lower and upper, respectively.
The budget is computed from integer count samples, not by rounding a separately computed full-draw mean or quantile of $n\epsilon$.
For each H-B variant, observations are removed in ascending Gaussian log likelihood evaluated at the posterior mean of $(\beta,\log\sigma)$.
This is a plug-in likelihood ranking, not the mean of draw-wise log likelihoods and not a final ranking of the FoD averages.
For a homoscedastic Gaussian, it is equivalent to decreasing absolute residual magnitude under that fitted mean.

\paragraph{Cleaning comparisons and refits.}
No removal retains all training rows.
Isolation forest, local outlier factor, and a one-class support vector machine (SVM) are fitted to separately standardised concatenations of original covariates and observed responses, using the central DPD budget.
Isolation forest has $200$ trees; local outlier factor uses $\min(20,n-1)$ neighbours; one-class SVM uses a radial basis function (RBF) kernel, \texttt{gamma=scale}, and $\nu=\min\{0.5,\max(K/n,0.01)\}$.
Under injected contamination, random removal samples that same number uniformly without replacement.
The oracle prioritises injected observations and, if its budget exceeds their realised count, selects the remaining required rows uniformly among non-injected observations.
It is budget-matched to central DPD, not a refit on all and only non-injected rows.

Every retained training set is evaluated using the same ordinary Gaussian posterior, the fixed preprocessing transformation, and the unchanged test set.
Thus lower and upper budgets change retained rows, not the downstream model.
Removal precision, recall, $F_1$, and the fraction of non-injected rows removed are computed per split against the injected labels, then averaged.
Removal counts are also averaged, so displayed counts need not be integers.
The no-injection condition has thirteen cleaning rules; the injected condition adds random and oracle removal, giving fifteen.

\paragraph{Aggregation and reproducibility.}
Each primary real-data entry is based on all five splits, with sample SD using divisor four.
Posterior means, interval endpoints, and counts are averaged separately across splits; an interval with averaged endpoints is not a pooled posterior credible interval.
Printed cleaning comparisons use $\gamma=0.5$; machine-readable results retain all $\gamma$ values and both injection conditions.
The only reported omission is the explicitly post-hoc, common-four-split sensitivity analysis for \texttt{cal-housing}; it does not replace the five-split primary results.
Results, source hashes, sampler settings, per-fit diagnostics, and the numerical sources for the tables are retained with the accompanying reproducibility files.

\section{Additional experimental results}
\label{app:additional-experimental-results}

All primary figures and tables retain all $30$ synthetic replicates or all five real-data splits.
Predictive entries use mean (sample SD), while audit fractions, interval endpoints, removal counts, and removal-quality measures are averaged separately.
The accompanying machine-readable tables retain the full $\gamma$ grids, both real-data injection conditions, and all method-level metrics.

\subsection{Synthetic sensitivity to \texorpdfstring{$\gamma$ and $\phi$}{gamma and phi}}
\label{app:synthetic-sensitivity}

Figures~\ref{fig:app-density-gamma} and~\ref{fig:app-regression-gamma} show that the effect of $\gamma$ depends on contamination and task.
Small $\gamma$ can leave contamination bias; larger $\gamma$ can incur efficiency or scale-estimation costs.
The experiments do not choose $\gamma$ automatically.
Figure~\ref{fig:app-phi-sensitivity} compares the fixed saturation specifications with DPD at the main $\gamma$ values.
Their close synthetic results do not establish equivalence under misspecification.
Figure~\ref{fig:app-ranking-only} separates ranking from budget inference: accurate ranking need not estimate $\epsilon$.

\begin{figure}[p]
\centering
\includegraphics[width=\textwidth]{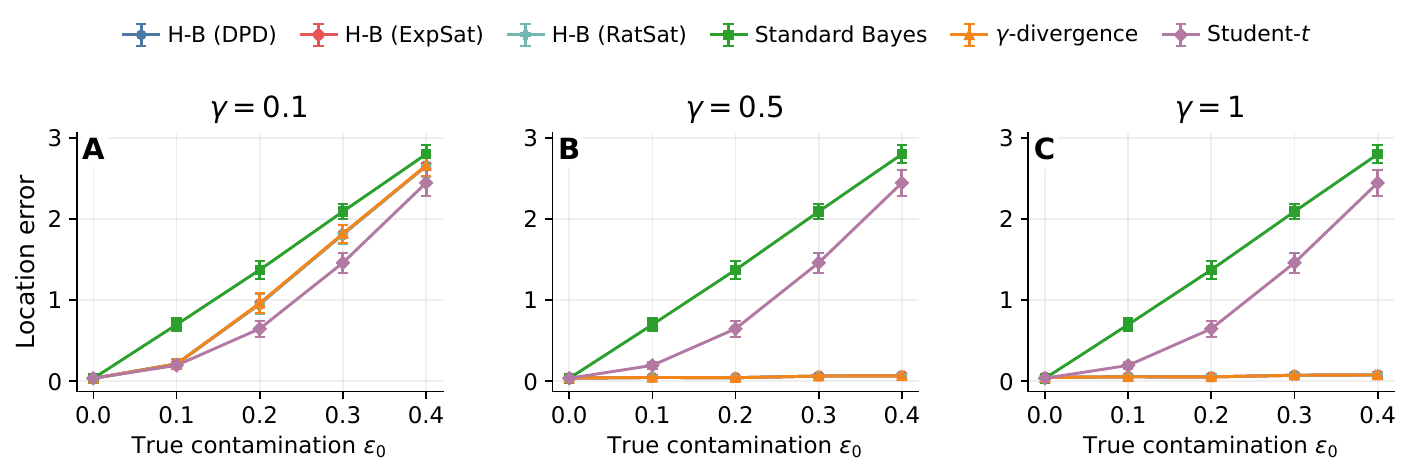}
\includegraphics[width=\textwidth]{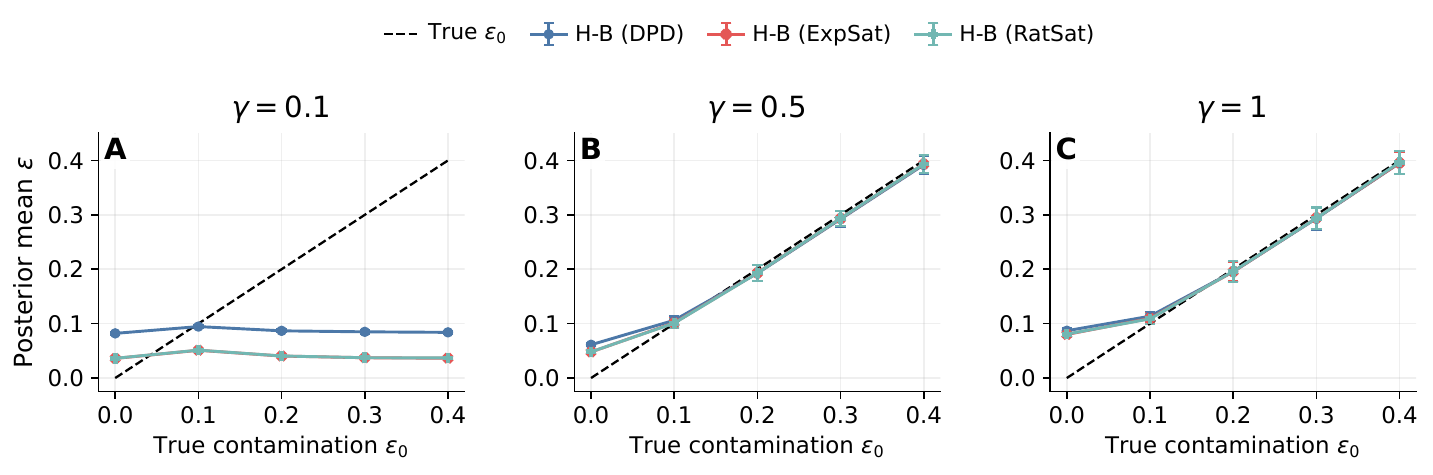}
\includegraphics[width=\textwidth]{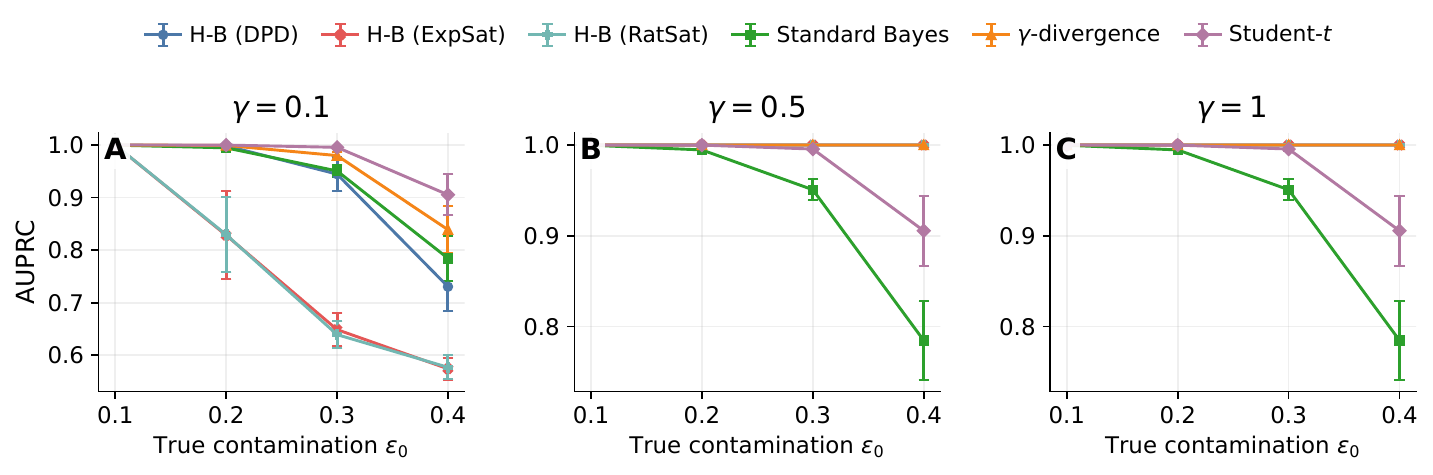}
\caption{Location--scale sensitivity to $\gamma$, using all $30$ replicates per condition. Rows show posterior-mean location error, posterior-mean audit fraction, and AUPRC.
Columns correspond to $\gamma=0.1,0.5,1.0$. Error bars are between-dataset SDs.
Only the shifted-Gaussian contamination setting is used. Ranking against injected labels is undefined at zero contamination.}
\label{fig:app-density-gamma}
\end{figure}
\begin{figure}[p]
\centering
\includegraphics[width=\textwidth]{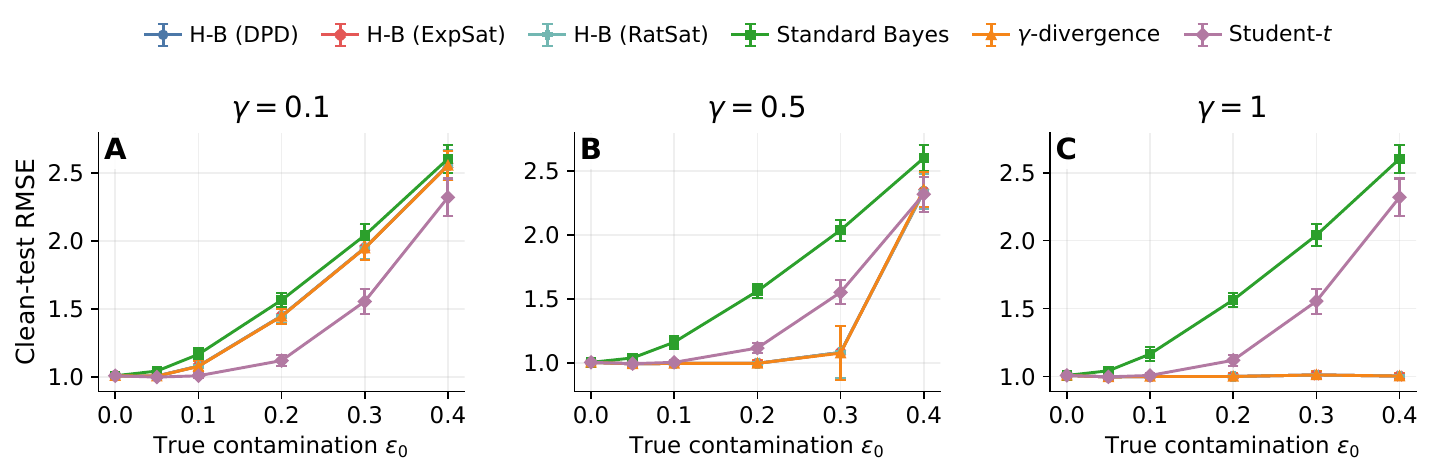}
\includegraphics[width=\textwidth]{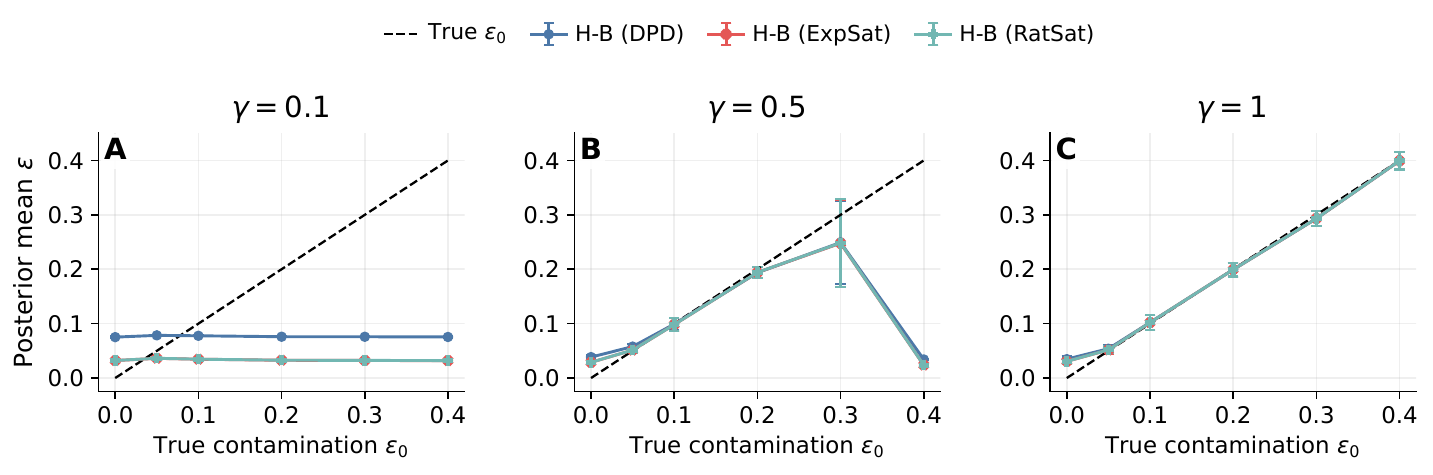}
\includegraphics[width=\textwidth]{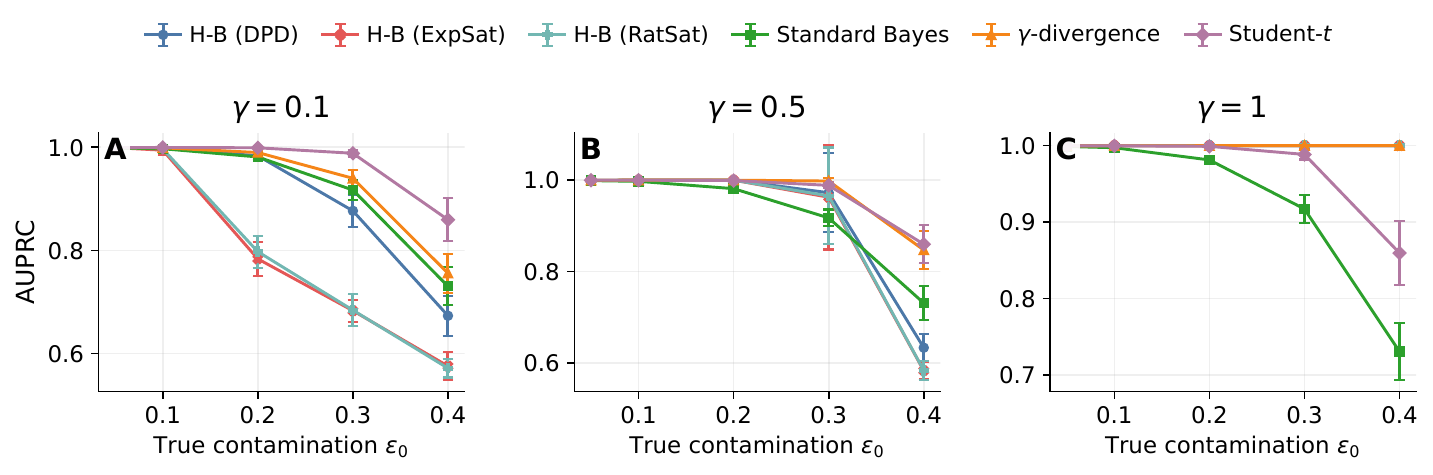}
\caption{Synthetic regression sensitivity to $\gamma$, using all $30$ replicates per condition. Rows show clean-test RMSE, posterior-mean audit fraction, and AUPRC, with columns $\gamma=0.1,0.5,1.0$.
Response replacement shifts the intercept by $6$ and leaves the covariates unchanged.
Error bars are between-dataset SDs; the identity line refers to audit-fraction recovery.}
\label{fig:app-regression-gamma}
\end{figure}
\begin{figure}[p]
\centering
\includegraphics[width=\textwidth]{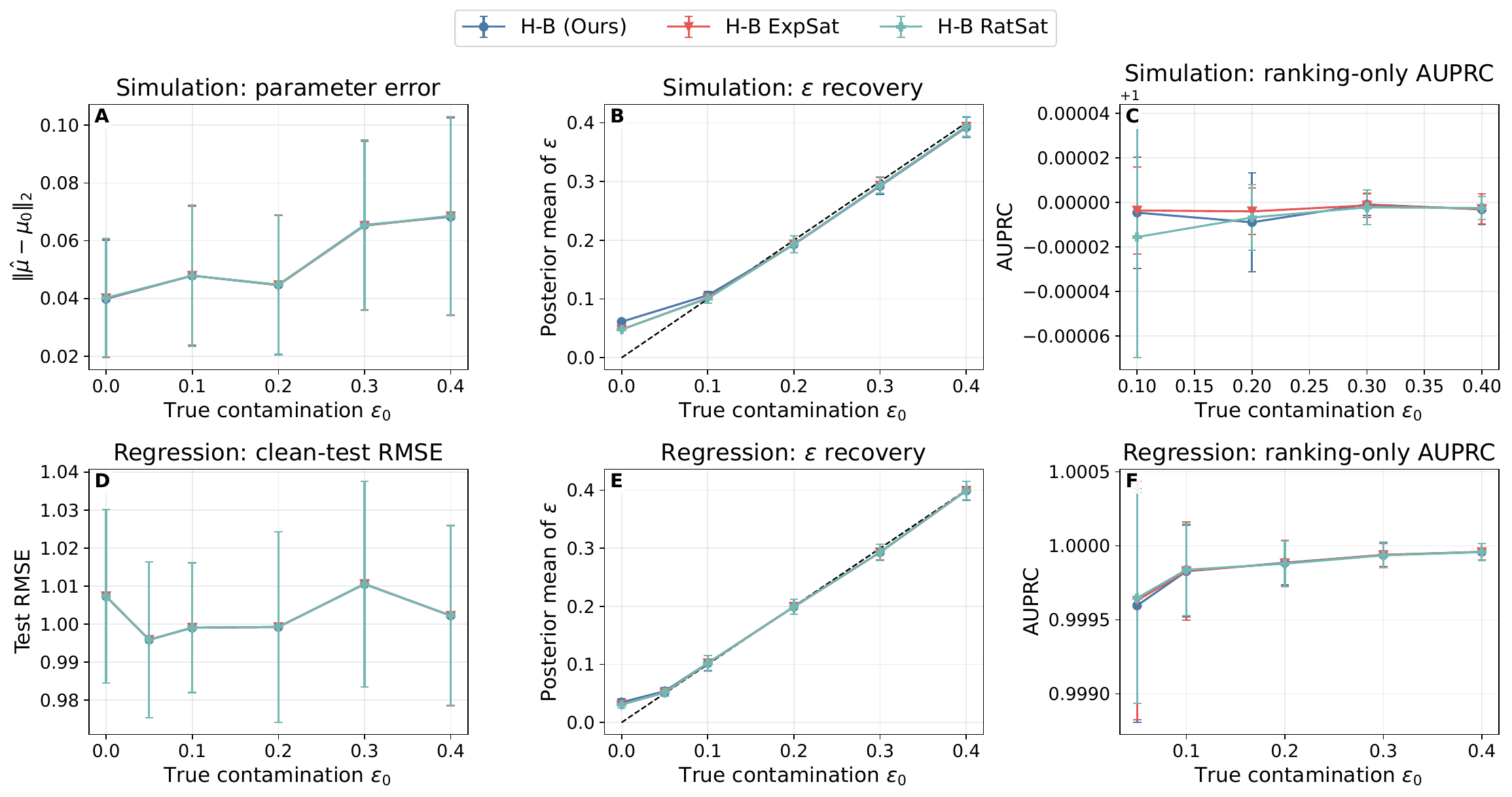}
\caption{Sensitivity to $\phi$: DPD, ExpSatDPD with $\kappa=0.5$, and RatSatDPD with $\kappa=0.25$.
The upper row is location--scale estimation at $\gamma=0.5$; the lower row is regression at $\gamma=1$.
Columns show parameter or prediction error, audit-fraction recovery, and AUPRC.
All $30$ replicates are included, with SD error bars; the AUPRC panels zoom in near one.}
\label{fig:app-phi-sensitivity}
\end{figure}
\begin{figure}[p]
\centering
\includegraphics[width=\textwidth]{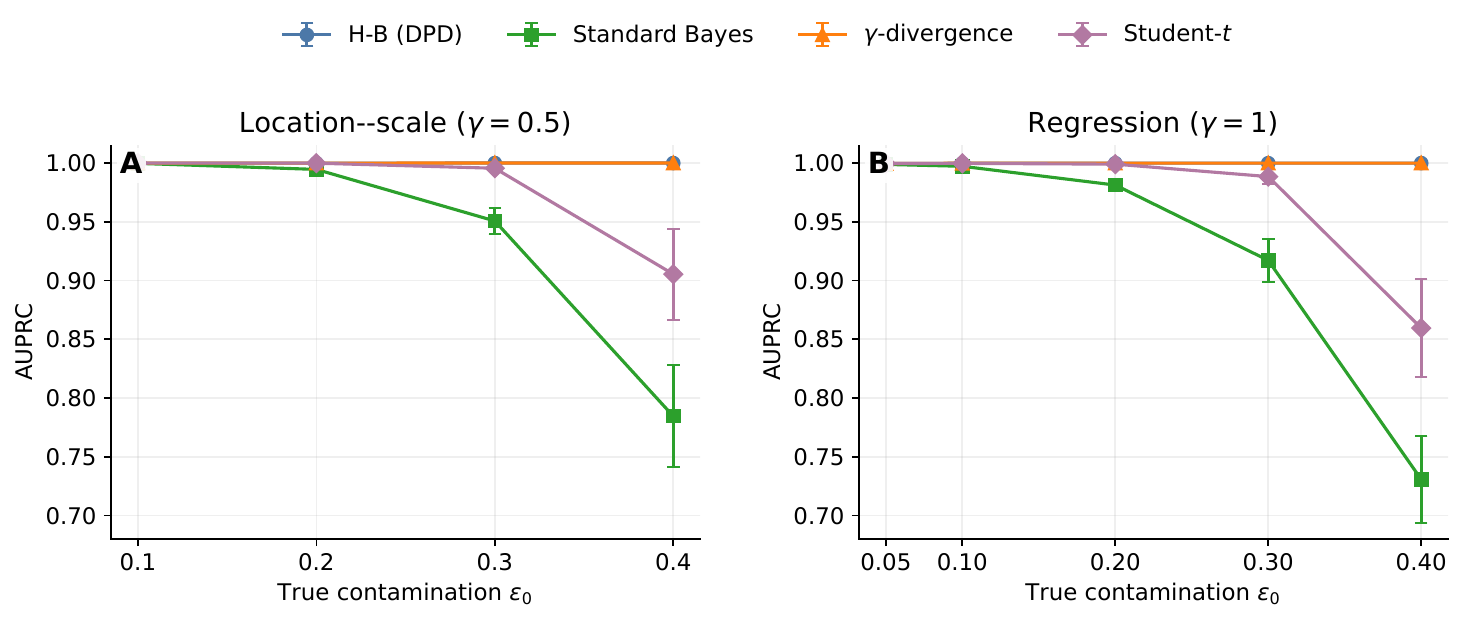}
\caption{Ranking comparison at the main $\gamma$ values, with all $30$ replicates. H-B uses FoD, Standard Bayes and GammaDiv use Gaussian predictive-density ranking, and StudentT uses Student-$t$ predictive-density ranking.
The baseline rankings do not infer an audit budget. Error bars are between-dataset SDs.
}
\label{fig:app-ranking-only}
\end{figure}

\clearpage
\subsection{Full real-data predictive and audit tables}
\label{app:full-realdata-tables}

The following predictive tables cover all six methods and all three $\gamma$ values, with and without injected response contamination.
Their MC-NLPD columns follow the common Gaussian predictive construction specified in Section~\ref{app:realdata-benchmark}, including the stated qualification for StudentT.
The audit table reports all three H-B specifications in both injection conditions.
Positive audit fractions in the original data reflect incompatibility with the working model, rather than a known rate of injected errors.

\begingroup
\footnotesize\setstretch{1.0}\setlength{\tabcolsep}{4pt}
\begin{longtable}{@{}llccc@{}}
\caption{Initial predictive performance under $20\%$ response injection.
Entries are mean (SD) over all five splits. All six methods are shown for each run's $\gamma$; NLL and StudentT do not use $\gamma$ in their likelihood target.
MC-NLPD uses the common Gaussian predictive construction defined in Section~\ref{app:realdata-benchmark}, including for parameters fitted by StudentT; it is not the Student-$t$ posterior-predictive score.
Bold: lowest displayed means within each dataset and $\gamma$; all ties included.}\label{tab:app-realdata-predictive-before-cleaning}\\
\toprule
$\gamma$ & Method & RMSE & MAE & MC-NLPD \\
\midrule
\endfirsthead
\multicolumn{5}{l}{\tablename~\thetable\ (continued)}\\
\toprule
$\gamma$ & Method & RMSE & MAE & MC-NLPD \\
\midrule
\endhead
\midrule\multicolumn{5}{r}{Continued on next page}\\
\endfoot
\bottomrule
\endlastfoot
\multicolumn{5}{@{}l}{\textbf{\texttt{cal-housing}}}\\*
0.1 & Standard Bayes & $1.5690\;(0.0262)$ & $1.4547\;(0.0273)$ & $2.134\;(0.009)$ \\
0.1 & H-B (DPD) & $1.4807\;(0.2828)$ & $1.2474\;(0.0302)$ & $2.075\;(0.010)$ \\
0.1 & $\gamma$-divergence & $1.4901\;(0.2800)$ & $1.2534\;(0.0308)$ & $2.074\;(0.010)$ \\
0.1 & Student-$t$ & $\mathbf{1.0376}\;(0.5575)$ & $\mathbf{0.6256}\;(0.0177)$ & $\mathbf{1.562}\;(0.151)$ \\
0.1 & H-B (ExpSat) & $1.4765\;(0.2738)$ & $1.2467\;(0.0290)$ & $2.074\;(0.010)$ \\
0.1 & H-B (RatSat) & $1.4786\;(0.2799)$ & $1.2463\;(0.0299)$ & $2.073\;(0.010)$ \\
\addlinespace
0.5 & Standard Bayes & $1.5691\;(0.0262)$ & $1.4549\;(0.0274)$ & $2.134\;(0.009)$ \\
0.5 & H-B (DPD) & $0.9661\;(0.5291)$ & $\mathbf{0.4980}\;(0.0094)$ & $2.163\;(2.025)$ \\
0.5 & $\gamma$-divergence & $\mathbf{0.9659}\;(0.5286)$ & $\mathbf{0.4980}\;(0.0094)$ & $2.141\;(1.971)$ \\
0.5 & Student-$t$ & $1.0378\;(0.5578)$ & $0.6256\;(0.0177)$ & $\mathbf{1.566}\;(0.158)$ \\
0.5 & H-B (ExpSat) & $0.9661\;(0.5291)$ & $\mathbf{0.4980}\;(0.0094)$ & $2.119\;(1.928)$ \\
0.5 & H-B (RatSat) & $\mathbf{0.9659}\;(0.5287)$ & $\mathbf{0.4980}\;(0.0094)$ & $2.137\;(1.962)$ \\
\addlinespace
1 & Standard Bayes & $1.5690\;(0.0261)$ & $1.4547\;(0.0272)$ & $2.134\;(0.009)$ \\
1 & H-B (DPD) & $0.9224\;(0.4159)$ & $\mathbf{0.5003}\;(0.0098)$ & $2.374\;(2.060)$ \\
1 & $\gamma$-divergence & $0.9223\;(0.4156)$ & $\mathbf{0.5003}\;(0.0098)$ & $2.204\;(1.720)$ \\
1 & Student-$t$ & $1.0377\;(0.5577)$ & $0.6257\;(0.0177)$ & $\mathbf{1.556}\;(0.135)$ \\
1 & H-B (ExpSat) & $0.9223\;(0.4157)$ & $\mathbf{0.5003}\;(0.0098)$ & $2.361\;(2.031)$ \\
1 & H-B (RatSat) & $\mathbf{0.9221}\;(0.4153)$ & $\mathbf{0.5003}\;(0.0098)$ & $2.363\;(2.034)$ \\
\addlinespace
\multicolumn{5}{@{}l}{\textbf{\texttt{kin8nm}}}\\*
0.1 & Standard Bayes & $0.3837\;(0.0103)$ & $0.3283\;(0.0119)$ & $0.687\;(0.016)$ \\
0.1 & H-B (DPD) & $0.3431\;(0.0100)$ & $0.2841\;(0.0114)$ & $0.643\;(0.018)$ \\
0.1 & $\gamma$-divergence & $0.3428\;(0.0100)$ & $0.2838\;(0.0115)$ & $0.634\;(0.018)$ \\
0.1 & Student-$t$ & $\mathbf{0.2337}\;(0.0049)$ & $\mathbf{0.1801}\;(0.0036)$ & $\mathbf{0.188}\;(0.023)$ \\
0.1 & H-B (ExpSat) & $0.3431\;(0.0097)$ & $0.2841\;(0.0111)$ & $0.641\;(0.017)$ \\
0.1 & H-B (RatSat) & $0.3429\;(0.0099)$ & $0.2838\;(0.0114)$ & $0.641\;(0.018)$ \\
\addlinespace
0.5 & Standard Bayes & $0.3837\;(0.0103)$ & $0.3284\;(0.0119)$ & $0.687\;(0.016)$ \\
0.5 & H-B (DPD) & $\mathbf{0.2063}\;(0.0033)$ & $\mathbf{0.1625}\;(0.0022)$ & $\mathbf{-0.157}\;(0.018)$ \\
0.5 & $\gamma$-divergence & $\mathbf{0.2063}\;(0.0033)$ & $\mathbf{0.1625}\;(0.0022)$ & $\mathbf{-0.157}\;(0.018)$ \\
0.5 & Student-$t$ & $0.2337\;(0.0049)$ & $0.1801\;(0.0036)$ & $0.187\;(0.023)$ \\
0.5 & H-B (ExpSat) & $\mathbf{0.2063}\;(0.0033)$ & $\mathbf{0.1625}\;(0.0022)$ & $\mathbf{-0.157}\;(0.018)$ \\
0.5 & H-B (RatSat) & $\mathbf{0.2063}\;(0.0033)$ & $\mathbf{0.1625}\;(0.0022)$ & $\mathbf{-0.157}\;(0.018)$ \\
\addlinespace
1 & Standard Bayes & $0.3836\;(0.0103)$ & $0.3283\;(0.0119)$ & $0.687\;(0.016)$ \\
1 & H-B (DPD) & $\mathbf{0.2106}\;(0.0036)$ & $\mathbf{0.1634}\;(0.0021)$ & $-0.125\;(0.023)$ \\
1 & $\gamma$-divergence & $\mathbf{0.2106}\;(0.0036)$ & $\mathbf{0.1634}\;(0.0021)$ & $\mathbf{-0.127}\;(0.023)$ \\
1 & Student-$t$ & $0.2337\;(0.0049)$ & $0.1801\;(0.0036)$ & $0.188\;(0.023)$ \\
1 & H-B (ExpSat) & $\mathbf{0.2106}\;(0.0036)$ & $\mathbf{0.1634}\;(0.0021)$ & $-0.125\;(0.023)$ \\
1 & H-B (RatSat) & $\mathbf{0.2106}\;(0.0036)$ & $\mathbf{0.1634}\;(0.0021)$ & $-0.125\;(0.023)$ \\
\addlinespace
\multicolumn{5}{@{}l}{\textbf{\texttt{pumadyn-32fh}}}\\*
0.1 & Standard Bayes & $0.0452\;(0.0012)$ & $0.0387\;(0.0011)$ & $-1.452\;(0.012)$ \\
0.1 & H-B (DPD) & $0.0412\;(0.0012)$ & $0.0347\;(0.0011)$ & $-1.459\;(0.014)$ \\
0.1 & $\gamma$-divergence & $0.0409\;(0.0012)$ & $0.0345\;(0.0011)$ & $-1.498\;(0.014)$ \\
0.1 & Student-$t$ & $\mathbf{0.0297}\;(0.0008)$ & $\mathbf{0.0237}\;(0.0006)$ & $\mathbf{-1.882}\;(0.020)$ \\
0.1 & H-B (ExpSat) & $0.0412\;(0.0011)$ & $0.0347\;(0.0011)$ & $-1.461\;(0.014)$ \\
0.1 & H-B (RatSat) & $0.0411\;(0.0012)$ & $0.0347\;(0.0011)$ & $-1.461\;(0.014)$ \\
\addlinespace
0.5 & Standard Bayes & $0.0452\;(0.0012)$ & $0.0387\;(0.0011)$ & $-1.452\;(0.012)$ \\
0.5 & H-B (DPD) & $\mathbf{0.0270}\;(0.0004)$ & $\mathbf{0.0207}\;(0.0001)$ & $\mathbf{-2.193}\;(0.015)$ \\
0.5 & $\gamma$-divergence & $\mathbf{0.0270}\;(0.0004)$ & $\mathbf{0.0207}\;(0.0001)$ & $\mathbf{-2.193}\;(0.015)$ \\
0.5 & Student-$t$ & $0.0297\;(0.0008)$ & $0.0237\;(0.0006)$ & $-1.882\;(0.020)$ \\
0.5 & H-B (ExpSat) & $\mathbf{0.0270}\;(0.0004)$ & $\mathbf{0.0207}\;(0.0001)$ & $\mathbf{-2.193}\;(0.015)$ \\
0.5 & H-B (RatSat) & $\mathbf{0.0270}\;(0.0004)$ & $\mathbf{0.0207}\;(0.0001)$ & $\mathbf{-2.193}\;(0.015)$ \\
\addlinespace
1 & Standard Bayes & $0.0452\;(0.0012)$ & $0.0387\;(0.0011)$ & $-1.452\;(0.012)$ \\
1 & H-B (DPD) & $0.0310\;(0.0009)$ & $0.0221\;(0.0004)$ & $-1.649\;(0.144)$ \\
1 & $\gamma$-divergence & $0.0308\;(0.0009)$ & $\mathbf{0.0220}\;(0.0004)$ & $-1.751\;(0.128)$ \\
1 & Student-$t$ & $\mathbf{0.0297}\;(0.0008)$ & $0.0237\;(0.0006)$ & $\mathbf{-1.882}\;(0.020)$ \\
1 & H-B (ExpSat) & $0.0310\;(0.0009)$ & $0.0221\;(0.0004)$ & $-1.650\;(0.144)$ \\
1 & H-B (RatSat) & $0.0310\;(0.0009)$ & $0.0221\;(0.0004)$ & $-1.649\;(0.143)$ \\
\addlinespace
\end{longtable}
\endgroup

\begingroup
\footnotesize\setstretch{1.0}\setlength{\tabcolsep}{4pt}
\begin{longtable}{@{}llccc@{}}
\caption{Initial predictive performance without injected contamination.
Entries are mean (SD) over all five splits. All six methods are shown for each run's $\gamma$; NLL and StudentT do not use $\gamma$ in their likelihood target.
MC-NLPD uses the common Gaussian predictive construction defined in Section~\ref{app:realdata-benchmark}, including for parameters fitted by StudentT; it is not the Student-$t$ posterior-predictive score.
Bold: lowest displayed means within each dataset and $\gamma$; all ties included.}\label{tab:app-realdata-predictive-no-injection}\\
\toprule
$\gamma$ & Method & RMSE & MAE & MC-NLPD \\
\midrule
\endfirsthead
\multicolumn{5}{l}{\tablename~\thetable\ (continued)}\\
\toprule
$\gamma$ & Method & RMSE & MAE & MC-NLPD \\
\midrule
\endhead
\midrule\multicolumn{5}{r}{Continued on next page}\\
\endfoot
\bottomrule
\endlastfoot
\multicolumn{5}{@{}l}{\textbf{\texttt{cal-housing}}}\\*
0.1 & Standard Bayes & $\mathbf{0.7218}\;(0.0091)$ & $0.5313\;(0.0053)$ & $\mathbf{1.093}\;(0.012)$ \\
0.1 & H-B (DPD) & $1.0703\;(0.7546)$ & $0.5058\;(0.0109)$ & $1.701\;(1.366)$ \\
0.1 & $\gamma$-divergence & $1.0689\;(0.7514)$ & $0.5058\;(0.0109)$ & $2.182\;(2.371)$ \\
0.1 & Student-$t$ & $1.0056\;(0.6227)$ & $\mathbf{0.4992}\;(0.0095)$ & $2.651\;(3.002)$ \\
0.1 & H-B (ExpSat) & $1.0690\;(0.7515)$ & $0.5058\;(0.0109)$ & $1.717\;(1.405)$ \\
0.1 & H-B (RatSat) & $1.0682\;(0.7501)$ & $0.5057\;(0.0109)$ & $1.715\;(1.398)$ \\
\addlinespace
0.5 & Standard Bayes & $\mathbf{0.7218}\;(0.0091)$ & $0.5313\;(0.0053)$ & $\mathbf{1.093}\;(0.012)$ \\
0.5 & H-B (DPD) & $0.9742\;(0.5477)$ & $\mathbf{0.4980}\;(0.0095)$ & $2.141\;(1.978)$ \\
0.5 & $\gamma$-divergence & $0.9739\;(0.5471)$ & $\mathbf{0.4980}\;(0.0095)$ & $2.199\;(2.090)$ \\
0.5 & Student-$t$ & $1.0056\;(0.6226)$ & $0.4992\;(0.0095)$ & $2.595\;(2.880)$ \\
0.5 & H-B (ExpSat) & $0.9737\;(0.5467)$ & $\mathbf{0.4980}\;(0.0095)$ & $2.143\;(1.981)$ \\
0.5 & H-B (RatSat) & $0.9739\;(0.5472)$ & $\mathbf{0.4980}\;(0.0095)$ & $2.134\;(1.962)$ \\
\addlinespace
1 & Standard Bayes & $\mathbf{0.7218}\;(0.0091)$ & $0.5313\;(0.0053)$ & $\mathbf{1.093}\;(0.012)$ \\
1 & H-B (DPD) & $0.9269\;(0.4253)$ & $0.5004\;(0.0096)$ & $2.283\;(1.860)$ \\
1 & $\gamma$-divergence & $0.9264\;(0.4240)$ & $0.5004\;(0.0096)$ & $2.253\;(1.812)$ \\
1 & Student-$t$ & $1.0057\;(0.6228)$ & $\mathbf{0.4992}\;(0.0095)$ & $2.615\;(2.922)$ \\
1 & H-B (ExpSat) & $0.9267\;(0.4248)$ & $0.5004\;(0.0096)$ & $2.314\;(1.929)$ \\
1 & H-B (RatSat) & $0.9267\;(0.4249)$ & $0.5004\;(0.0096)$ & $2.295\;(1.887)$ \\
\addlinespace
\multicolumn{5}{@{}l}{\textbf{\texttt{kin8nm}}}\\*
0.1 & Standard Bayes & $\mathbf{0.2036}\;(0.0028)$ & $0.1642\;(0.0026)$ & $-0.172\;(0.014)$ \\
0.1 & H-B (DPD) & $0.2037\;(0.0029)$ & $0.1635\;(0.0025)$ & $\mathbf{-0.173}\;(0.014)$ \\
0.1 & $\gamma$-divergence & $0.2037\;(0.0028)$ & $0.1635\;(0.0025)$ & $-0.172\;(0.014)$ \\
0.1 & Student-$t$ & $0.2049\;(0.0029)$ & $\mathbf{0.1626}\;(0.0023)$ & $-0.114\;(0.023)$ \\
0.1 & H-B (ExpSat) & $0.2037\;(0.0029)$ & $0.1635\;(0.0025)$ & $-0.172\;(0.014)$ \\
0.1 & H-B (RatSat) & $0.2037\;(0.0028)$ & $0.1635\;(0.0025)$ & $\mathbf{-0.173}\;(0.014)$ \\
\addlinespace
0.5 & Standard Bayes & $\mathbf{0.2036}\;(0.0028)$ & $0.1642\;(0.0026)$ & $\mathbf{-0.172}\;(0.014)$ \\
0.5 & H-B (DPD) & $0.2062\;(0.0029)$ & $\mathbf{0.1624}\;(0.0022)$ & $-0.157\;(0.016)$ \\
0.5 & $\gamma$-divergence & $0.2061\;(0.0029)$ & $0.1625\;(0.0022)$ & $-0.158\;(0.016)$ \\
0.5 & Student-$t$ & $0.2049\;(0.0029)$ & $0.1626\;(0.0023)$ & $-0.114\;(0.023)$ \\
0.5 & H-B (ExpSat) & $0.2062\;(0.0029)$ & $\mathbf{0.1624}\;(0.0022)$ & $-0.157\;(0.016)$ \\
0.5 & H-B (RatSat) & $0.2062\;(0.0029)$ & $\mathbf{0.1624}\;(0.0022)$ & $-0.157\;(0.016)$ \\
\addlinespace
1 & Standard Bayes & $\mathbf{0.2036}\;(0.0028)$ & $0.1642\;(0.0026)$ & $\mathbf{-0.172}\;(0.014)$ \\
1 & H-B (DPD) & $0.2105\;(0.0029)$ & $0.1633\;(0.0020)$ & $-0.127\;(0.017)$ \\
1 & $\gamma$-divergence & $0.2105\;(0.0029)$ & $0.1633\;(0.0020)$ & $-0.128\;(0.017)$ \\
1 & Student-$t$ & $0.2049\;(0.0029)$ & $\mathbf{0.1626}\;(0.0023)$ & $-0.115\;(0.023)$ \\
1 & H-B (ExpSat) & $0.2105\;(0.0029)$ & $0.1633\;(0.0020)$ & $-0.127\;(0.017)$ \\
1 & H-B (RatSat) & $0.2105\;(0.0029)$ & $0.1633\;(0.0020)$ & $-0.127\;(0.017)$ \\
\addlinespace
\multicolumn{5}{@{}l}{\textbf{\texttt{pumadyn-32fh}}}\\*
0.1 & Standard Bayes & $\mathbf{0.0265}\;(0.0003)$ & $0.0208\;(0.0001)$ & $\mathbf{-2.211}\;(0.011)$ \\
0.1 & H-B (DPD) & $\mathbf{0.0265}\;(0.0003)$ & $\mathbf{0.0207}\;(0.0001)$ & $-2.208\;(0.010)$ \\
0.1 & $\gamma$-divergence & $\mathbf{0.0265}\;(0.0003)$ & $\mathbf{0.0207}\;(0.0001)$ & $\mathbf{-2.211}\;(0.011)$ \\
0.1 & Student-$t$ & $0.0267\;(0.0003)$ & $\mathbf{0.0207}\;(0.0001)$ & $-2.160\;(0.018)$ \\
0.1 & H-B (ExpSat) & $\mathbf{0.0265}\;(0.0003)$ & $\mathbf{0.0207}\;(0.0001)$ & $-2.208\;(0.010)$ \\
0.1 & H-B (RatSat) & $\mathbf{0.0265}\;(0.0003)$ & $\mathbf{0.0207}\;(0.0001)$ & $-2.208\;(0.010)$ \\
\addlinespace
0.5 & Standard Bayes & $\mathbf{0.0265}\;(0.0003)$ & $0.0208\;(0.0001)$ & $\mathbf{-2.211}\;(0.011)$ \\
0.5 & H-B (DPD) & $0.0269\;(0.0003)$ & $\mathbf{0.0207}\;(0.0001)$ & $-2.194\;(0.012)$ \\
0.5 & $\gamma$-divergence & $0.0270\;(0.0003)$ & $\mathbf{0.0207}\;(0.0001)$ & $-2.193\;(0.012)$ \\
0.5 & Student-$t$ & $0.0267\;(0.0003)$ & $\mathbf{0.0207}\;(0.0001)$ & $-2.160\;(0.018)$ \\
0.5 & H-B (ExpSat) & $0.0269\;(0.0003)$ & $\mathbf{0.0207}\;(0.0001)$ & $-2.194\;(0.012)$ \\
0.5 & H-B (RatSat) & $0.0269\;(0.0003)$ & $\mathbf{0.0207}\;(0.0001)$ & $-2.194\;(0.012)$ \\
\addlinespace
1 & Standard Bayes & $\mathbf{0.0265}\;(0.0003)$ & $0.0208\;(0.0001)$ & $\mathbf{-2.211}\;(0.011)$ \\
1 & H-B (DPD) & $0.0307\;(0.0004)$ & $0.0219\;(0.0002)$ & $-1.746\;(0.055)$ \\
1 & $\gamma$-divergence & $0.0307\;(0.0004)$ & $0.0219\;(0.0002)$ & $-1.763\;(0.054)$ \\
1 & Student-$t$ & $0.0267\;(0.0003)$ & $\mathbf{0.0207}\;(0.0001)$ & $-2.160\;(0.018)$ \\
1 & H-B (ExpSat) & $0.0307\;(0.0004)$ & $0.0219\;(0.0002)$ & $-1.746\;(0.055)$ \\
1 & H-B (RatSat) & $0.0307\;(0.0004)$ & $0.0219\;(0.0002)$ & $-1.746\;(0.054)$ \\
\addlinespace
\end{longtable}
\endgroup

\begingroup
\footnotesize\setstretch{1.0}\setlength{\tabcolsep}{4pt}
\begin{longtable}{@{}llccc@{}}
\caption{H-B audit summaries, with all five splits. Posterior means and endpoints of $95\%$ credible intervals (CrIs) are averaged separately; the displayed interval is not a pooled posterior credible interval.
The three counts are means of the operational lower, central, and upper cleaning budgets, with integer-count quantiles and rounding as defined in Section~\ref{app:realdata-benchmark}.
FoD AUPRC is averaged over splits and is undefined without injected labels.
Bold: largest displayed AUPRC within each dataset, injection condition, and $\gamma$; all ties included.}\label{tab:app-realdata-audit-fraction}\\
\toprule
$\gamma$ & H-B variant & $\epsilon$: mean [mean CrI] & $(K_{\rm l},K_{\rm c},K_{\rm u})$ & FoD AUPRC \\
\midrule
\endfirsthead
\multicolumn{5}{l}{\tablename~\thetable\ (continued)}\\
\toprule
$\gamma$ & H-B variant & $\epsilon$: mean [mean CrI] & $(K_{\rm l},K_{\rm c},K_{\rm u})$ & FoD AUPRC \\
\midrule
\endhead
\midrule\multicolumn{5}{r}{Continued on next page}\\
\endfoot
\bottomrule
\endlastfoot
\multicolumn{5}{@{}l}{\textbf{\texttt{cal-housing}}, $\epsilon_0=0$}\\*
0.1 & DPD & $0.021\;[0.001,\,0.056]$ & $(21.8,\,351.6,\,931.0)$ & -- \\
0.1 & ExpSat & $0.010\;[0.001,\,0.026]$ & $(7.8,\,160.2,\,417.8)$ & -- \\
0.1 & RatSat & $0.010\;[0.000,\,0.026]$ & $(8.8,\,159.6,\,429.2)$ & -- \\
\addlinespace
0.5 & DPD & $0.050\;[0.029,\,0.071]$ & $(504.2,\,830.6,\,1168.6)$ & -- \\
0.5 & ExpSat & $0.050\;[0.034,\,0.066]$ & $(567.4,\,830.8,\,1092.8)$ & -- \\
0.5 & RatSat & $0.050\;[0.034,\,0.066]$ & $(566.4,\,832.6,\,1097.6)$ & -- \\
\addlinespace
1 & DPD & $0.096\;[0.079,\,0.112]$ & $(1318.4,\,1583.2,\,1852.0)$ & -- \\
1 & ExpSat & $0.096\;[0.081,\,0.110]$ & $(1343.4,\,1584.0,\,1819.0)$ & -- \\
1 & RatSat & $0.096\;[0.081,\,0.111]$ & $(1346.8,\,1585.2,\,1831.0)$ & -- \\
\addlinespace
\multicolumn{5}{@{}l}{\textbf{\texttt{cal-housing}}, $\epsilon_0=0.2$}\\*
0.1 & DPD & $0.020\;[0.001,\,0.056]$ & $(14.6,\,343.6,\,891.6)$ & $\mathbf{0.546}$ \\
0.1 & ExpSat & $0.009\;[0.000,\,0.025]$ & $(5.6,\,149.6,\,409.0)$ & $0.357$ \\
0.1 & RatSat & $0.009\;[0.000,\,0.025]$ & $(6.6,\,146.0,\,401.4)$ & $0.353$ \\
\addlinespace
0.5 & DPD & $0.240\;[0.224,\,0.256]$ & $(3707.8,\,3968.2,\,4222.4)$ & $\mathbf{0.924}$ \\
0.5 & ExpSat & $0.240\;[0.228,\,0.252]$ & $(3768.8,\,3965.6,\,4161.0)$ & $0.902$ \\
0.5 & RatSat & $0.240\;[0.228,\,0.252]$ & $(3768.4,\,3965.8,\,4165.0)$ & $0.902$ \\
\addlinespace
1 & DPD & $0.277\;[0.267,\,0.287]$ & $(4404.6,\,4570.0,\,4734.6)$ & $\mathbf{0.769}$ \\
1 & ExpSat & $0.277\;[0.268,\,0.286]$ & $(4417.8,\,4570.8,\,4724.4)$ & $0.763$ \\
1 & RatSat & $0.277\;[0.268,\,0.286]$ & $(4423.2,\,4572.4,\,4725.4)$ & $0.764$ \\
\addlinespace
\multicolumn{5}{@{}l}{\textbf{\texttt{kin8nm}}, $\epsilon_0=0$}\\*
0.1 & DPD & $0.031\;[0.001,\,0.085]$ & $(9.8,\,204.0,\,552.0)$ & -- \\
0.1 & ExpSat & $0.013\;[0.001,\,0.036]$ & $(3.8,\,84.6,\,242.6)$ & -- \\
0.1 & RatSat & $0.013\;[0.001,\,0.036]$ & $(3.4,\,84.6,\,245.0)$ & -- \\
\addlinespace
0.5 & DPD & $0.019\;[0.001,\,0.049]$ & $(7.4,\,120.8,\,302.8)$ & -- \\
0.5 & ExpSat & $0.015\;[0.001,\,0.038]$ & $(5.0,\,95.8,\,248.2)$ & -- \\
0.5 & RatSat & $0.015\;[0.001,\,0.038]$ & $(6.0,\,99.8,\,264.2)$ & -- \\
\addlinespace
1 & DPD & $0.026\;[0.002,\,0.057]$ & $(14.4,\,166.0,\,367.4)$ & -- \\
1 & ExpSat & $0.025\;[0.003,\,0.053]$ & $(15.2,\,160.6,\,342.2)$ & -- \\
1 & RatSat & $0.025\;[0.003,\,0.054]$ & $(18.2,\,162.8,\,351.6)$ & -- \\
\addlinespace
\multicolumn{5}{@{}l}{\textbf{\texttt{kin8nm}}, $\epsilon_0=0.2$}\\*
0.1 & DPD & $0.032\;[0.001,\,0.088]$ & $(10.6,\,210.6,\,573.4)$ & $\mathbf{0.720}$ \\
0.1 & ExpSat & $0.014\;[0.001,\,0.037]$ & $(4.4,\,90.2,\,253.4)$ & $0.482$ \\
0.1 & RatSat & $0.014\;[0.001,\,0.038]$ & $(4.0,\,88.8,\,240.2)$ & $0.475$ \\
\addlinespace
0.5 & DPD & $0.206\;[0.177,\,0.234]$ & $(1161.0,\,1350.8,\,1532.8)$ & $\mathbf{0.999}$ \\
0.5 & ExpSat & $0.206\;[0.184,\,0.227]$ & $(1199.8,\,1347.6,\,1490.4)$ & $\mathbf{0.999}$ \\
0.5 & RatSat & $0.206\;[0.184,\,0.227]$ & $(1211.4,\,1346.8,\,1485.2)$ & $\mathbf{0.999}$ \\
\addlinespace
1 & DPD & $0.219\;[0.199,\,0.240]$ & $(1301.2,\,1435.8,\,1572.6)$ & $\mathbf{0.999}$ \\
1 & ExpSat & $0.219\;[0.201,\,0.238]$ & $(1312.0,\,1437.6,\,1562.0)$ & $0.998$ \\
1 & RatSat & $0.220\;[0.201,\,0.238]$ & $(1313.8,\,1439.0,\,1560.4)$ & $\mathbf{0.999}$ \\
\addlinespace
\multicolumn{5}{@{}l}{\textbf{\texttt{pumadyn-32fh}}, $\epsilon_0=0$}\\*
0.1 & DPD & $0.031\;[0.001,\,0.087]$ & $(8.6,\,212.0,\,587.8)$ & -- \\
0.1 & ExpSat & $0.013\;[0.001,\,0.036]$ & $(3.8,\,84.6,\,241.0)$ & -- \\
0.1 & RatSat & $0.013\;[0.001,\,0.037]$ & $(4.0,\,83.0,\,232.6)$ & -- \\
\addlinespace
0.5 & DPD & $0.022\;[0.001,\,0.055]$ & $(8.2,\,144.0,\,361.4)$ & -- \\
0.5 & ExpSat & $0.019\;[0.001,\,0.045]$ & $(7.8,\,123.0,\,294.0)$ & -- \\
0.5 & RatSat & $0.019\;[0.001,\,0.045]$ & $(7.6,\,124.2,\,296.0)$ & -- \\
\addlinespace
1 & DPD & $0.138\;[0.098,\,0.176]$ & $(652.0,\,900.6,\,1150.4)$ & -- \\
1 & ExpSat & $0.138\;[0.102,\,0.173]$ & $(670.0,\,904.6,\,1130.4)$ & -- \\
1 & RatSat & $0.138\;[0.102,\,0.173]$ & $(672.2,\,904.8,\,1127.4)$ & -- \\
\addlinespace
\multicolumn{5}{@{}l}{\textbf{\texttt{pumadyn-32fh}}, $\epsilon_0=0.2$}\\*
0.1 & DPD & $0.032\;[0.001,\,0.089]$ & $(11.4,\,214.0,\,587.2)$ & $\mathbf{0.847}$ \\
0.1 & ExpSat & $0.014\;[0.001,\,0.037]$ & $(3.8,\,88.0,\,238.6)$ & $0.615$ \\
0.1 & RatSat & $0.014\;[0.001,\,0.038]$ & $(3.2,\,88.2,\,237.8)$ & $0.613$ \\
\addlinespace
0.5 & DPD & $0.209\;[0.179,\,0.239]$ & $(1164.0,\,1371.0,\,1576.0)$ & $\mathbf{0.999}$ \\
0.5 & ExpSat & $0.209\;[0.186,\,0.232]$ & $(1220.4,\,1369.0,\,1516.4)$ & $\mathbf{0.999}$ \\
0.5 & RatSat & $0.209\;[0.186,\,0.232]$ & $(1230.8,\,1368.8,\,1516.2)$ & $\mathbf{0.999}$ \\
\addlinespace
1 & DPD & $0.320\;[0.296,\,0.343]$ & $(1940.6,\,2092.8,\,2246.2)$ & $\mathbf{0.720}$ \\
1 & ExpSat & $0.319\;[0.298,\,0.341]$ & $(1954.2,\,2093.4,\,2231.2)$ & $0.715$ \\
1 & RatSat & $0.319\;[0.298,\,0.341]$ & $(1960.0,\,2097.6,\,2242.2)$ & $0.713$ \\
\addlinespace
\end{longtable}
\endgroup

\subsection{Downstream cleaning comparisons}
\label{app:cleaning-comparisons}

The predictive and removal-quality tables use $\gamma=0.5$, matching the main real-data comparisons, and include every cleaning rule.
The same central DPD budget is supplied to the generic, random, and oracle comparisons; each H-B specification supplies its own lower, central, and upper budgets.
Budget matching controls the removal count, not which design points remain.
On \texttt{kin8nm} and \texttt{pumadyn-32fh}, central H-B refits approach the oracle in RMSE and MC-NLPD.
On \texttt{cal-housing}, a large split-specific prediction error affects the aggregate, as examined in the following subsection.
The full machine-readable cleaning tables additionally cover $\gamma=0.1,1.0$ and the no-injection condition.

\begingroup
\footnotesize\setstretch{1.0}\setlength{\tabcolsep}{4pt}
\begin{longtable}{@{}lcccc@{}}
\caption{Downstream Gaussian prediction after cleaning under $20\%$ injection, $\gamma=0.5$. Predictive entries are mean (SD) over all five splits; $K$ is the mean removal count.
Every rule uses the same Gaussian downstream model and unchanged test set. Generic and oracle comparisons use the central DPD budget; each H-B variant uses its own count posterior.
Bold: lowest displayed predictive means within each dataset, excluding the oracle reference; all ties included.}\label{tab:app-realdata-audit-policy}\\
\toprule
Cleaning rule & $K$ & RMSE & MAE & MC-NLPD \\
\midrule
\endfirsthead
\multicolumn{5}{l}{\tablename~\thetable\ (continued)}\\
\toprule
Cleaning rule & $K$ & RMSE & MAE & MC-NLPD \\
\midrule
\endhead
\midrule\multicolumn{5}{r}{Continued on next page}\\
\endfoot
\bottomrule
\endlastfoot
\multicolumn{5}{@{}l}{\textbf{\texttt{cal-housing}}}\\*
No removal & $0.0$ & $1.5691\;(0.0262)$ & $1.4549\;(0.0274)$ & $2.134\;(0.009)$ \\
Oracle (central) & $3968.2$ & $0.7226\;(0.0091)$ & $0.5315\;(0.0057)$ & $1.093\;(0.012)$ \\
DPD lower & $3707.8$ & $1.0274\;(0.6648)$ & $0.5019\;(0.0100)$ & $2.289\;(2.493)$ \\
DPD central & $3968.2$ & $1.0017\;(0.6109)$ & $0.4996\;(0.0095)$ & $2.419\;(2.629)$ \\
DPD upper & $4222.4$ & $\mathbf{0.9776}\;(0.5586)$ & $\mathbf{0.4983}\;(0.0094)$ & $2.464\;(2.560)$ \\
ExpSat lower & $3768.8$ & $1.0241\;(0.6583)$ & $0.5014\;(0.0098)$ & $2.365\;(2.629)$ \\
ExpSat central & $3965.6$ & $1.0006\;(0.6080)$ & $0.4995\;(0.0095)$ & $2.424\;(2.643)$ \\
ExpSat upper & $4161.0$ & $0.9818\;(0.5684)$ & $0.4985\;(0.0093)$ & $2.485\;(2.653)$ \\
RatSat lower & $3768.4$ & $1.0270\;(0.6642)$ & $0.5014\;(0.0100)$ & $2.357\;(2.604)$ \\
RatSat central & $3965.8$ & $1.0016\;(0.6107)$ & $0.4996\;(0.0095)$ & $2.414\;(2.618)$ \\
RatSat upper & $4165.0$ & $0.9806\;(0.5659)$ & $0.4984\;(0.0093)$ & $2.450\;(2.575)$ \\
Random & $3968.2$ & $1.5712\;(0.0310)$ & $1.4555\;(0.0331)$ & $2.134\;(0.012)$ \\
Isolation forest & $3968.2$ & $1.5353\;(0.9100)$ & $0.9601\;(0.0508)$ & $\mathbf{1.951}\;(0.178)$ \\
Local outlier factor & $3968.2$ & $1.7519\;(0.6184)$ & $1.3601\;(0.0281)$ & $2.116\;(0.049)$ \\
One-class SVM & $3968.2$ & $1.4842\;(0.4571)$ & $1.1503\;(0.0300)$ & $1.977\;(0.039)$ \\
\addlinespace
\multicolumn{5}{@{}l}{\textbf{\texttt{kin8nm}}}\\*
No removal & $0.0$ & $0.3837\;(0.0103)$ & $0.3284\;(0.0119)$ & $0.687\;(0.016)$ \\
Oracle (central) & $1350.8$ & $0.2038\;(0.0028)$ & $0.1643\;(0.0025)$ & $-0.172\;(0.014)$ \\
DPD lower & $1161.0$ & $0.2070\;(0.0027)$ & $0.1639\;(0.0018)$ & $-0.106\;(0.016)$ \\
DPD central & $1350.8$ & $\mathbf{0.2037}\;(0.0030)$ & $0.1632\;(0.0023)$ & $-0.170\;(0.017)$ \\
DPD upper & $1532.8$ & $0.2057\;(0.0032)$ & $\mathbf{0.1625}\;(0.0022)$ & $-0.135\;(0.026)$ \\
ExpSat lower & $1199.8$ & $0.2056\;(0.0027)$ & $0.1636\;(0.0020)$ & $-0.134\;(0.008)$ \\
ExpSat central & $1347.6$ & $\mathbf{0.2037}\;(0.0029)$ & $0.1633\;(0.0023)$ & $\mathbf{-0.171}\;(0.016)$ \\
ExpSat upper & $1490.4$ & $0.2053\;(0.0032)$ & $\mathbf{0.1625}\;(0.0022)$ & $-0.145\;(0.023)$ \\
RatSat lower & $1211.4$ & $0.2053\;(0.0027)$ & $0.1635\;(0.0021)$ & $-0.141\;(0.011)$ \\
RatSat central & $1346.8$ & $\mathbf{0.2037}\;(0.0030)$ & $0.1633\;(0.0023)$ & $\mathbf{-0.171}\;(0.016)$ \\
RatSat upper & $1485.2$ & $0.2052\;(0.0032)$ & $\mathbf{0.1625}\;(0.0023)$ & $-0.146\;(0.023)$ \\
Random & $1350.8$ & $0.3858\;(0.0091)$ & $0.3308\;(0.0106)$ & $0.691\;(0.014)$ \\
Isolation forest & $1350.8$ & $0.2807\;(0.0077)$ & $0.2215\;(0.0085)$ & $0.416\;(0.016)$ \\
Local outlier factor & $1350.8$ & $0.3502\;(0.0080)$ & $0.2917\;(0.0092)$ & $0.629\;(0.012)$ \\
One-class SVM & $1350.8$ & $0.3323\;(0.0081)$ & $0.2733\;(0.0092)$ & $0.548\;(0.017)$ \\
\addlinespace
\multicolumn{5}{@{}l}{\textbf{\texttt{pumadyn-32fh}}}\\*
No removal & $0.0$ & $0.0452\;(0.0012)$ & $0.0387\;(0.0011)$ & $-1.452\;(0.012)$ \\
Oracle (central) & $1371.0$ & $0.0265\;(0.0003)$ & $0.0208\;(0.0001)$ & $-2.210\;(0.011)$ \\
DPD lower & $1164.0$ & $0.0268\;(0.0004)$ & $0.0210\;(0.0002)$ & $-2.163\;(0.015)$ \\
DPD central & $1371.0$ & $\mathbf{0.0266}\;(0.0003)$ & $\mathbf{0.0207}\;(0.0001)$ & $\mathbf{-2.208}\;(0.012)$ \\
DPD upper & $1576.0$ & $0.0270\;(0.0004)$ & $\mathbf{0.0207}\;(0.0001)$ & $-2.169\;(0.023)$ \\
ExpSat lower & $1220.4$ & $\mathbf{0.0266}\;(0.0004)$ & $0.0208\;(0.0002)$ & $-2.190\;(0.014)$ \\
ExpSat central & $1369.0$ & $\mathbf{0.0266}\;(0.0003)$ & $\mathbf{0.0207}\;(0.0001)$ & $\mathbf{-2.208}\;(0.012)$ \\
ExpSat upper & $1516.4$ & $0.0268\;(0.0004)$ & $\mathbf{0.0207}\;(0.0001)$ & $-2.185\;(0.018)$ \\
RatSat lower & $1230.8$ & $\mathbf{0.0266}\;(0.0004)$ & $0.0208\;(0.0002)$ & $-2.194\;(0.013)$ \\
RatSat central & $1368.8$ & $\mathbf{0.0266}\;(0.0003)$ & $\mathbf{0.0207}\;(0.0001)$ & $\mathbf{-2.208}\;(0.012)$ \\
RatSat upper & $1516.2$ & $0.0268\;(0.0004)$ & $\mathbf{0.0207}\;(0.0001)$ & $-2.185\;(0.017)$ \\
Random & $1371.0$ & $0.0454\;(0.0010)$ & $0.0388\;(0.0010)$ & $-1.449\;(0.011)$ \\
Isolation forest & $1371.0$ & $0.0399\;(0.0012)$ & $0.0334\;(0.0012)$ & $-1.560\;(0.020)$ \\
Local outlier factor & $1371.0$ & $0.0392\;(0.0009)$ & $0.0326\;(0.0009)$ & $-1.580\;(0.014)$ \\
One-class SVM & $1371.0$ & $0.0413\;(0.0008)$ & $0.0348\;(0.0008)$ & $-1.545\;(0.012)$ \\
\addlinespace
\end{longtable}
\endgroup

\begingroup
\footnotesize\setstretch{1.0}\setlength{\tabcolsep}{4pt}
\begin{longtable}{@{}lccccc@{}}
\caption{Removal quality under $20\%$ injection, $\gamma=0.5$. Counts and quality metrics are means over all five splits.
Precision, recall, and $F_1$ use injected labels; IFR is the fraction of non-injected training rows removed.
The oracle matches the central DPD budget and may also remove non-injected rows when the budget exceeds the realised injected count.
Bold: largest displayed $F_1$ within each dataset, excluding the oracle reference; all ties included.}\label{tab:app-realdata-removal-quality}\\
\toprule
Cleaning rule & $K$ & Precision & Recall & $F_1$ & IFR \\
\midrule
\endfirsthead
\multicolumn{6}{l}{\tablename~\thetable\ (continued)}\\
\toprule
Cleaning rule & $K$ & Precision & Recall & $F_1$ & IFR \\
\midrule
\endhead
\midrule\multicolumn{6}{r}{Continued on next page}\\
\endfoot
\bottomrule
\endlastfoot
\multicolumn{6}{@{}l}{\textbf{\texttt{cal-housing}}}\\*
No removal & $0.0$ & -- & -- & -- & -- \\
Oracle (central) & $3968.2$ & $0.833$ & $1.000$ & $0.909$ & $0.050$ \\
DPD lower & $3707.8$ & $0.891$ & $1.000$ & $\mathbf{0.942}$ & $0.031$ \\
DPD central & $3968.2$ & $0.833$ & $1.000$ & $0.909$ & $0.050$ \\
DPD upper & $4222.4$ & $0.783$ & $1.000$ & $0.878$ & $0.069$ \\
ExpSat lower & $3768.8$ & $0.877$ & $1.000$ & $0.934$ & $0.035$ \\
ExpSat central & $3965.6$ & $0.833$ & $1.000$ & $0.909$ & $0.050$ \\
ExpSat upper & $4161.0$ & $0.794$ & $1.000$ & $0.885$ & $0.065$ \\
RatSat lower & $3768.4$ & $0.877$ & $1.000$ & $0.934$ & $0.035$ \\
RatSat central & $3965.8$ & $0.833$ & $1.000$ & $0.909$ & $0.050$ \\
RatSat upper & $4165.0$ & $0.793$ & $1.000$ & $0.885$ & $0.065$ \\
Random & $3968.2$ & $0.199$ & $0.239$ & $0.218$ & $0.241$ \\
Isolation forest & $3968.2$ & $0.424$ & $0.509$ & $0.462$ & $0.173$ \\
Local outlier factor & $3968.2$ & $0.225$ & $0.270$ & $0.245$ & $0.233$ \\
One-class SVM & $3968.2$ & $0.331$ & $0.397$ & $0.361$ & $0.201$ \\
\addlinespace
\multicolumn{6}{@{}l}{\textbf{\texttt{kin8nm}}}\\*
No removal & $0.0$ & -- & -- & -- & -- \\
Oracle (central) & $1350.8$ & $0.978$ & $1.000$ & $0.989$ & $0.006$ \\
DPD lower & $1161.0$ & $1.000$ & $0.879$ & $0.936$ & $0.000$ \\
DPD central & $1350.8$ & $0.974$ & $0.996$ & $0.985$ & $0.007$ \\
DPD upper & $1532.8$ & $0.861$ & $0.999$ & $0.925$ & $0.041$ \\
ExpSat lower & $1199.8$ & $1.000$ & $0.909$ & $0.952$ & $0.000$ \\
ExpSat central & $1347.6$ & $0.976$ & $0.996$ & $\mathbf{0.986}$ & $0.006$ \\
ExpSat upper & $1490.4$ & $0.885$ & $0.999$ & $0.939$ & $0.033$ \\
RatSat lower & $1211.4$ & $1.000$ & $0.917$ & $0.957$ & $0.000$ \\
RatSat central & $1346.8$ & $0.976$ & $0.996$ & $\mathbf{0.986}$ & $0.006$ \\
RatSat upper & $1485.2$ & $0.888$ & $0.999$ & $0.940$ & $0.032$ \\
Random & $1350.8$ & $0.197$ & $0.202$ & $0.199$ & $0.207$ \\
Isolation forest & $1350.8$ & $0.517$ & $0.528$ & $0.522$ & $0.125$ \\
Local outlier factor & $1350.8$ & $0.338$ & $0.346$ & $0.342$ & $0.171$ \\
One-class SVM & $1350.8$ & $0.359$ & $0.367$ & $0.363$ & $0.166$ \\
\addlinespace
\multicolumn{6}{@{}l}{\textbf{\texttt{pumadyn-32fh}}}\\*
No removal & $0.0$ & -- & -- & -- & -- \\
Oracle (central) & $1371.0$ & $0.963$ & $1.000$ & $0.981$ & $0.010$ \\
DPD lower & $1164.0$ & $1.000$ & $0.881$ & $0.937$ & $0.000$ \\
DPD central & $1371.0$ & $0.960$ & $0.997$ & $0.978$ & $0.010$ \\
DPD upper & $1576.0$ & $0.837$ & $0.999$ & $0.911$ & $0.049$ \\
ExpSat lower & $1220.4$ & $1.000$ & $0.924$ & $0.961$ & $0.000$ \\
ExpSat central & $1369.0$ & $0.961$ & $0.997$ & $\mathbf{0.979}$ & $0.010$ \\
ExpSat upper & $1516.4$ & $0.869$ & $0.998$ & $0.929$ & $0.038$ \\
RatSat lower & $1230.8$ & $1.000$ & $0.932$ & $0.965$ & $0.000$ \\
RatSat central & $1368.8$ & $0.962$ & $0.997$ & $\mathbf{0.979}$ & $0.010$ \\
RatSat upper & $1516.2$ & $0.870$ & $0.998$ & $0.930$ & $0.038$ \\
Random & $1371.0$ & $0.201$ & $0.209$ & $0.205$ & $0.209$ \\
Isolation forest & $1371.0$ & $0.340$ & $0.353$ & $0.346$ & $0.173$ \\
Local outlier factor & $1371.0$ & $0.361$ & $0.375$ & $0.368$ & $0.167$ \\
One-class SVM & $1371.0$ & $0.304$ & $0.316$ & $0.310$ & $0.182$ \\
\addlinespace
\end{longtable}
\endgroup

\subsection{Split sensitivity and predictive diagnostics in cal-housing}
\label{app:cal-predictive-diagnostics}

This analysis concerns \texttt{cal-housing} at $\gamma=0.5$ and $20\%$ injection.
Predictive diagnostics use all $16000$ saved draws for each of the six fits named in Table~\ref{tab:cal-predictive-diagnostics}.

\paragraph{Sensitivity to the unstable split.}
One split exhibited a large H-B predictive error at a test observation with an extreme covariate; we call it the influential split.
We omit it from every method in the following exploratory, post-hoc sensitivity analysis, retaining the same four remaining splits for all methods.
All test observations within each retained split remain in the scores.
The four-split results describe these selected splits and do not replace the all-five-split results or establish stability across the planned splits.
The omission is based on predictive behaviour, not failed MCMC convergence.

\begin{table}[t]
\centering
\caption{Initial posterior fits for \texttt{cal-housing} at $\gamma=0.5$ and injected fraction $0.2$.
All-five-split results are the primary analysis.
The four-split columns omit the influential split from every method as an exploratory post-hoc sensitivity analysis.
Entries are mean (split SD); lower values are better.
All test observations within each retained split are used.
Bold: lowest displayed mean in each column; all ties included.}
\label{tab:cal-split-sensitivity-benchmark}
\scriptsize
\resizebox{\textwidth}{!}{%
\begin{tabular}{lcccc}
\toprule
Method & \multicolumn{2}{c}{Test RMSE} & \multicolumn{2}{c}{Test MC-NLPD} \\
\cmidrule(lr){2-3}\cmidrule(lr){4-5}
 & All five splits & Four splits & All five splits & Four splits \\
\midrule
Standard Bayes & $1.5691\;(0.0262)$ & $1.5667\;(0.0296)$ & $2.134\;(0.009)$ & $2.132\;(0.009)$ \\
H-B (DPD) & $0.9661\;(0.5291)$ & $\mathbf{0.7297}\;(0.0246)$ & $2.163\;(2.025)$ & $\mathbf{1.258}\;(0.061)$ \\
$\gamma$-divergence & $\mathbf{0.9659}\;(0.5286)$ & $\mathbf{0.7297}\;(0.0246)$ & $2.141\;(1.971)$ & $1.260\;(0.063)$ \\
Student-$t$ & $1.0378\;(0.5578)$ & $0.7884\;(0.0216)$ & $\mathbf{1.566}\;(0.158)$ & $1.496\;(0.023)$ \\
H-B (ExpSat) & $0.9661\;(0.5291)$ & $\mathbf{0.7297}\;(0.0246)$ & $2.119\;(1.928)$ & $\mathbf{1.258}\;(0.061)$ \\
H-B (RatSat) & $\mathbf{0.9659}\;(0.5287)$ & $\mathbf{0.7297}\;(0.0246)$ & $2.137\;(1.962)$ & $1.260\;(0.063)$ \\
\bottomrule
\end{tabular}%
}
\end{table}

\begin{table}[t]
\centering
\caption{Downstream Gaussian refits for \texttt{cal-housing} at $\gamma=0.5$ and injected fraction $0.2$.
All-five-split results are the primary analysis.
The four-split columns omit the influential split from every method as an exploratory post-hoc sensitivity analysis.
Entries are mean (split SD); lower values are better.
All test observations within each retained split are used.
Bold: lowest displayed mean in each column, excluding the oracle reference; all ties included.}
\label{tab:cal-split-sensitivity-cleaning}
\scriptsize
\resizebox{\textwidth}{!}{%
\begin{tabular}{lcccc}
\toprule
Method & \multicolumn{2}{c}{Test RMSE} & \multicolumn{2}{c}{Test MC-NLPD} \\
\cmidrule(lr){2-3}\cmidrule(lr){4-5}
 & All five splits & Four splits & All five splits & Four splits \\
\midrule
No removal & $1.5691\;(0.0262)$ & $1.5667\;(0.0296)$ & $2.134\;(0.009)$ & $2.132\;(0.009)$ \\
Oracle (central) & $0.7226\;(0.0091)$ & $0.7254\;(0.0075)$ & $1.093\;(0.012)$ & $1.097\;(0.010)$ \\
H-B DPD lower & $1.0274\;(0.6648)$ & $0.7303\;(0.0314)$ & $2.289\;(2.493)$ & $\mathbf{1.175}\;(0.061)$ \\
H-B DPD central & $1.0017\;(0.6109)$ & $0.7287\;(0.0280)$ & $2.419\;(2.629)$ & $1.243\;(0.070)$ \\
H-B DPD upper & $\mathbf{0.9776}\;(0.5586)$ & $0.7279\;(0.0257)$ & $2.464\;(2.560)$ & $1.319\;(0.075)$ \\
H-B ExpSat lower & $1.0241\;(0.6583)$ & $0.7300\;(0.0309)$ & $2.365\;(2.629)$ & $1.189\;(0.066)$ \\
H-B ExpSat central & $1.0006\;(0.6080)$ & $0.7289\;(0.0283)$ & $2.424\;(2.643)$ & $1.243\;(0.070)$ \\
H-B ExpSat upper & $0.9818\;(0.5684)$ & $0.7278\;(0.0260)$ & $2.485\;(2.653)$ & $1.299\;(0.074)$ \\
H-B RatSat lower & $1.0270\;(0.6642)$ & $0.7302\;(0.0312)$ & $2.357\;(2.604)$ & $1.192\;(0.066)$ \\
H-B RatSat central & $1.0016\;(0.6107)$ & $0.7287\;(0.0281)$ & $2.414\;(2.618)$ & $1.243\;(0.070)$ \\
H-B RatSat upper & $0.9806\;(0.5659)$ & $\mathbf{0.7277}\;(0.0260)$ & $2.450\;(2.575)$ & $1.299\;(0.076)$ \\
Isolation forest & $1.5353\;(0.9100)$ & $1.1289\;(0.0556)$ & $\mathbf{1.951}\;(0.178)$ & $1.873\;(0.044)$ \\
Local outlier factor & $1.7519\;(0.6184)$ & $1.4755\;(0.0188)$ & $2.116\;(0.049)$ & $2.094\;(0.004)$ \\
One-class SVM & $1.4842\;(0.4571)$ & $1.2805\;(0.0472)$ & $1.977\;(0.039)$ & $1.961\;(0.015)$ \\
Random removal & $1.5712\;(0.0310)$ & $1.5649\;(0.0318)$ & $2.134\;(0.012)$ & $2.131\;(0.011)$ \\
\bottomrule
\end{tabular}%
}
\end{table}

For the four retained splits, the central DPD H-B refit has mean RMSE $0.7287$ and MC-NLPD $1.243$, compared with $0.7254$ and $1.097$ for the oracle and $1.5667$ and $2.132$ for no removal.
Its RMSE is close to the oracle, while a predictive-density gap remains.
The mean coverage of its nominal $95\%$ predictive intervals is $88.9\%$, versus $94.6\%$ for the oracle, so undercoverage persists after omitting the unstable split.
The four-split comparison does not resolve the substantial variation across the five planned splits.

\paragraph{Predictive quantities.}
In the original response coordinates, let $\mu_s(x)$ and $\sigma_s$ be the conditional mean and residual standard deviation at saved draw $s$, and write $\bar\mu(x)=S^{-1}\sum_{s=1}^S\mu_s(x)$.
With $S=16000$, the reported finite-draw predictive density and variance are
\[
 \widehat p_S(y\mid x)=\frac1S\sum_{s=1}^S
 \mathcal N\{y;\mu_s(x),\sigma_s^2\},\qquad
 V_S(x)=\underbrace{\frac1S\sum_{s=1}^S\sigma_s^2}_{A}
 +\underbrace{\frac1S\sum_{s=1}^S\{\mu_s(x)-\bar\mu(x)\}^2}_{B(x)}.
\]
Thus coefficient uncertainty is included through $B(x)$.
Central $80\%$ and $95\%$ predictive intervals for the response are obtained by inverting the finite-mixture cumulative distribution function (CDF), with absolute CDF error at most $10^{-8}$.
Coverage of these response intervals is descriptive for the five fixed, overlapping splits.

\begin{table}[t]
\centering
\caption{Predictive diagnostics for \texttt{cal-housing} at $\gamma=0.5$ and injected fraction $0.2$.
All five splits and all $4128$ test observations per split are retained.
MC-NLPD is mean (split SD); other entries are split averages.
Coverage is in percent. Widths are in the original response units; variance components are in squared response units.
$A$ and $B$ are defined in the text; $B$ also averages over test covariates.
The extreme-tail MC-NLPD qualification below applies to these reported finite-draw scores.}
\label{tab:cal-predictive-diagnostics}
\scriptsize
\resizebox{\textwidth}{!}{%
\begin{tabular}{lcccccc}
\toprule
Method & MC-NLPD & $80\%$ cov. & $95\%$ cov. & $95\%$ width & $A$ & $B$ \\
\midrule
No removal & $2.134\;(0.009)$ & $99.9$ & $100.0$ & $11.435$ & $8.506$ & $0.004036$ \\
Oracle (central) & $1.093\;(0.012)$ & $86.3$ & $94.6$ & $2.844$ & $0.526$ & $0.000339$ \\
H-B lower budget & $2.289\;(2.493)$ & $79.5$ & $91.2$ & $2.202$ & $0.315$ & $0.000794$ \\
H-B central budget & $2.419\;(2.629)$ & $76.7$ & $89.1$ & $2.022$ & $0.266$ & $0.000731$ \\
H-B upper budget & $2.464\;(2.560)$ & $74.1$ & $87.5$ & $1.888$ & $0.232$ & $0.000665$ \\
Initial H-B & $2.163\;(2.025)$ & $76.2$ & $88.5$ & $1.981$ & $0.255$ & $0.001602$ \\
\bottomrule
\end{tabular}%
}
\end{table}

\paragraph{Undercoverage and variance.}
Mean $95\%$ coverage is $89.1\%$ for the central refit and $94.6\%$ for the oracle at identical retained counts.
Central coverage ranges from $87.9\%$ to $90.1\%$ across splits.
The mean residual variance contribution is $0.266$ versus $0.526$.
The central coefficient contribution is larger ($0.000731$ versus $0.000339$), but its smaller residual contribution dominates.
Higher H-B budgets decrease coverage and interval width.
No removal gives $100\%$ coverage with much wider intervals.
Thus increased overall predictive variance does not explain the deterioration.

\paragraph{Concentration of the reported MC-NLPD.}
In the influential split, central and oracle refits have MAE $0.511$ and $0.523$, RMSE $2.094$ and $0.711$, and MC-NLPD $7.121$ and $1.077$, respectively.
A single non-injected test observation accounts for $99.3\%$ of their signed MC-NLPD difference: its contribution to the split mean difference is $5.999$, out of $6.044$.
No observation is removed from the reported metrics.
One standardised covariate at this observation has magnitude $45.7$ and lies outside the corresponding retained-training range for the central refit.
The central prediction is $-123.743$ for an observed response of $3.5$; the initial H-B fit also predicts $-111.116$, before any cleaning.
Local coefficient uncertainty is substantial: $B(x)=9.900$ for the central refit, even though its average over the test set is small.
The corresponding total predictive standard deviation is about $3.189$, far smaller than the absolute prediction error $127.243$.
The error already precedes downstream removal.
The same saved covariate--response pair is absent from the other four test sets.

The retained training designs also differ despite identical counts of $12521$ for central H-B and the oracle.
At this observation, the design leverage $x^\top(X_{\mathrm{keep}}^\top X_{\mathrm{keep}})^{-1}x$ is $37.353$ for central H-B and $0.587$ for the oracle.
This loss of design support is consistent with an additional extrapolation problem after removal.
The large error already present in initial H-B also points to sensitivity of the robust Gaussian linear fit at an extreme covariate, rather than downstream removal alone.
These observations explain where the instability appears, but do not identify a unique causal mechanism.

\begin{table}[t]
\centering
\caption{The test observation with an extreme covariate in the influential split, with response $y=3.5$.
$\bar\mu$ is its predictive mean, and $A$ and $B(x)$ are its variance components.
$N_w$ measures concentration of predictive-density contributions over the saved draws; it is distinct from MCMC effective sample size.
All entries describe this same observation, including for the oracle.}
\label{tab:cal-influential-observation}
\small
\resizebox{\textwidth}{!}{%
\begin{tabular}{lccccc}
\toprule
Method & $\bar\mu$ & $A$ & $B(x)$ & Pointwise MC-NLPD & $N_w$ \\
\midrule
No removal & $3.677$ & $8.627$ & $1.131$ & $2.06$ & $15888.3$ \\
Oracle (central) & $-0.644$ & $0.529$ & $0.317$ & $11.03$ & $113.9$ \\
H-B lower budget & $-132.096$ & $0.315$ & $10.565$ & $23462.86$ & $1.0$ \\
H-B central budget & $-123.743$ & $0.268$ & $9.900$ & $24774.95$ & $1.0$ \\
H-B upper budget & $-115.621$ & $0.232$ & $9.001$ & $24146.78$ & $1.0$ \\
Initial H-B & $-111.116$ & $0.254$ & $20.095$ & $19096.78$ & $1.0$ \\
\bottomrule
\end{tabular}%
}
\end{table}

\paragraph{Accuracy of extreme-tail density integration.}
Reproducing the finite-draw score does not establish accurate integration of a far-tail predictive density.
To diagnose concentration, set $\ell_s=\mathcal N\{y;\mu_s(x),\sigma_s^2\}$,
$w_s=\ell_s/\sum_r\ell_r$, and $N_w=1/\sum_s w_s^2$.
At the influential observation, $N_w$ is approximately one for the initial H-B fit and all three downstream H-B refits.
For the central refit, evaluating MC-NLPD separately with each chain gives values ranging from $24773.57$ to $25444.25$.
This concerns density integration, separately from parameter-chain convergence.
We report the original finite-draw scores, including the extreme observation, and qualify their magnitudes as Monte Carlo estimates whose tail accuracy is not established by parameter-chain diagnostics alone.

\paragraph{Implications.}
Likelihood-based removal may further narrow the residual distribution by excluding non-injected observations that the working model fits poorly, but does not explain the initial H-B prediction error.
These comparisons do not identify a unique causal mechanism.
Refitting that accounts for selection and more flexible conditional models, including treatment of extreme covariates, warrant study.
Their ability to improve calibration while preserving robustness remains to be evaluated.

\end{document}